\documentclass[12pt]{article}

\usepackage[margin=2in,tmargin=1.5in, bmargin=1.5in]{geometry}
\usepackage{float}
\usepackage[utf8]{inputenc}
\usepackage[UKenglish]{babel}
\usepackage[T1]{fontenc}
\usepackage{csquotes}
\usepackage{algorithm}
\usepackage{algpseudocode}
    \algrenewcommand\algorithmicrequire{\textbf{Input:}}
    \algrenewcommand\algorithmicensure{\textbf{Output:}}
\usepackage[dvipsnames]{xcolor}

\usepackage{amsfonts,amsthm,amsmath,amssymb,mathtools}
\allowdisplaybreaks
\usepackage{dsfont,newtxtext,newtxmath,microtype}
\usepackage{makecell}

\usepackage{authblk}
\usepackage{booktabs,multirow}
\usepackage{braket}
\usepackage{cancel}
\usepackage[font=small]{caption}
\usepackage{comment}
\usepackage{diagbox}
\usepackage{enumerate}
\usepackage{float}
\usepackage{fullpage}
\usepackage[section]{placeins}
\usepackage{subcaption}
\usepackage{slashed}
\usepackage{todonotes}
\usepackage{bbm} 
\usepackage{yfonts}
\usepackage{bbding,enumitem}
\usepackage{setspace}
\usepackage{yfonts}

\pgfdeclarelayer{background}
\pgfsetlayers{background,main}

\definecolor{navyblue}{rgb}{0.0, 0.0, 0.5}
\definecolor{LightPink}{rgb}{0.858, 0.188, 0.478}

\usepackage[colorlinks=true, linkcolor=Periwinkle, citecolor=MidnightBlue, urlcolor=black, linktocpage=true]{hyperref}

\newcommand{\ketbra}[1]{\ket{#1}\!\!\bra{#1}}
\DeclarePairedDelimiter\abs{\lvert}{\rvert}
\DeclarePairedDelimiter\norm{\lVert}{\rVert}

\DeclareMathOperator{\tr}{Tr}
\DeclareMathOperator{\polylog}{polylog}
\DeclareMathOperator{\E}{\mathbb{E}}
\newcommand{\vertiii}[1]{{\left\vert\kern-0.25ex\left\vert\kern-0.25ex\left\vert #1 
    \right\vert\kern-0.25ex\right\vert\kern-0.25ex\right\vert}}

\newcommand{\acomm}{\alpha_\mathrm{comm}}

\newcommand{\lcomm}{{\lambda}_\mathrm{comm}}
\newcommand{\lcommklocal}{\widetilde{\lambda}_{k\text{-}\mathrm{local}}}
\newcommand{\lcommi}{{\lambda}_{\mathrm{comm},i}}
\newcommand{\Heff}{H_\mathrm{eff}}
\newcommand{\theff}{\widetilde{H}_\mathrm{eff}}

\newcommand{\cP}{\PC}

\newcommand{\cC}{\mathcal{C}}

\newcommand{\cO}{O}
\newcommand{\sym}{\sigma}

\newcommand{\PC}{\mathcal{P}}

\usepackage[capitalise]{cleveref}
\Crefname{lemma}{Lemma}{Lemmas}
\Crefname{proposition}{Proposition}{Propositions}
\Crefname{definition}{Definition}{Definitions}
\Crefname{theorem}{Theorem}{Theorems}
\Crefname{conjecture}{Conjecture}{Conjectures}
\Crefname{corollary}{Corollary}{Corollaries}
\Crefname{example}{Example}{Examples}
\Crefname{section}{Section}{Sections}
\Crefname{appendix}{Appendix}{Appendices}
\Crefname{figure}{Fig.}{Figs.}
\Crefname{equation}{Eq.}{Eqs.}
\Crefname{table}{Table}{Tables}
\Crefname{item}{Property}{Properties}
\Crefname{remark}{Remark}{Remarks}
\Crefname{fact}{Fact}{Facts}

\newtheorem{theorem}{Theorem}
\newtheorem{definition}[theorem]{Definition}

\newtheorem{lemma}[theorem]{Lemma}

\usepackage{ltablex}

\newcommand\prob\textsc
\makeatletter
\newcommand{\probleminput}[1]{\gdef\@probleminput{#1}}
\newcommand{\problemquestion}[1]{\gdef\@problemquestion{#1}}
\newcommand{\problempromise}[1]{\gdef\@problempromise{#1}}
\DeclareDocumentEnvironment{problem}{}{
\probleminput{}\problempromise{}\problemquestion{}
\leavevmode
}{
  \par\addvspace{0\baselineskip}
  \noindent
  \begin{itemize}
      \item[\textbf{Input:}] \@probleminput
\ifx\@problempromise\empty%
\else%
      \item[\textbf{Promise:}] \@problempromise
\fi%
      \item[\textbf{Output:}] \@problemquestion
  \end{itemize}
}
\makeatother

\usepackage[backend=biber,style=alphabetic,doi=false,isbn=false,url=false,maxbibnames=7,maxcitenames=2]{biblatex}
\bibliography{sambib}
\newbibmacro{string+doi}[1]{\iffieldundef{doi}{#1}{\href{https://dx.doi.org/\thefield{doi}}{#1}}}
\DeclareFieldFormat{title}{\usebibmacro{string+doi}{\mkbibemph{#1}}}
\DeclareFieldFormat[article]{title}{\usebibmacro{string+doi}{\mkbibquote{#1}}}
\DeclareFieldFormat[incollection]{title}{\usebibmacro{string+doi}{\mkbibquote{#1}}}

\usepackage{tcolorbox}

\title{Resource-efficient quantum eigenvalue transform with commutator scaling}

\author[1,2]{\href{}{Arul Rhik Mazumder}\thanks{arul.rhik@gmail.com}}
\author[3]{\href{}{James D. Watson}\thanks{watsonjames@google.com}}
\author[1]{\href{}{Samson Wang}\thanks{samsonwang@outlook.com}}

\affil[1]{Caltech}
\affil[2]{Carnegie Mellon University}
\affil[3]{Google Quantum AI}

\date{}

\begin{document}

{\begingroup
\hypersetup{urlcolor=navyblue}
\maketitle
\endgroup}

\begin{abstract}
  We develop quantum algorithms for estimating properties of general matrix functions of Hermitian matrices, with applications to phase estimation, Green’s function evaluation, and estimating measurement distributions of time-evolved states. The resulting methods exhibit commutator scaling in matrix parameters similar to that usually found for product formulae, lower circuit depth in other parameters, and require only a single ancillary qubit. Our central primitive consists of classically postprocessing randomly chosen product formulae circuits, which mathematically corresponds to an approximation of a Richardson extrapolation. 
  Within our framework, we introduce a protocol for approximating the measurement distributions of quantum states, extending beyond standard observable estimation. We also provide tightened gate complexity bounds for practically relevant systems, including those with $k$-local interactions, long-tailed matrix ensembles, and conserved quantities. Finally, numerical experiments confirm that our method can achieve significantly shallower circuit depths than standard product formulae in certain parameter regimes, and highlight the potential of their heuristic application.
\end{abstract}

\vspace{1em}

\begin{center}
\small
\setstretch{0.85}{
\tableofcontents
}
\end{center}

\section{Introduction}

\yinipar{T}{\sc he} task of processing large matrices is pervasive in scientific computing. Within the realm of quantum algorithms, it can be unifying and useful to frame such problems as \textit{matrix function tasks} -- that is, some task where one extracts some property of some function of a matrix of interest. For instance, ground state problems in chemistry and many-body physics, time evolution, linear systems, and quantum walks can be framed as matrix function problems. In this picture, the quantum singular value transform (QSVT) allows efficient algorithms for smooth functions by instantiating matrix polynomials, and often obtains optimal complexity for such problems \cite{gilyen2018QSingValTransf, martyn2021GrandUnificationQAlgs}. Whilst the QSVT gives a powerful and general tool, it may not be the best suited to early quantum hardware; the QSVT requires many multi-qubit controlled gates and non-trivial ancillary space overhead. 

In this work we focus on \textit{early fault-tolerant} algorithms. Here, we pursue algorithms with analytical runtime guarantees, in a setting where small quantum computers are constrained by limited qubit count and circuit depth. We imagine such quantum computers could have modest use of error correction, or no error correction at all. Whilst overall runtime complexity may not be optimal, our goal is to fit algorithms to these constraints on \textit{quantum} resources by leveraging \textit{classical} power -- in performing preprocessing, in postprocessing data from many different quantum circuits, and in using the power of randomness. We follow a line of work with similar motivations which have proposed algorithms for spectral analysis \cite{somma2019QEE}, ground state problems \cite{lin2022HeisenbergLimited, wan2021RandPhaseEst, wang2023QAlgGSEE, wang2023fasterGSEE, dong2022groundStatePrep, kiss2025earlyFTinPractice}, quantum dynamics \cite{campbell2019randomCompiler, campbell2021earlyFTHubbard} Gibbs state preparation \cite{ding2025end2end}, general matrix functions \cite{wang2023qubitefflinalg, chakraborty2025quantum}, as well as work that has considered practical implementation of the aforementioned approaches \cite{blunt2023statistical, kiss2025earlyFTinPractice}.

An additional aim in our work not shared in all of the above approaches is to specify end-to-end algorithms. For the purpose of this work, by this we refer to algorithms that are specified with (1) efficient classical input and efficiently-obtainable classical output, and (2) all quantum oracles explicitly instantiated. As we aim to propose candidate algorithms for potential concrete application, we argue these two requirements should be satisfied at minimum in order to probe feasibility. 

The central message in our work is that product formulae augmented with classical pre- and post-processing could be a promising route towards early fault-tolerant algorithms. Product formulae (including Trotter-Suzuki formulae) are a method for approximating time evolution, which is itself an algorithm, but is also a pervasively used primitive in quantum algorithms more broadly. They are implementable for Hermitian matrices (such as physical Hamiltonians) whenever the matrix can be decomposed into a linear combination of constituent matrices each of whom can be time-evolved efficiently -- for instance, this includes matrices that can be decomposed into the Pauli basis, or sparse matrices with efficiently computable non-zero row entries \cite{aharonov2003adiabQStateGeneration}. Product formulae require no ancillary overhead to simulate time evolution and their simple structure allow widespread use already in quantum computing experiments today, including recent landmark demonstrations \cite{kim2023evidence, google2025observation, haghshenas2025digitalMagnetism}. 
A key advantage of product formulae is that their gate complexity depends on the commutator structure of the constituent matrices \cite{childs2021TheoryTrotter} -- this is in contrast to all other known methods of simulating time evolution, where complexity depends on often much larger (generic) matrix parameters. Further, this so called \textit{commutator scaling} can be refined when there is a prior on the input state; when it has a well-defined fermion number \cite{su2021NearlyTightTrottInerElect, mcardle2022ExploitingFermionNumber, low2023complexityTrotter} or belongs to a low-energy subspace \cite{Sahinoglu2021Hamiltonian, hejazi2024LowEnergyProduct, mizuta2025trotterizationLowEnergy}.

\subsection{Contributions}

We develop and refine analyses for two simple but general subroutines. The first 
enables the estimation of the following quantities:
\begin{equation*}\label{eq:task-1}
(1):\qquad \text{estimate (i)}:
\text{Tr}[\rho U f(H)] \;\ \text{and\; (ii)}: \text{Tr}[f(H) \rho f(H)^{\dagger} O] \; \text{to additive precision $\varepsilon$}\,,
\end{equation*}
where $\rho$ is a quantum state, $U$ is a unitary, $O$ is an observable, all of which we assume to be efficiently implementable. We denote $f(H)$ as the eigenvalue transform of the function $f$ on a Hermitian matrix $H$ (see Definition \ref{def:eigval-transform}). Our second subroutine allows us to approximate measurement statistics of the quantum state proportional to $f(H) \ket{\psi}$. Specifically, we ask to:
\begin{equation}\label{eq:task-2}
    (2):\qquad \text{return a vector $\vec{v}$ such that}\; \| \vec{v} - \vec{p} \|_2\leq \varepsilon \; \text{for} \; p_i := |\bra{i} U f(H) \ket{\psi} |^2.
\end{equation}
We stress that here we return a full vector whose size scales with the dimension of $H$.\footnote{While it may seem unusual that an efficient algorithm can be given here, we note that efficiency is possible due to the fact that we are asking for an approximation in $\ell_2$-norm, and for a large $\varepsilon$ will output a sparse approximation. This is the same statistical approximation guarantee available in estimating the probability distribution corresponding to any quantum state.} In both cases we are satisfied with an algorithm that succeeds with high success probability.

While it is not uncommon to view the output of a quantum algorithm as a quantum state, most truly end-to-end use cases ultimately require a \textit{classical} output~\cite{dalzell2023EndToEnd}. Thus, we view these two tasks as natural and broadly applicable frameworks that accommodate a lot of quantum algorithms targeting practical applications -- in both estimating scalar quantities or in sampling tasks. To make this concrete, consider a simple example: if \( f \) is the inverse function and \( \rho = \ket{\vec{y}}\!\bra{\vec{y}} \) is a pure state that encodes a data vector $\vec{y}$, this setup corresponds to the well-known quantum linear systems problem~\cite{harrow2009QLinSysSolver}. If the observable \( O \) in Task \hyperref[eq:task-1]{1(ii)} is taken to be a local measurement, then the resulting quantity gives a marginal of the solution to the linear system. If $U$ in Task \hyperref[eq:task-1]{1(i)} is chosen to be a state preparation unitary corresponding to some probe vector (and unpreparation of $\ket{\vec{y}}$), then the output is the overlap of the solution vector $H^{-1}\vec{y}$ with that probe vector. More generally, Tasks \hyperref[eq:task-1]{1} and \hyperref[eq:task-2]{2} can be used to extract spectral information about a target matrix \( H \) (see discussion surrounding Eq.~\eqref{eq:CDF}).


Our first contribution is to give algorithms to instantiate Tasks \hyperref[eq:task-1]{1} and \hyperref[eq:task-2]{2} whenever given a Fourier series approximation to the function $f$, and ability to Trotterize $H$. The central idea is to consider an extrapolation of circuits consisting of product formulae of different step sizes, and to extract estimates of this extrapolation in a sample-efficient manner via randomization. General matrix functions can be compiled as part of the randomization loop by sampling from the Fourier series. From hereon we will refer to this central primitive as \textit{Randomized Extrapolation Trotterization}, and we provide a schematic thereof in Figure \ref{fig:flowchart}. Using this technique, the gate complexity can be substantially improved over algorithms for matrix functions using regular product formulae, with a tradeoff that one needs to collect (what we will argue is) mildly more data from a quantum computer. Importantly, our algorithms exhibit commutator scaling in their dependence on Hamiltonian parameters, similar to that found in product formulae. Further, they only use a single ancillary qubit.

\begin{figure}[t]
    \centering
\includegraphics[width=0.9\linewidth]{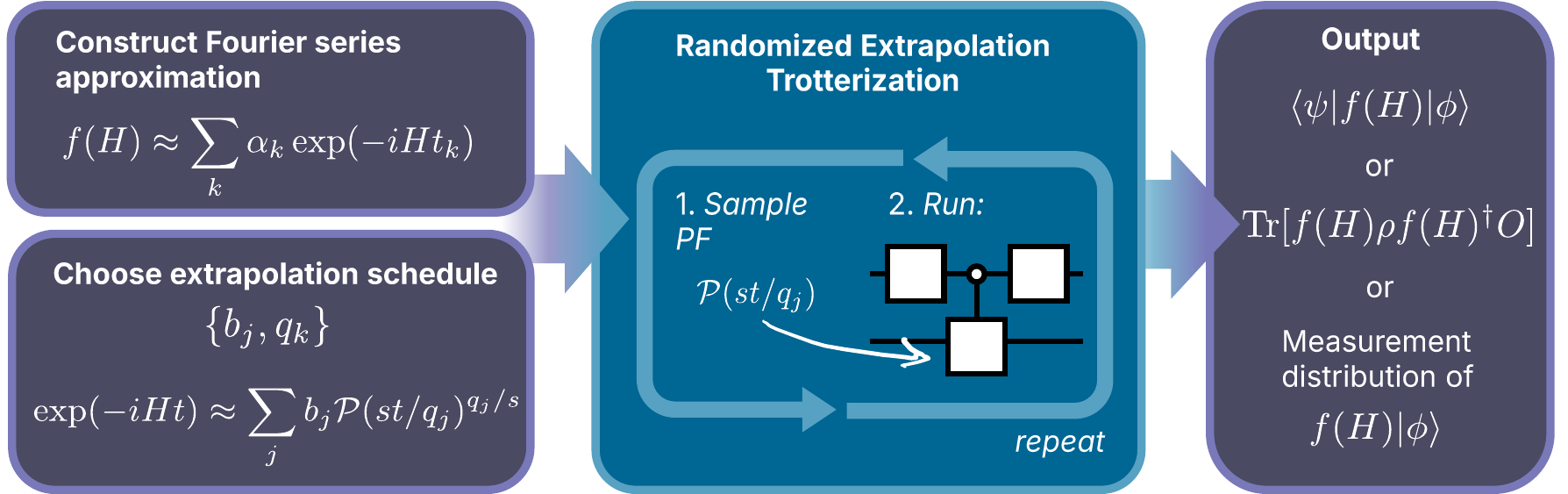}
    \caption{\textbf{Schematic of Randomized Extrapolation Trotterization.} Given any user-defined Fourier series approximation of $f$ on the domain of the eigenspectrum of $H$, given any extrapolation schedule, and given any product formula (PF), Randomized Extrapolation Trotterization enables algorithms that return an estimate of various properties of $f(H)$. Our algorithms randomly sample the product formula at different Trotter step sizes according to a carefully constructed distribution, implement each of them in a Hadamard test, and postprocess the measurement statistics classically.}
    \label{fig:flowchart}
\end{figure}

\begin{table}[h!]
\resizebox{\columnwidth}{!}{
\renewcommand{\arraystretch}{2}
    \centering
    \begin{tabular}{r|c|c|c|c}
        method & \makecell{ancillary \\qubits } & \makecell{max depth per sample \\ $\widetilde{\cO}(\cdot)$ } & \makecell{sample overhead \\${\cO}(\cdot)$} & \makecell{commutator \\ scaling?} \\\hline
        \makecell[r]{Qubitization \\\cite{low2016HamSimQSignProc}}& $\lceil \log(\Gamma) + 3\rceil$ & $\displaystyle \Gamma\Big(\Lambda T+\frac{\log (1 / \varepsilon)}{\log \log (1 / \varepsilon)} \Big)$ & $\displaystyle 1/\varepsilon^2$ & \\
        \makecell[r]{Multiproduct \\ formula \cite{aftab2024multi}}  & $\widetilde{O}(\log\log(\lcomm T/\varepsilon))$ & $\displaystyle \Gamma \lcomm T \log^2(\lcomm T/\varepsilon) $ & $\displaystyle 1/\varepsilon^2$ & \checkmark \\
        \makecell[r]{Product formula \\ \cite{childs2021TheoryTrotter}}  & 1 & $\displaystyle {\Gamma (\alpha^{(p+1)}_{\text{comm}})^{1/p}\, T^{1+1/p}}{(1/\varepsilon)^{1/p}}$ & $\displaystyle 1/\varepsilon^2$ & \checkmark \\
        \makecell[r]{Random compiler \\\cite{wan2021RandPhaseEst}}  & 1 & $\displaystyle \Lambda^2 T^2$ & $\displaystyle 1/\varepsilon^2$ & \\
        \makecell[r]{\textbf{Rand.~Extrap.} \\
        \textbf{Trotterization } } & 1 & $\displaystyle \Gamma \left(\lambda_{\text{comm}} T\right)^{1+1 / p} \log (1 / \varepsilon)$ & $\displaystyle\frac{(\log \log (1/\varepsilon))^2}{\varepsilon^2}$ & \checkmark 
    \end{tabular}
    }
    \caption{\textbf{Compilation of a time signal $\tr[\rho e^{-iHT}]$ to additive precision $\varepsilon$.} In order to allow comparison with the random compiler of \cite{wan2021RandPhaseEst} we have assumed that the matrix $H$ is given as a linear combination of $\Gamma$ Pauli matrices with $\ell_1$-norm of coefficients $\Lambda$, although all other methods in this table take a more general access model. Where product formulae are instantiated, we assume a $p$-th order formula with constant $p$ (see Definition \ref{def:staged_pf}). The quantity $\alpha^{(p+1)}_{\text{comm}} \leq \frac{1}{2}(2 \Lambda)^{p+1}$ is a sum of norms of nested commutators (see Definition \ref{def:acomm}), and $\lcomm \leq 4\Lambda$ is constructed from such sums (see Definition \ref{def:lcomm}), both of which can often be much smaller than $\Lambda$ for models of physical systems. We present a generic scaling for Randomized Extrapolation Trotterization in the final row of the table, with refined expressions to the gate depth presented in Table \ref{tab:time-signal-2}. All expressions generalize naturally to our algorithms for Task \hyperref[eq:task-1]{1} and \hyperref[eq:task-2]{2}. \label{tab:time-signal}}
\end{table}

The fine details of the complexity depend on the properties of the function of interest $f$ and are presented in Theorems~\ref{thm:overlap}, \ref{thm:overlap-with-observable} and \ref{thm:distribution-recovery}. Generally, both algorithms require $\widetilde{\cO}(1/\varepsilon^2)$ samples from quantum circuits to estimate both quantities to additive precision $\varepsilon$, with additional factors again depending on the function. In the following subsection and in Table \ref{tab:time-signal} we elucidate the complexity of a core subroutine which also serves as a warm-up example.

We instantiate our algorithm framework for a number of applications (summarized in Section \ref{sec:applications}): ground state energy estimation, evaluation of Green's functions, and distribution learning for a time-evolved state. For ground state energy estimation, this leads to an algorithm that (1) exhibits commutator scaling akin to that found for product formulae as applied to Hamiltonian time evolution; (2) requires circuit depths which depend only polylogarithmically on the ansatz state overlap; and (3) uses only one ancillary qubit. To the authors' knowledge so far these three properties are only simultaneously achievable via product formulae extrapolation.

Our second contribution is to demonstrate refined gate complexities for our methods in a number of settings that exploit commutator scaling. These were all previously unknown or not explicitly instantiated for extrapolated product formulae circuits -- our results are applicable to both Randomized Extrapolation Trotterization and also prior art in extrapolated product formulae \cite{watson2024exponentiallyReducedTrotter, chakraborty2025quantum}. We cover:

\begin{itemize}
    \item Scenarios where commutator scaling for a product formula is only understood for a fixed order (e.g.~$p=2$). Previous studies could only exploit commutator scaling when scaling is understood at every order.
    \item Interpolation of our scaling with that of randomized product formulae, exploiting the ideas of \cite{gunther2025phase}, which can be particularly advantageous for matrix ensembles with long tails.
    \item Commutator scaling for $k$-local systems, subsuming with it models of power-law interactions. Here, previous studies on extrapolation (\cite{watson2024exponentiallyReducedTrotter, chakraborty2025quantum}) could not exploit this physical assumption due to a divergent commutator factor. Results here are possible due to recent insights of Mizuta \cite{mizuta2026commutator}.
    \item 
    States that have a well-defined quantity corresponding to a matrix (Hamiltonian) symmetry, with total Fermion number as an example.
\end{itemize}

Finally, we present numerical results which demonstrate the advantageous performance of Randomized Extrapolation Trotterization in practice. We show that beyond asymptotic scaling, heuristic optimization of extrapolation schemes allows for strong improved performance over product formulae, whilst still retaining performance guarantees. We also show that outside of rigorous bounds, Randomized Extrapolation Trotterization shows promise in greatly reducing algorithmic error in practice. 

\subsection{Related work}

Recently, it has been shown that simulations of time-evolved observables using product formulae can be classically extrapolated, leading to exponentially improved gate complexities in the error parameter \cite{watson2024exponentiallyReducedTrotter, chakraborty2025quantum}. Such extrapolated product formulae were shown to have gate complexities governed by the quantity $\lcomm$ (see Definition \ref{def:lcomm}) which was found to enjoy commutator scaling in certain settings -- this is the same quantity that appears in the analysis of our methods. 
We note that extrapolated product formulae can be distinguished from (and yet are inspired by) similar mathematical results in multiproduct formulae \cite{aftab2024multi}, which require coherent implementation with a "linear combination of unitaries" subroutine that uses ancillary qubits and amplitude amplification steps. Multiproduct formulae can also be compiled in a randomized fashion
\cite{faehrmann2022randomizing}, though we note that for our applications, this leads to an unwieldy sample overhead that is exponential in the number of Trotter steps.

In \cite{watson2024exponentiallyReducedTrotter} the authors introduce a Trotter extrapolation algorithm for the simulation of time-evolved observables. This is a special case of Task \hyperref[eq:task-1]{1(ii)} where the function of interest is the complex exponential (real time evolution) function, and our algorithm reduces to this prior method.

We remark that Task \hyperref[eq:task-1]{1} is also instantiable with the same gate complexity as our approach using recent work \cite{chakraborty2025quantum}, which introduces a more general primitive in the setting of observable estimation for general quantum circuits. Nevertheless, to tell a self-contained story we still demonstrate how to instantiate Task \hyperref[eq:task-1]{1} using our randomized compiler, and present the matrix function framework in full whilst providing a simpler streamlined proof specified for this task. Further, our method of randomized extrapolation allows for a refined sample complexity, and our new analyses of gate complexity allow for novel refined scaling in Hamiltonian parameters across a number of new settings. Our Task \hyperref[eq:task-2]{2} has not traditionally been considered before in the literature of classical-quantum or early fault-tolerant algorithms, and it opens up a new way to probe output distributions of quantum states despite using an algorithm where matrix processing nor state preparation is fully instantiated coherently.

In the ensuing sections, we elucidate our results by instantiating different examples, and along the way compare the performance of our methods with prior art in the literature.

\subsection{Warm-up: compiling a time signal}

In this section we discuss a special case of Task \hyperref[eq:task-1]{1}: estimating the time signal $\tr[\rho e^{-iHT}]$ for a given time $T$ to precision $\varepsilon$. This serves both as a warm-up and as essentially a core primitive throughout our work (although to obtain refined complexities general functions are compiled with an outer randomization loop, rather than explicitly instantiating this as a subroutine). Further, this gives us a chance to compare the complexity of Randomized Extrapolation Trotterization with other known primitives to instantiate time evolution in a clean transparent way.

Given a method for Hamiltonian simulation, such as a product formula circuit, the time signal can be estimated by taking data from a Hadamard test circuit which uses one ancillary qubit (on top of any ancillary qubits used for the simulation method). The core idea we use, Randomized Extrapolation Trotterization, is simple algorithmically: collect data from randomly chosen product formulae circuits of different Trotter step sizes, from a carefully chosen distribution. This leads to strong suppression of algorithmic errors if step sizes are sufficiently small that they lie within a radius of convergence, to be defined. Mathematically, this procedure enables a randomized compilation of a Richardson extrapolation -- where the extrapolation parameter is the Trotter step size.


We suggest that this approach could be well-suited to early quantum hardware. Randomized Extrapolation Trotterization offloads complexity from coherent quantum circuit depth to increased circuit samples -- yielding shorter circuits by collecting more measurement data from circuits. Further, as we demonstrate below, the increase in sample overhead is very mild by virtue of our randomized protocol.

In order to present results, we first introduce some basic notation.

\begin{definition}\label{def:matrix}
We consider Hermitian matrices \( H \in \mathbb{C}^{2^n \times 2^n} \), expressed as a sum of \( \Gamma \) terms:
\begin{equation}\label{eq:matrix}
    H = \sum_{\gamma=1}^{\Gamma} H_\gamma\,,
\end{equation}
where it is assumed that the unitaries $\exp (-i H_{\gamma}t)$ are instantiable for all $t$ and $\gamma$ in $O(1)$ gate depth.
We define \( \Lambda := \sum_{\gamma=1}^{\Gamma} \|H_\gamma\| \) as the sum of the strength of individual terms, where $\|\cdot\|$ denotes the operator norm.
\end{definition}
Non-Hermitian matrices may be considered in our framework by means of a Hermitian dilation using a single ancillary qubit, and we hereon propagate the assumption of Hermiticity (see Section \ref{sec:preliminaries} for more discussion).

Generically, the gate complexity of algorithms for time evolution depends on $\Lambda$, which can be large for many models of physical systems. However, the complexity of $p$-th order product formulae can be shown to scale with a refined quantity $(\acomm^{(p+1)})^{1/p}$, which we define below.

\begin{definition}[Commutator factor]\label{def:acomm}
Given the decomposition \( H = \sum_{\gamma=1}^{\Gamma} H_\gamma \) from Definition \ref{def:matrix}, we define the order-\(j\) commutator factor
\begin{equation}
    \alpha_{\mathrm{comm}}^{(j)} 
:= 
\sum_{\gamma_1, \dots, \gamma_j = 1}^{\Gamma} \left\|
\bigl[ H_{\gamma_1}, \bigl[ H_{\gamma_2}, \dots, [ H_{\gamma_{j-1}}, H_{\gamma_j} ] \dots \bigr] \bigr] \right\|,
\end{equation}
that is, the sum over norms of each \(j\)-fold nested commutator formed from constituent matrix terms \(H_\gamma\).
\end{definition}

Whilst one always has the generic bound $(\acomm^{(p+1)})^{1/p} \leq 2\Lambda^{1+1/p}$, it is known that in many settings $(\acomm^{(p+1)})^{1/p} \ll \Lambda$ \cite{childs2021TheoryTrotter}, which leads to significant advantage of product formula over other methods in dependence on Hamiltonian parameters.

We present the resources used by Randomized Extrapolation Trotterization to estimate a time signal to additive precision $\varepsilon$ and constant success probability in Table \ref{tab:time-signal}, and contextualize the result against other approaches for Hamiltonian simulation. We remark that for now for this simple subroutine, the same expression for circuit depth is also obtainable with the results from \cite{chakraborty2025quantum}.  We also remark that for all algorithms the overall scaling with $O(1/\varepsilon)$ can be enhanced if coherent preparation access to $\rho$ is available, by means of amplitude estimation \cite{brassard2002AmpAndEst, montanaro2015QMonteCarlo}. We do not include discussions of such techniques in our table nor in the rest of our manuscript for simplicity and focus on optimizing quantum resources per coherent circuit run.

While qubitization achieves the best asymptotic scaling in gate depth in terms of precision and simulation time, as above we argue it may not be the strongest candidate for early fault-tolerant (EFT) devices. This is due to the confluence of a few properties: qubitization has a larger space overhead,\footnote{Here we account for space overhead only in logical qubits. If one wishes to implement error correction, resource counts may become more subtle.} it requires complex multi-qubit-controlled gates, and to the authors' knowledge there are not currently ways to improve its gate complexity by exploiting physical properties of the Hamiltonian such as locality in the way that product formulae do. Thus, we include it in Table \ref{tab:time-signal} as a standard-bearer of a fully-coherent protocol that attains optimal complexity.
Product formulae, conversely, can be more hardware-friendly (especially on non-error-corrected hardware) and can take advantage of commutator scaling in dependence on matrix parameters. However, they suffer from exponentially worse gate complexity with respect to the target precision. The randomized algorithm proposed in~\cite{wan2021RandPhaseEst} (built in the spirit of qDRIFT \cite{campbell2019randomCompiler} yet distinct and catered specifically for this task) sidesteps the exponentially worse error dependence, and is unique in having no dependence on the number of Hamiltonian terms $\Gamma$. However, this comes at the cost of forfeiting commutator scaling and incurring a quadratic scaling with evolution time.  Further, we note that the approach of \cite{wan2021RandPhaseEst} requires a decomposition of $H$ in the Pauli basis, which may not always be efficiently obtainable or be the most useful decomposition. The multiproduct formula of \cite{aftab2024multi} can also sidestep the exponentially worse error dependence of product formulae, but requires additional ancillary qubits and an amplitude amplification subroutine.

In discussion of our algorithm we first look at a generically available circuit depth bound stated in Table \ref{tab:time-signal} (also obtainable with methods of \cite{chakraborty2025quantum}), before discussing refined complexities in the following subsection. Our algorithm retains the advantageous properties of product formulae -- no additional space overhead, near-linear scaling in the time parameter \( T \), and what we will later argue is commutator-sensitive gate complexity -- while achieving precision \( \varepsilon \) using circuits with depth scaling logarithmically with \( 1/\varepsilon \). Thus, extrapolated product formulae yield an exponential improvement in precision dependence compared to standard product formulae methods. By exploiting a randomized scheme, the cost of this comes at a mild increase in the number of samples: a factor polylog-logarithmic in $1/\varepsilon$. Additionally, our approach works as long as $H$ can be decomposed into efficiently simulable Hamiltonians (i.e.~the same as all product formulae methods), which is a more general data access assumption than Pauli access.


\subsection{Gate complexity}\label{sec:into-commutator}

We will now discuss finer points on the gate complexity of our algorithms, remaining in the setting of Table \ref{tab:time-signal} as a guide. 

\subsubsection{Extrapolated commutator factor}

The complexity of our algorithms depends on a quantity $\lcomm$ which we define below.\footnote{In our detailed theorem statements we consider an adjusted definition where we mildly tighten the domain for the supremum, but for practical purposes this will not make a difference to any conclusions drawn in this manuscript.}

\begin{definition}[Extrapolated commutator factor]\label{def:lcomm}
\begin{equation}
\lambda_{\mathrm{comm}} := 
\sup_{\substack{j \in \mathbb{Z}^+ {\geq p} \\ 1 \leq \ell \leq \lfloor j/p \rfloor}}
\bigg(
\sum_{\substack{j_1, \dots, j_\ell \in \mathbb{Z}^+ {\geq p} \\
j_1 + \cdots + j_\ell = j}}
\prod_{\kappa=1}^{\ell}
\frac{\alpha_{\mathrm{comm}}^{(j_\kappa + 1)}}{(j_\kappa + 1)^2}
\bigg)^{\frac{1}{j + \ell}}.
\end{equation}
\end{definition}

This quantity requires understanding of $\acomm^{(j)}$ at every order $j$. For instance, if $\acomm^{(j)} \leq A\cdot c^j$ for constants $A\geq 1$ and $c\geq 0$, then $\lcomm^{1+1/p} = O((\acomm^{(p+1)})^{1/p})$ \cite[Remark 3.2]{chakraborty2025quantum},\footnote{This means that the respective contributions to gate complexity for extrapolated Trotter circuits and regular Trotter circuits are the same asymptotically.} which is known to be the case of second-quantized Hamiltonians in the plane wave dual basis \cite{childs2021TheoryTrotter} (see also Lemma \ref{lem:lcomm-tight-bound}). In general, knowledge of $\acomm^{(j)}$ at every order is not always available, or known bounds are too large and lead to a divergence in $\lcomm$, as is the case for $k$-local systems (see Sections \ref{sec:intro-k-local} and \ref{sec:k-local} for full discussion). In the rest of this subsection we discuss resolutions of these issues, as well as refinements and alternative bounds on general gate complexities, which we instantiate as rows (b) and (c) of Table \ref{tab:time-signal-2}.

\begin{table}[t]
\resizebox{\columnwidth}{!}{
\renewcommand{\arraystretch}{2.2}
    \centering
    \begin{tabular}{r|c|c|c}
        method & \makecell{max depth per sample \\ $\widetilde{\cO}(\cdot)$ } & \makecell{$k$-local \\ $\widetilde{\cO}(\cdot)$ } & \makecell{fixed $p$ bound \\ $\widetilde{\cO}(\cdot)$ } \\\hline
        \makecell[r]{\textbf{Rand.~Extrap.} \\
        \textbf{Trotterization \hyperref[sec:intro-max-bound]{(a)}} } & $\displaystyle \Gamma \left(\lambda_{\text{comm}} T\right)^{1+1 / p} $ & $\displaystyle \Gamma \Lambda^{1/p} T^{1+1 / p} $ & $\displaystyle \Gamma \left(\Lambda T\right)^{1+1 / p} $ \\
        \textbf{Max scaling \hyperref[sec:intro-max-bound]{(b)}} & $\displaystyle \Gamma \cdot \max\left\{  \Lambda T, (\acomm^{(p+1)})^{1 / p} T^{1+ 1 / p}\right\}  $ & $\displaystyle \Gamma \cdot \max\left\{  \Lambda T, \Lambda^{1/p} T^{1+ 1 / p}\right\}  $ & $\displaystyle \Gamma \cdot \max\left\{  \Lambda T, A_p T^{1+ 1 / p}\right\}  $ \\
        \makecell[r]{\textbf{Interpolative} \\
        \textbf{scaling \hyperref[sec:intro-partial]{(c)}}} & $\displaystyle \Gamma_A(\lcomm  T)^{1+ 1 / p} +\Lambda_B^2 T^2  $ & $\displaystyle \Gamma_A \Lambda^{1/p} T^{1+ 1 / p} +\Lambda_B^2 T^2 $ & $\displaystyle \Gamma_A (\Lambda T)^{1+ 1 / p} +\Lambda_B^2 T^2 $
    \end{tabular}
    }
    \caption{\textbf{Refinements to gate complexity.} We catalog two alternative expressions for the circuit depth using our methods for the same task in Table \ref{tab:time-signal} -- one which is independent of $\lcomm$, and the second which provides a strictly interpolative scaling between our generic scaling in \ref{tab:time-signal} and that of the randomized compiler of \cite{wan2021RandPhaseEst} for a matrix decomposition $H = H_A + H_B$. For terseness of presentation we have assumed constant error (thus dropping factors of $\log(1/\varepsilon)$ from each cell). The third column details complexities for the special case where $H$ is a $k$-local Hamiltonian with constant onsite energy and constant $k$. The fourth column details complexities where there is a known (non-trivial) commutator bound $(\acomm^{(p+1)})^{1/p}\leq A_p (< 4\Lambda)$ for some fixed order $p$, but not for generic order. In each column at least one of our refined scalings (row (b) or (c)) unconditionally improve the generic scaling in row (a), under the assumption that $(\acomm^{(p+1))^{1/p}} = \lcomm^{1+1/p}$. As in Table \ref{tab:time-signal}, all expressions generalize naturally to our algorithms for both Task \hyperref[eq:task-1]{1} and \hyperref[eq:task-2]{2}, and expressions for gate depth extend to prior art in extrapolated product formulae.  \label{tab:time-signal-2}}
\end{table}

\subsubsection{Max bound}\label{sec:intro-max-bound}

In many settings, one may only understand the behavior of $\acomm^{(j)}$ at some finite order $j$, rather than at every order. For instance, one may use some analytical or numerical tricks to understand commutator behavior at low order \cite{campbell2021earlyFTHubbard, blunt2025monteCarlo, maxwell2026practical}, without being able to elucidate a meaningful bound at every order. In this case, we can specify an alternative bound on the gate complexity, instantiated as row (b) in Table \ref{tab:time-signal-2}, and we present full discussion in Section \ref{sec:lcomm-fixed-order}. In general, we find that we can always replace 
\begin{equation}
    (\lcomm T)^{1+1/p} \rightarrow \max\left\{  \Lambda T,\,(\acomm^{(p+1)})^{1 / p} T^{1+ 1 / p}\right\}\,,
\end{equation}
i.e~the resulting gate complexity now depends instead on the standard commutator factor at a chosen fixed order $p$, rather than the extrapolated commutator factor.
This retains the exponentially-improved dependence on precision compared to standard product formulae, and the remaining contribution to the gate complexity is bounded by the maximum of the complexity of product formulae, or (optimal) linear dependence in the simulation time without commutator scaling. In this way, the performance of our algorithm either behaves in line with product formulae (with exponentially improved dependence on precision), or scaling similar to that of QSVT, all the while using minimal ancillary qubits. All together, this is a way to bound the performance of Randomized Extrapolation Trotterization without requiring any additional information beyond what is needed to bound the performance of regular product formulae circuits.

\subsubsection{Doubly-Randomized Extrapolation Trotterization}\label{sec:intro-partial}

Certain matrices, particularly physically-motivated matrices, may exhibit particular structure in the decomposition in Eq.~\eqref{eq:matrix}. For instance, electronic-structure Hamiltonians in second quantization can often be dominated by a few large terms and otherwise exhibit a long tail. Here, the randomized algorithm of \cite{wan2021RandPhaseEst} could be useful in exploiting the fact that in these circumstances $\Lambda$ may be smaller than the number of constituent terms $\Gamma$. One may go further -- \cite{gunther2025phase} show that it is possible to \textit{partially} randomize product formulae -- by partitioning $H$ as $H=H_A + H_B$ into an $H_A$ with a small number of terms $\Gamma_A << \Gamma$ of large size, and $H_B$ with the remaining long tail of terms where $\Lambda_B << \Lambda$ if the tail is small. We demonstrate that one can use this idea to compile product formulae that appear in Randomized Extrapolation Trotterization with an inner extrapolation loop. This leads to a doubly randomized compilation of Trotter extrapolation. For the task of compiling a time signal, the complexity is stated in row (c) of Table \ref{tab:time-signal-2}. In the general case, we can always replace
\begin{equation}
    \Gamma (\lcomm T)^{1+1/p} \rightarrow \Gamma_A(\lcomm  T)^{1+ 1 / p} +\Lambda_B^2 T^2
\end{equation}
We see that this interpolates in between our scaling in row (a) and that of the algorithm of \cite{wan2021RandPhaseEst}, thus giving the best scaling of all schemes that only use one ancillary qubit (assuming $\lcomm^{1+1/p} = O( (\acomm^{(p+1)})^{1/p})$). We present a full discussion and algorithm construction in Section \ref{sec:partial-randomization}. 

\subsubsection{$k$-local systems}\label{sec:intro-k-local}

For $k$-local systems on $n$ qubits with maximum single-site energy $g$, it is known that $(\acomm^{(p+1)})^{1/p} = O(kg p\Lambda^{1/p})$ \cite{childs2021TheoryTrotter}. When the product formula order satisfies $p>1$ this can be a significant improvement in dependence on Hamiltonian parameters over non product-formulae methods which scale linearly in the Hamiltonian norm $\Lambda$. However, for extrapolated methods, direct substitution of this bound into our expression for the extrapolated commutator factor $\lcomm$ leads to a divergent quantity. We show that in our extrapolation algorithms (as well as prior art \cite{watson2024exponentiallyReducedTrotter, chakraborty2025quantum}) we can instead consider $\lcomm \rightarrow \lcommklocal$ where for time signal compilation to additive precision $\varepsilon$ we have

\begin{equation}
        (\lcommklocal)^{1+1/p}  = {O}\!\left( kg (p\Lambda^{1/p} + g \log(n g T/\varepsilon)) \right) \,,
\end{equation}
and for general applications, there is an additional additive logarithmic term. 
Thus, the key features of commutator scaling for $k$-local systems are replicated, with mild logarithmic overhead. The analysis is based on recent insights of Mizuta \cite{mizuta2026commutator} on truncated BCH expansions, and we present this in detail in Section \ref{sec:k-local}.

\subsubsection{Exploiting Hamiltonian symmetries}

When product formulae are used on input states with a well-defined Fermion number, the gate complexity has been shown to depend on the Fermionic semi-norms of nested commutators rather than operator norms. This leads to a commutator dependence $\acomm \rightarrow \alpha_{\text{comm}}^{(\eta)}$ which is tighter for systems of low Fermion number compared to number of orbitals \cite{mcardle2022ExploitingFermionNumber, su2021NearlyTightTrottInerElect, low2023complexityTrotter}. We demonstrate that in our setting, gate complexities depend on an analogous quantity \( \lcomm \rightarrow \lambda_{\text{comm}}^{(\eta)}\), and leads to similar scaling benefits. We present the explicit form of \(\lambda_{\text{comm}}^{(\eta)}\) and comparative analysis in Section \ref{sec:Fermionic}. More generally, we show that our randomized extrapolation procedure attains refined scaling for Hamiltonians that possess a symmetry, given an input state that belongs to a particular symmetry sector.

\subsection{Using Trotter extrapolation in practice}\label{sec:intro-extrapolation-in-practice}

So far we have only discussed asymptotic complexities. Here we add some notes on using extrapolated product formulae in practice. Our ensuing discussion also applies to prior art in product formulae extrapolation and all settings that we consider in this work -- for simplicity of exposition, we continue using time signal compilation as an example.

First, extrapolation may be used at will without asking for rigorous error bounds, in the same way it is used extensively in numerical methods. However, one might encounter a scenario for which it is unclear if extrapolation improves the estimate beyond performing no extrapolation at all, for instance if verification is not efficient. For this intermediate setting, we provide the following lemma.

\begin{lemma}[Radius of convergence]\label{lem:radius-of-convergence}
    Let $\PC(sT)$ denote a product formula of fixed order $p$ and step size $sT$. We have 
    \begin{equation}
        \tr[\rho \PC(sT)^{1/s}] - \tr[\rho e^{-iHT}] = \sum_{j \in \sigma \mathbb{Z}_+ \geq p} s^j \delta_j\,, \end{equation}
    where
    \begin{equation}
        |\delta_j| \leq (c \lcomm T)^{j(1+1/p)}\,,
    \end{equation}
    for a specified constant $c$.
    Thus, if $s(c \lcomm T)^{(1+1/p)} \leq C$ for any $C<1$, the error takes the form of a geometrically decaying power series.
\end{lemma}

Lemma \ref{lem:radius-of-convergence} gives an upper bound on the radius of convergence for Trotter error. Specifically, it states that above a specified Trotter step number $\propto (\lcomm T)^{1+1/p}$, extrapolation decreases the error at an exponential rate with each additional extrapolation point. An end user who does not need a priori error bounds, though wishes to guarantee that extrapolation is reducing algorithmic error may use Lemma \ref{lem:radius-of-convergence} along with their own extrapolation schedule of choice. We refer the reader to proof of Lemma \ref{lem:generic-Richardson-error} and surrounding discussion for the proof.

Second, we remark that in order to write down asymptotic complexities free of extrapolation parameters (as in discussion in earlier subsections), in this manuscript we use the extrapolation scheme of Low-Kliuchnikov-Wiebe \cite{low2019Multiproduct}. However, potentially stronger results for algorithm complexities are achievable with different extrapolation schemes. We also provide bounds on resources for a generic extrapolation scheme, which allows an end user to specify gate counts with rigorous approximation guarantees using their own extrapolation scheme.

\begin{theorem}[Time signal with user-defined extrapolation scheme (informal version of Theorem \ref{thm:time-signal-generic})]\label{thm:heuristic-extrapolation}
    Consider any $p$-th order product formula and consider any extrapolation scheme that cancels powers $p, p+\sigma, ..., p+(m-2)\sigma $ where $\sigma = 2$ if $\PC$ is a symmetric product formula and $\sigma=1$ otherwise. The time signal can be estimated to precision $\varepsilon$ using $C_{\mathrm{sample}}$ calls to Hadamard tests with gate depths at most $C_{\mathrm{gate}}$ where
    \begin{equation}\label{eq:depth-generic-extrapolation}
        C_{\mathrm{gate}} = O\left( \Gamma (\lcomm T)^{1+1/p} \left(\frac{4\|\vec{b}\|}{\varepsilon}\right)^{\frac{1}{\sigma(m-1)+p}} \right)\,,\quad C_{\mathrm{sample}} = O\left( \frac{\|\vec{b}\|_1^2}{\varepsilon^2}\right)\,,
    \end{equation}
    where  $\|\vec{b}\|_1$ is the condition number of the extrapolation schedule (see Definition \ref{def:extrapolation} to come). 
\end{theorem}


We remark that this theorem actually exploits a refined analysis of the error expansion of time evolution operators that the extrapolation scheme of Low-Kliuchnikov-Wiebe is unable to take advantage of \cite{low2019Multiproduct} -- this is articulated in a stronger second exponent in the expression in Eq.~\eqref{eq:depth-generic-extrapolation} which is never larger than $1/p$. We provide further discussion and proofs in Section \ref{sec:extrapolation-error}.

\begin{figure}[t]
    \centering
    \includegraphics[width=\linewidth]{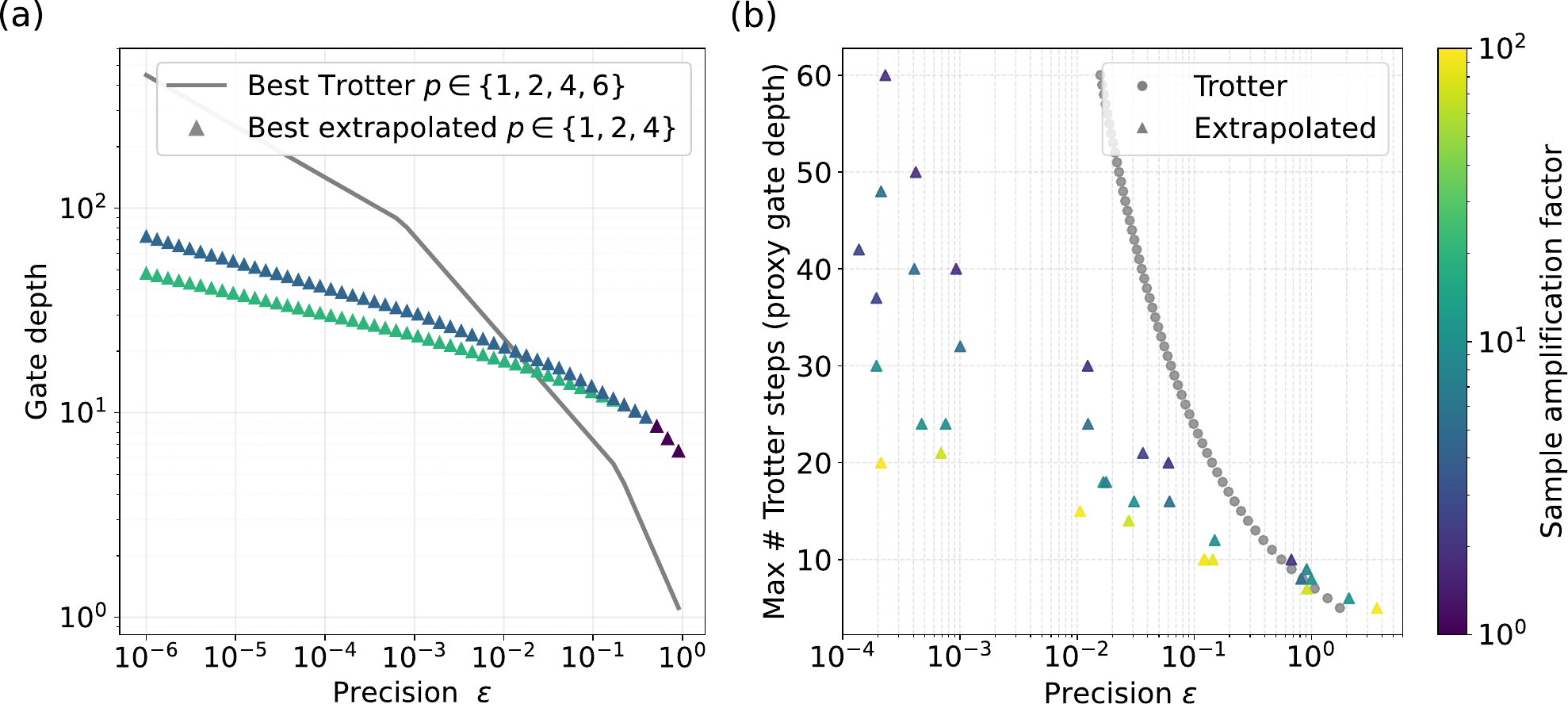}
    \caption{\textbf{Summary of numerical results.}  \textit{(a): Bounds.} For an electronic structure problem we plot the best available bounds on (a quantity proportional to and smaller than) the gate depth for regular Suzuki-Trotter circuits against Randomized Extrapolation Trotterization across orders $p \in {1,2,4,6}$. Here, all extrapolated data requires a multiplicative sample overhead increase of factor 10.  \textit{(b): Empirical Error.} For a spin chain physics problem on 8 qubits we extract the exact precision for $p=2$. In both plots we use the color of the data points to indicate the multiplicative sample amplification factor required of Randomized Extrapolation Trotterization for Task \hyperref[eq:task-1]{1(i)}.}
    \label{fig:intro-numerics}
\end{figure}

Finally, in Section~\ref{sec:numerical-comparison} we numerically demonstrate the promise of heuristic choice of extrapolation schemes using Theorem \ref{thm:heuristic-extrapolation}, as well as heuristic use of Randomized Extrapolation Trotterization. First, we observe stronger performance over the well-conditioned extrapolation scheme of \cite{low2019Multiproduct} (from which our asymptotic bounds are derived) by performing heuristic optimization over extrapolation parameters. For the plane-wave electronic structure Hamiltonian considered, we then compare gate complexities over the optimal Trotter-Suzuki formula for any order for a particular problem posed as a candidate for quantum advantage \cite{babbush2018LowDepthQSimMaterial}. There, under certain assumptions made to generate constant factors, we find (bounds on) gate complexities for extrapolated product formulae fall below that of standard Trotter formulae at any order for any mild target precision below $10^{-2}$, if a 10$\times$ factor additional sample overhead is allowed (flexible tradeoff of more/fewer samples is possible). Moreover, we find that extrapolation schedules generated by simple brute-force search can improve on the asymptotically well-conditioned LKW extrapolation schedule greatly. Second, to gain insight into performance of Randomized Extrapolation Trotterization in practice, for a small system size spin Hamiltonian we exactly evaluate the empirical (true) error and observe that Randomized Extrapolation Trotterization attains a much smaller precision over regular Trotter circuits using circuits of smaller gate depth. 

We present a summary of numerical results in Figure \ref{fig:intro-numerics}, and present full numerical details in Section~\ref{sec:numerical-comparison}. In Figure \hyperref[fig:intro-numerics]{2(a)} we plot bounds on gate depth for the best performing extrapolated product formulae compared to the best performing Suzuki-Trotter product formula over a range of orders for the task of compiling a time signal. In Figure \hyperref[fig:intro-numerics]{2(b)} we plot the empirical error for various extrapolation schemes against empirical Trotter error. In both cases we see improvement over regular product formulae whilst using mild sample overhead (denoted by the color of the data points).

\subsection{Applications}\label{sec:applications}

In this subsection we discuss applications for our algorithm framework. First, we demonstrate an application for ground state energy estimation. Second, we show how to estimate Green's functions in many-body physics. Finally, we demonstrate an application for estimating distribution information of time-evolved states.

\subsubsection{Ground state energy estimation}

We consider the task of
estimating the ground energy $E_{0}$ of a given Hamiltonian $H$. Such algorithms require an ansatz state $\rho$ of good overlap with the ground space, and we denote this overlap $\eta = \tr[\Pi_0 \rho]$, where $\Pi_0$ is the projector onto the ground space. We use the approach proposed by Lin and Tong \cite{lin2022HeisenbergLimited} in which we approximate the Heaviside function using a Fourier series approximation. The key idea is the Heaviside function serves as a filter for eigenvalues, and prepares a cumulative distribution function corresponding to the Hamiltonian's eigenspectrum, weighted by $\rho$. Thus, by probing values of the (approximate) Heaviside function, one can determine the ground state energy by finding the first step of the cumulative distribution. It is important to note that Lin and Tong \cite{lin2022HeisenbergLimited} specify their algorithm in terms of time evolution oracles, which needs to be instantiated to perform full end-to-end analysis. We instantiate their oracle with Randomized Extrapolation Trotterization, compiling the cumulative distribution function in a randomized fashion all in one go as a special case of Task \hyperref[eq:task-1]{1(i)}.

\begin{table}[H]
\renewcommand{\arraystretch}{2.2}
    \centering
    \resizebox{\columnwidth}{!}{
    \begin{tabular}{r|c|c|c||c}
        method & \makecell{ancillary \\qubits } & \makecell{max depth per sample \\ $\widetilde{\cO}(\cdot)$ } & \makecell{samples \\${\widetilde{\cO}}(\cdot)$} & \makecell{max depth per sample \\ $k$-local\ $\widetilde{O}(\cdot)$} \\\hline
        \makecell[r]{Textbook QFT + \\ qubitization \cite{low2016HamSimQubitization}}  & \makecell[c]{$\lceil \log(\Gamma) \rceil + $\\ $ \lceil \log(\varepsilon^{-1}) \rceil$ + 2} & $\displaystyle \Gamma \Lambda\, {\eta}^{-1} {\varepsilon}^{-1}  $ & $\displaystyle  {\eta}^{-1}$ & $\displaystyle \Gamma \Lambda\, {\eta}^{-1} {\varepsilon}^{-1}  $ \\
        \makecell[r]{QET-U + qubitization \\\cite{dong2022groundStatePrep, low2016HamSimQubitization}}  & $\lceil \log(\Gamma) \rceil +3$ & $\displaystyle \Gamma \Lambda \varepsilon^{-1} $ & $\displaystyle  {\eta}^{-1}$ & $\Gamma \Lambda \varepsilon^{-1}$ \\
        \makecell[r]{Textbook QFT \\+ Trotter}  & $\lceil \log(\varepsilon^{-1}) \rceil$ & $\displaystyle \Gamma ({\acomm^{(p+1)}})^{1/p} {\eta}^{-(1+2/p)} {\varepsilon}^{-(1+1 / p)}  $ & $\displaystyle  {\eta}^{-1}$ & $\Gamma \Lambda^{1/p} {\eta}^{-(1+2/p)} {\varepsilon}^{-(1+1 / p)}$ \\
        \makecell[r]{QET-U + Trotter \\\cite{dong2022groundStatePrep}}  & 1 & $\displaystyle \Gamma ({\acomm^{(p+1)}})^{1/p} {\eta}^{-1/p} {\varepsilon}^{-(1+1 / p)}  $ & $\displaystyle  {\eta}^{-1}$ & $\Gamma \Lambda^{1 / p} {\eta}^{-1/p} {\varepsilon}^{-(1+1 / p)}$ \\
        \makecell[r]{Random compiler \\\cite{wan2021RandPhaseEst}}  & 1 & $ {\Lambda^2}{\varepsilon^{-2}}$ & $\displaystyle  {\eta}^{-2}$ & ${\Lambda^2}{\varepsilon^{-2}}$ \\
        \makecell[r]{\textbf{Rand.~Extrap.} \\ \textbf{Trotterization}} & 1 & $\displaystyle \Gamma  \lcomm^{1+1 / p} {\varepsilon}^{-(1+1 / p)} $ & $\displaystyle  {\eta}^{-2} $ & $\Gamma  \Lambda^{1 / p} {\varepsilon}^{-(1+1 / p)} $
    \end{tabular}
    }
    \caption{\textbf{Ground state energy estimation complexity comparison.} We consider a Hamiltonian $H = \sum_{\ell=1}^{\Gamma} h_{\ell} $, with $\Lambda = \sum_{\ell=1}^{\Gamma} \|h_{\ell}\|$, and $\eta = \tr[\rho \Pi_0]$ the initial state overlap. In the final column we explicitly give the gate depth for $k$-local Hamiltonians, assuming constant per-site energy and constant $k$.  \label{tab:phase-estimation}} 
\end{table}

We present our algorithm complexity in Table \ref{tab:phase-estimation} and contextualize it with other approaches in the literature. We compare our algorithm against textbook phase estimation based on the quantum Fourier transform (QFT), the QET-U eigenvalue processing algorithm \cite{dong2022groundStatePrep}, and the statistical phase estimation algorithm of \cite{wan2021RandPhaseEst}. The former two algorithms were originally presented with a time evolution oracle, so we instantiate that oracle with qubitization \cite{low2016HamSimQubitization} and Trotter formulae. In a similar comparison of our method with other modern approaches for compiling time evolution in Table \ref{tab:time-signal}, we observe the following trade-offs. Disregarding Hamiltonian parameters, the algorithm with the best gate complexity is the QET-U algorithm with time evolution instantiated by qubitization. However, this may not be the most amenable approach for early hardware due to arguments made prior. On the other hand, algorithms that utilize regular Trotter circuits introduce an additional gate overhead of ${\eta}^{-1/p}$, whereas other approaches have gate complexity polylogarithmic in $\eta$. 

Our approach combines many favorable aspects of prior art: it uses one ancillary qubit, it attains close to Heisenberg scaling in $\varepsilon$ (for higher order product formulae), and the gate complexity depends logarithmically on the ground state overlap. To our knowledge, our approach is the only method that attains these three features simultaneously. Moreover, our algorithm can have significantly refined dependence on Hamiltonian parameters, as is showcased in the final column of Table \ref{tab:phase-estimation} for the example of $k$-local systems. We recall from Subsection \ref{sec:intro-partial} that interpolative scaling with that of \cite{wan2021RandPhaseEst} is always available, so our approach achieves the best gate depth with respect to Hamiltonian parameters for $k$-local systems (up to logarithmic factors).

\subsubsection{Green's Functions}

Green's functions characterize how a quantum many-body system responds to external perturbations. It can be used directly to evaluate physical quantities such as the spectral density or dynamic structure factor, and its evaluation is a central computational step in DMFT calculations. In this section we focus on Green's function in the frequency domain, although estimation of the time-domain Green's function is also readily available with our methods. We consider the resolvent operator, defined as
\begin{equation}
R(\omega, \eta_{\text{broad}}) = (\omega + i\eta_{\text{broad}} - \hat{H})^{-1},
\end{equation}
where $\hat{H} = H-E_0$ is a normalized Hamiltonian with ground state energy 0, and $\omega$ is the energy variable. A small broadening factor ($\eta_{\text{broad}} > 0$) is introduced to shift the poles of the resolvent into the complex plane -- this determines the resolution of spectral quantities. The (retarded) Green's function is defined as 
\begin{equation}\label{eq:greens-function}
    G(\omega, \eta_{\text{broad}}) = \bra{\psi_0} \hat{a}\, R(\omega, \eta_{\text{broad}}) \hat{a}^{\dag} \ket{\psi_0}
\end{equation}
where $\ket{\psi_0}$ is the ground state and $\hat{a}^{\dag}$ and $\hat{a}$ are creation and annihilation operators respectively.

We observe that the resolvent operator is nothing other than a shifted inverse function. Thus, by considering standard mappings of creation/annihilation operators to Pauli operators, the estimation of Green's function falls neatly into that of Task \hyperref[eq:task-1]{1(i)}. We use the Fourier series approximation of the resolvent found in \cite{keen2021quantumGSGreens} and our algorithm framework to give the following result.

\begin{theorem}[Green's function evaluation (informal version of Theorem~\ref{thm:resolvent_estimate})]
The Green's function defined in Eq.~\eqref{eq:greens-function} can be estimated to additive error~$\varepsilon$ and arbitrary constant success probability with $\widetilde{O}\big(1/\eta_{\text{broad}}^2 \, \varepsilon^2\big)$ runs of circuits with maximum circuit depth $\widetilde{O}\big(\Gamma (  \lambda_{\text{comm}}/{\eta_{\text{broad}}} )^{1+1/p}\big)$.
\end{theorem}

We present full details and discussion in Section \ref{sec:greens}.

\subsubsection{Time-evolved states}

In our final application, we estimate the measurement distribution of a time-evolved state in the computational basis. This can also be simply generalized to measurement distributions in any (efficiently-reachable) basis.

We remark that
our algorithm for Task \hyperref[eq:task-1]{1(ii)} allows for the estimation of any probability mass $p_i := |\bra{i}e^{-iHT}\ket{\psi}|^2$ to additive error with the same complexity as in Table \ref{tab:time-signal}. However, by using our algorithm for Task \hyperref[eq:task-2]{2} we can learn about the full probability distribution at once.

\begin{theorem}[Sampling distribution of time-evolved states (informal version of Theorem \ref{thm:distribution-time-evol})]\label{thm:intro-distribution-time-evol}
    We give an algorithm to approximate the probability distribution with values $p_i := |\bra{i}e^{-iHT}\ket{\psi}|^2$ with error $\varepsilon$ in $\ell_2$-distance, with arbitrary constant success probability. The algorithm uses $\widetilde{O}(1/\varepsilon^2)$ runs of circuits with one ancillary qubit and maximum gate depth $ O\big( \Gamma \left( \lambda_{\mathrm{comm}} \, T \right)^{1 + \frac{1}{p}} \log(1/\varepsilon) \big)$.
\end{theorem}

We see that the complexity here is equivalent to that of the task of compiling a time signal in Table \ref{tab:time-signal}, up to logarithmic factors in the sample complexity. Let us briefly discuss how this result compares to the standard route of time-evolving a state with a chosen Hamiltonian simulation algorithm directly and measuring the approximated quantum state in the computational basis. In this case, it is sufficient to take the standard gate complexity and ancillary overhead to instantiate the time evolution unitary $e^{-iHT}$ to operator norm error $\varepsilon$, and the sample overhead is $O(1/\varepsilon^2)$. Thus, the comparison of our approach to other algorithms is almost entirely equivalent to that of Table \ref{tab:time-signal} for the task of compiling a time signal -- the only difference is that the sample overhead of our algorithm for Theorem \ref{thm:intro-distribution-time-evol} is a poly-log-logarithmic factor larger than that for our algorithm to compile the time signal. We present further details in Section \ref{sec:time-evol-distribution}.

\subsection{Discussion}

We have performed a broad study of algorithms for matrix function problems where data is extracted from quantum circuits consisting of Hadamard tests of product formulae circuits. 
By choosing a particular (randomized) scheduling of Trotter step numbers and classical postprocessing, we can obtain exponentially reduced gate complexities over using product formulae alone in a parameter corresponding to the precision of simulating time evolution. 
Mathematically, this corresponds to extrapolating product formulae towards zero \textit{algorithmic error}. Meanwhile, we show that the other contributions to the gate complexity behave in the same refined manner as standard product formulae in many settings -- this extends the observations of \cite{chakraborty2025quantum} which demonstrates refined performance when the commutator factor obeys a certain analytical scaling. We also demonstrate an interpolative scaling with the randomized algorithms of \cite{wan2021RandPhaseEst, wang2023qubitefflinalg}, meaning performance can never be worse than this prior art.

Whilst we have explicitly written circuits that utilize controlled time evolution, we remark that in implementation in certain settings, the product formula itself does not have to be explicitly controlled, and can instead be conjugated by controlled gates \cite[Appendix E]{lin2022HeisenbergLimited}; \cite{yoshioka2025krylovDiagonalization}. Recent study has also asked what properties of $A$ can be determined from direct time evolution and no control gates at all \cite{clinton2024quantumPhaseEstimation}, or limited control \cite{patel2026quantumPhaselift}. We frame the following interesting open question in the general context of matrix functions: what properties can be gleaned from $f(H)$ given non-controlled time evolution of $A$?

Our methods may yield promising performance in practice in a heuristic or semi-heuristic fashion in a similar spirit to the efficacy of product formulae observed in experiment and simulation. For this, we provide intermediate results for users who may not need error guarantees (but wish to know that extrapolation is improving performance), or for users who wish to numerically optimize their own extrapolation schedules. For instance, one can optimize an extrapolation scheme along whatever one's preferred axis is between gate complexity and sample complexity. Our own numerical analysis demonstrates that heuristic optimization can be useful in practice and significantly improve upon gate complexities of known well-conditioned extrapolation schemes, particularly at larger product formula order. We leave open space for exploration of refined optimization schemes beyond the simple brute-force search considered in our work.

A final important direction is to see if additional useful properties of product formulae also hold for extrapolated product formulae. It is known for $k$-local systems product formulae perform better on input states that lie in a low-energy subspace \cite{Sahinoglu2021Hamiltonian, hejazi2024LowEnergyProduct, mizuta2025trotterizationLowEnergy}, and that product formulae have stronger average-case guarantees than worst-case bounds provide \cite{chen2021conTrotter, Zhao2021HamiltonianSW}. Whether the same applies for extrapolated product formulae is an open question which we leave for future work.


\section{Preliminaries}\label{sec:preliminaries}

\textit{Basic notation.} For vectors, we denote $\|\vec{v}\|_p$ as the $\ell_p$ norm.
For operators, we denote $\|H\|$ as the spectral norm and denote $\|H\|_p$ as the Schatten-$p$ norm. To facilitate colloquial exposition, we will use the notation $a \overset{\varepsilon}{\approx} b$ to signify $|a-b|\leq \varepsilon$ for scalars and $A \overset{\varepsilon}{\approx} B$ to signify $\|A-B\|\leq \varepsilon$ for operators. We will often discuss nested commutators, for which we adopt the simplifying notation
\begin{equation}
    \left[H_{\gamma_{1}} ,\cdots, H_{\gamma_j}\right]:= \left[H_{\gamma_{1}}, \cdots\left[H_{\gamma_{j-1}}, H_{\gamma_j}\right]\right]\,.    
\end{equation}
We use $\widetilde{O}(\cdot)$ as big-$O$ notation with polylogarithmic factors omitted. 

\textit{Product formulae.} In our work we consider general product formulae, which encompass Lie--Trotter and higher-order Suzuki formulas.

\begin{definition}[Staged Product Formula {\cite{childs2021TheoryTrotter}}]\label{def:staged_pf}
Given a matrix $H$ of the form in Definition \ref{def:matrix}, a staged product formula \( \PC(t) \) of symmetry class $\sigma$ is an approximation to \( e^{-iHt} \) of the form
\begin{equation}\label{eq:staged_product_formula}
\PC(t) := \prod_{\upsilon=1}^{\Upsilon} \prod_{\gamma=1}^{\Gamma} e^{-i t \, a_{(\upsilon, \gamma)} \, H_{\pi_{\upsilon}(\gamma)}},
\end{equation}
where:
\begin{itemize}
  \item \( \Upsilon \in \mathbb{Z}^+ \) is the number of \emph{stages},
  \item \( a_{(\upsilon, \gamma)} \in \mathbb{R} \) are real-valued \emph{coefficients},
  \item \( \pi_\upsilon\) are {permutations} on the $\Gamma$ constituent terms of $H$,
  \item $\sigma$ is the symmetry class where $\sigma=2$ if $\PC(t)^{-1} = \PC(-t)$, or $\sigma=1$ otherwise
\end{itemize}
We denote $a_\mathrm{max} \coloneqq \max_{\upsilon,\gamma} \abs{a_{(\upsilon,\gamma)}}$.
\end{definition}

The order of a product formula denotes how well it approximates \( e^{-iHt} \) in the small-time limit. Specifically, a product formula is defined to be of order \( p \in \mathbb{Z}^+ \) if it satisfies
\begin{equation}\label{eq:product_formula_error}
\|\PC(t) - e^{-iHt}\| = O(t^{p+1}) \quad \text{as } t \to 0.
\end{equation}
To give a concrete example, the first order Trotter formula takes the form 
\begin{equation}
    \PC_1(t) := \prod_{\gamma=1}^{\Gamma} e^{-itH_{\gamma}}\,,
\end{equation}
i.e.~it has $\Upsilon = 1$, and it is non-symmetric. Higher order Suzuki formulae are symmetric, and the order $2k$-th formula has number of stages $\Upsilon = 2\cdot 5^{k-1}$. Throughout, we assume a cost model where each $e^{-itH_{\gamma}}$ can be implemented in unit time (gate depth) for any $t$.\footnote{In practice, if $H_{\gamma}$ is a Pauli string with $w$ non-identity Paulis, the depth of implementing $e^{-itH_{\gamma}}$ is $\log(w)$, with $O(w)$ $T$ gates.} When using product formulae for Hamiltonian simulation of some target time $T$, it is common to divide $T$ into $r$ equal pieces, so that 
\begin{equation}
    \PC(T/r)^r \approx e^{-iHT}\,,
\end{equation}
where the error in approximation is $O(T^{p+1}/r^p)$ and thus controllable by setting a large enough $r$. We refer to $r$ throughout as the \textit{Trotter step number}, and we will frequently use its inverse which we denote as $s=1/r$. It will be useful to consider the effective Hamiltonian $\Heff$ that a product formula generates, defined by the relation
\begin{align}\label{eq:effective-hamiltonian}
        \cP(t) = e^{-it\Heff(t)}\,.
\end{align}
Throughout all our analytical results we assume that $p = O(1)$.

\textit{Matrix functions.} When we refer to matrix functions throughout our manuscript we consider the eigenvalue transform.

\begin{definition}[Eigenvalue transform]\label{def:eigval-transform}
For Hermitian matrix $H$ with eigendecomposition $H = \sum_{i} e_i \ketbra{e_i}$ we denote
    \begin{equation}
        f(H) = \sum_{i} f(e_i) \ketbra{e_i}\,.
    \end{equation}
\end{definition}

Non-hermitian matrices can be embedded in a larger Hermitian matrix with the aid of one ancillary qubit as $\bar{H} = \ket{0}\!\!\bra{1}\otimes H + \ket{1}\!\!\bra{0} \otimes H^{\dag}$. In this dilation, the eigenvalue transformation of $\bar{H}$ for odd functions gives the singular value transform of $H$ in its off-diagonal block.

\textit{Richardson extrapolation.} Finally, we introduce some preliminaries about Richardson extrapolation. 

\begin{definition}[Extrapolation schedule]\label{def:extrapolation}
    Consider a functional form
    \begin{equation}
        g(\delta) = g(0)  + \sum_{j=1}^n g_j \delta^{\sigma_j} + O(\delta^{\sigma_{n+1}}) \,,
    \end{equation}
    where $\sigma_j$ are ascending integer powers and $\delta$ is some small parameter. We say that an extrapolation schedule of order $m\leq n$ is a list of tuples $\{(b_k, q_k) \}_{k=1}^m$ that satisfy 
    \begin{equation}\label{eq:vandermonde-system}
        V \vec{b} = (1,0,...,0)
    \end{equation}
    where $V$ is the $m \times m$ (generalized) Vandermonde matrix with entries $V_{jk} = (\delta/q_k)^{\sigma_{j-1}}$ with $\sigma_0 :=0$. This guarantees that
    the quantity $G^{(m)}(\delta) = \sum_{k=1}^m b_k\, g(\delta/q_k)$ satisfies
    \begin{equation}
        \left|G^{(m)}(\delta) - g(0)\right| \leq \|\vec{b}\|_1\,\Big(\sum_{j=m}^n g_j \delta^{\sigma_j} + O(\delta^{\sigma_{n+1}})\Big)\,,
    \end{equation}
    i.e.~that $G^{(m)}(\delta)$ serves as an estimator of $g(0)$ to leading order error $O(\delta^{\sigma_m})$, amplified by $\|\vec{b}\|_1$. We refer to $\|\vec{b}\|_1$ as the condition number of the extrapolation schedule.
\end{definition}

Throughout our manuscript we will index extrapolation factors in ascending order $q_1 \leq \dots \leq q_m$. A useful fact is that for structured error series where the error powers to be canceled are equally spaced, $\sigma_j = a j$ with $a \in \mathbb{Z} \setminus \{0\}$, the system~\eqref{eq:vandermonde-system} admits the closed-form solution
\begin{equation}\label{eq:exact-formula}
b_k = (V^{-1})_{k1} = \prod_{i \neq k}\frac{c_i^{a}}{c_i^{ a} - c_k^{a}}\,,\qquad c_k := \delta/q_k\,.
\end{equation}
Otherwise, one may solve~\eqref{eq:vandermonde-system} numerically.

\section{Series expansion for Trotterized time evolution}

In this section, we derive a error series expansion for product formulae, and use it to demonstrate conditions under which extrapolated product formulae obtain exponentially suppressed errors. This will form the theoretical basis of our Randomized Extrapolation Trotterization protocol. 

\subsection{Operator error series}

Our first goal is to express Trotterized time evolution as a power series in the inverse Trotter step number \( s:=1/r \). Characterization of this series expansion will be the key tool to allow use of Richardson extrapolation. We start by quoting a specific incarnation of the Magnus expansion.


\begin{lemma}[Effective Hamiltonian Error Series (\cite{watson2024exponentiallyReducedTrotter}, Lemma 2)]\label{lem:Effective_Hamiltonian_Error_Series}
    Let $\cP(t)$ be an $\Upsilon$-staged $p$-th order product formula (as defined in Definition \ref{def:staged_pf}) of symmetry class $\sigma$, and let $\Heff(t)$ be the corresponding effective Hamiltonian of $\cP$ defined by the relation in Eq.~\eqref{eq:effective-hamiltonian}. Suppose that there exists $J\in \mathbb{Z}_+$ and $C\in\mathbb{R}_+$ such that $\sup_{j\geq J} \acomm^{(j)}\left(a_\mathrm{max}\Upsilon \abs{t}\right)^j\leq C$.
    Then the effective Hamiltonian can be written as a convergent series
    \begin{align}
        \Heff(t) = H + \sum_{j = 1}^\infty E_{j+1} t^j \,,
    \end{align}
    where 
    $E_j$ has bounded size
    \begin{align}
        \norm{E_j} \leq \frac{(a_\mathrm{max}\Upsilon)^j}{j^2}\acomm^{(j)}\,.
    \end{align}
\end{lemma}

We can refine the above statement for higher order or symmetric product formulae.

\begin{lemma}[Refined effective Hamiltonian error series](\cite[Lemma 3]{watson2024exponentiallyReducedTrotter})\label{lem:hamiltonian-error-series2}
    Assume the conditions of Lemma \ref{lem:Effective_Hamiltonian_Error_Series}. The error operators $E_{j+1}$ are zero for all $j<p$, and for symmetric product formulae $E_{j+1}=0$ for odd $j$. Thus, we can write
    \begin{align}
        \Heff(t) = H + \sum_{j \in \sigma \mathbb{Z} \geq p}E_{j+1} t^j\,, 
    \end{align}
    where $\sigma \in \{1,2\}$ denotes the symmetry class of the product formula.
\end{lemma}

For convenience, we express the effective Hamiltonian as
\begin{equation}
    H_{\mathrm{eff}}(t) = H + \Delta(t)\,,
\end{equation}
where $\Delta(t) = \sum_{j=1}^{\infty}E_{j+1}t^{j}$ and $H$ is the Hamiltonian of interest that we are trying to simulate. We can apply the variation of parameters formula to this expression to write
\begin{equation}\label{eq:var-of-params-1}
    \cP^{1/s}(sT) = e^{-iTH_{\mathrm{eff}}(sT)} = e^{-iHT} + \int_{0}^{T} e^{-i(T-\tau)H} i\Delta(sT) e^{-i\tau H_{\mathrm{eff}}(t)}d\tau.
\end{equation}
This formula can then be iterated to prove the following lemma.
\begin{lemma}\label{lem:time-evol-error-series}
    Let $\cP(t)$ be an $\Upsilon$-staged $p$-th order product formula of symmetry class $\sigma$ and subsume the conditions in Lemma \ref{lem:Effective_Hamiltonian_Error_Series}. 
    Then, the approximation error (operator) of $\cP^{1/s}(sT)$ with inverse Trotter step number $s$ compared with an exact evolution of $H$
    may be expressed as
    \begin{equation} \label{Eq:Approx_Time_Evolution_High_Order}
       \cP^{1/s}(sT) - e^{-iHT}= \sum_{j\in \sym\mathbb{Z}_+\geq p} s^j \tilde{E}_{j+1}(T)\,,
    \end{equation}
    where, $\tilde{E}_{j+1}(T)$ are operators whose operator norm is bounded as 
    \begin{align}
        \norm{\tilde{E}_{j+1}(T)} 
        &\leq \sum_{l=1}^{\lfloor j/p \rfloor  } \frac{(a_\mathrm{max}\Upsilon T \lcomm)^{j+l}}{l!}\,.
    \end{align}
    where we have denoted
    \begin{equation}
        \lcomm:= \sup _{\substack{j \in \sigma \mathbb{Z}_{+} \geq p, \\ 1 \leq l \leq \lfloor j/p \rfloor}}\lambda_{j, l}\,, \qquad
        \lambda_{j, l}:=\bigg(\sum_{\substack{j_1 \ldots j_1 \in \sigma \mathbb{Z}_{+} \geq p \\ j_1+\cdots+j_l=j}} \prod_{\kappa=1}^l \frac{\acomm^{\left(j_\kappa+1\right)}}{\left(j_\kappa+1\right)^2}\bigg)^{1 /(j+l)}\,.
    \end{equation}
\end{lemma}
We provide the full proof in \cref{sec:proof_approx_series}.

\subsection{Extrapolation error}\label{sec:extrapolation-error}

In this subsection we use the error series on the time evolution operator to characterize the error incurred for any extrapolation scheme.


\begin{lemma}[Extrapolation error for time signal]\label{lem:generic-Richardson-error}
    Let $\PC$ be an $\Upsilon$-staged $p$-th order product formula of symmetry class $\sigma$. Take any $m$-term extrapolation schedule $R=\{(b_k, q_k)\}_{k=1}^m$ of inverse Trotter step sizes $s_k=s/q_k$, which cancels the powers $\{\sigma, {2 \sigma}, \ldots, {(m-1)\sigma} \}$, where $\sigma m \geq p$. For a target evolution time $T$ denote
    \begin{equation}
        \cP^{(R)}_{p,m}(T):= \sum_{k=1}^m b_k \cP^{1/s_k}(s_kT)\,,
    \end{equation}
    as the extrapolation of product formulae according to this schedule. We consider a Hamiltonian which assumes the conditions of Lemma \ref{lem:time-evol-error-series}. Choose the largest inverse Trotter step number $s_1$ such that  $s_1 (a_{\max } \Upsilon \lcomm T)^{1+1/p}\leq 1/2$. 
    Then, the error in the extrapolation of a Trotterized time signal, as compared to an exact time signal, satisfies
    \begin{equation}
    \left\| \cP^{(R)}_{p,m}(T) - e^{iHT} \right\| \leq 4 \|\vec{b}\|_1 s_1^{\sigma m} \cdot \max\big\{1,  \left(a_{\max } \Upsilon \lcomm T \right)^{\sigma m(1+1/p)} \big\}\,.
    \end{equation}
    for any quantum state $\rho$, where $\|\vec{b}\|_1=\sum_k\left|b_k\right|$, and we denote
    \begin{equation}\label{eq:lcomm}
        \lcomm:= \sup _{\substack{j \in \sigma \mathbb{Z}_{+} \geq \sigma m \\ 1 \leq l \leq \lfloor j/p \rfloor}}\lambda_{j, l}\,, \qquad
        \lambda_{j, l}:=\bigg(\sum_{\substack{j_1 \ldots j_1 \in \sigma \mathbb{Z}_{+} \geq p \\ j_1+\cdots+j_l=j}} \prod_{\kappa=1}^l \frac{\alpha_{\text {comm }}^{\left(j_\kappa+1\right)}}{\left(j_\kappa+1\right)^2}\bigg)^{1 /(j+l)}\,.
    \end{equation}
\end{lemma}

\begin{proof}
From Lemma \ref{lem:time-evol-error-series}, recall we have the error series in $s$ as
\begin{equation} 
   \cP^{1/s}(sT) - e^{-iHT}= \sum_{j\in \sym\mathbb{Z}_+\geq p} s^j  \tilde{E}_{j+1}(T) \,.
\end{equation}
Extrapolation removes all terms of size up to and including $O(s^{\sigma(m-1)})$ in the series. We denote extrapolation error (operator) as 
\begin{equation}
    \cP^{(R)}_{p,m}(T) - e^{-iHT} = \sum_{k=1}^m b_k   R_{\sigma(m-1)}(T, s_k)\,.
\end{equation}
where $R_{\sigma(m-1)}(T, s_k)$ denotes the Taylor remainder of degree $\sigma(m-1)$, taking the form
\begin{equation}\label{eq:R-exp}
    R_{\sigma(m-1)}(T, s_k):=\sum_{\substack{j \in \sigma \mathbb{Z}_{+} \\ j \geq \sigma m}} s_k^j \tilde{E}_{j+1}(T) \,,
\end{equation}
for each inverse Trotter step $s_k$. The Taylor remainder for inverse Trotter step $s_k$ has size 
\begin{align}\label{eq:remainder-bound}
    \left\|R_{\sigma(m-1)}(T, s_k)\right\| &\leq  \sum_{\substack{j \in \sigma \mathbb{Z}_{+} \\ j \geq \sigma m}} s_k^j\left\|\tilde{E}_{j+1}(T)\right\|  \leq \sum_{\substack{j \in \sigma \mathbb{Z}_{+} \\ j \geq \sigma m}} s_k^j \sum_{l=1}^{\lfloor j/p\rfloor} \frac{\left(a_{\max } \Upsilon \lcomm T \right)^{j+l}}{l!} \,,
\end{align}
where we have used Lemma \ref{lem:time-evol-error-series}. Further, note that the inner sum (bound on $\|\tilde{E}_{j+1}(T)\| $) satisfies
\begin{align}
    \sum_{l=1}^{\lfloor j/p\rfloor} \frac{\left(a_{\max } \Upsilon \lcomm T \right)^{j+l}}{l!} &\leq  \left(a_{\max } \Upsilon \lcomm T \right)^j \sum_{l=1}^{\lfloor j/p\rfloor} \frac{\left(a_{\max } \Upsilon \lcomm T \right)^{l}}{l!} \\
    & \leq \left(a_{\max } \Upsilon \lcomm T \right)^j \eta^{j/p} (e-1)\,,
\end{align}
where in the final line we have used the fact that $\sum_{l=1}^{\infty} \frac{1}{l!} \leq (e-1)$, and denoted $\eta = \max\{1, a_{\max } \Upsilon \lcomm T\}$.

We can bound the sum in the remainder in Eq.~\eqref{eq:remainder-bound} using the largest inverse step number $s_1$.  Using this fact and denoting $\vec{b}:=(b_1,..,b_m)$ as the vector of extrapolation coefficients, the size of the error can be bounded as 
\begin{align}\label{eq:richardson-error-intermediate}
    \| \cP^{(R)}_{p,m}(T) - e^{-iHT} \| &\leq \|\vec{b}\|_1 \cdot \left\|R_{\sigma(m-1)}(T, s_1)\right\| \\ 
    &\leq \|\vec{b}\|_1 \cdot (e-1)  \sum_{\substack{j \in \sigma \mathbb{Z}_{+} \\ j \geq \sigma m}} s_1^j \eta^{ j/p} \left(a_{\max } \Upsilon \lcomm T \right)^j \\
    &\leq \|\vec{b}\|_1 \cdot (e-1)  \left(s_1 a_{\max } \Upsilon \lcomm T \right)^{\sigma m} \eta^{ \sigma m/p} \sum_{\substack{j \in \sigma \mathbb{Z}_{+} \\ j \geq 0}} s_1^j \left(a_{\max } \Upsilon \lcomm T \right)^j \eta^{ j/p} \\
    &\leq \|\vec{b}\|_1 \cdot (e-1)  \left(s_1 a_{\max } \Upsilon \lcomm T \right)^{\sigma m} \eta^{ \sigma m/p} \sum_{\substack{j \in \sigma \mathbb{Z}_{+} \\ j \geq 0}}  \left(\frac{1}{2} \right)^j \\
    &\leq 4 \|\vec{b}\|_1 \eta^{\sigma m/p}  \left(s_1 a_{\max } \Upsilon \lcomm T \right)^{\sigma m}\,,
\end{align}
where the penultimate line is true if $s_1 (a_{\max } \Upsilon \lcomm T)^{1+1/p}\leq 1/2$.
\end{proof}

We remark that within the above analysis we have also provided a proof of Lemma \ref{lem:radius-of-convergence}. Namely, we have shown that there is an error series 
\begin{equation} 
   \cP^{1/s}(sT) - e^{-iHT}= \sum_{j\in \sym\mathbb{Z}_+\geq p} s^j  \tilde{E}_{j+1}(T) \,,
\end{equation}
where for late times, $\|\tilde{E}_{j+1}(T)\| \leq (e-1) (a_{\max}\Upsilon\lcomm T)^{j(1+1/p)}$. Thus, if  the inverse Trotter step size $s$ satisfies $s(a_{\max}\Upsilon\lcomm T)^{j(1+1/p)} \leq 1/(e-1)$, the error series is geometrically decaying.

Returning to our above bound on extrapolation error, we can now simply derive a sufficient number of Trotter steps to attain a desired precision $\varepsilon$.

\begin{lemma}[Trotter step number for extrapolated circuits]\label{lem:generic-Richardson-error-2}
Subsume the setting of Lemma \ref{lem:generic-Richardson-error}. In order to attain precision $\| \cP^{(R)}_{p,m}(T) - e^{-iHT} \| \leq \varepsilon \leq 4$, it is sufficient for the costliest product formula to use a number of Trotter steps no more than 
    \begin{equation}\label{eq:generic-Richardson-error-2}
        r_{\max} = q_m \left\lceil \frac{1}{q_1} \max \Big\{ 1,\left(a_{\max } \Upsilon \lcomm T\right)^{1+1/p} \Big\}\left(\frac{4\|\vec{b}\|_1}{\varepsilon}\right)^{\frac{1}{\sigma m}}\right\rceil\,.
    \end{equation}
\end{lemma}

\begin{proof}
    Following on from Lemma \ref{lem:generic-Richardson-error}, we obtain the desired error guarantee if 
    \begin{align}
        &4 \|\vec{b}\|_1 s_1^{\sigma m} \cdot \max\big\{1,  \left(a_{\max } \Upsilon \lcomm T \right)^{\sigma m (1+1/p)} \big\} \leq \varepsilon\,,
    \end{align}
    where $s_1 = s/q_1$ and so it is sufficient to take a base number of Trotter steps 
    \begin{equation}
        r_{\min} = 1/s = \left\lceil \frac{1}{q_1} \max\Big\{1, \left(a_{\max } \Upsilon \lcomm T\right)^{1+1/p} \Big\}\left(\frac{4\|\vec{b}\|_1}{\varepsilon}\right)^{\frac{1}{\sigma m}}\right\rceil
    \end{equation}
    where we have taken a ceiling function so that there are an integer number of steps for all the product formulae in the extrapolation. The lemma is completed by observing that in our extrapolation scheme the largest number of steps taken is $q_m/s$
\end{proof}

In the above we have presumed access to an extrapolation scheme which cancels integer powers for non-symmetric product formulae and even powers for symmetric product formulae. However, we have not used the fact that the error series in Lemma \ref{lem:Effective_Hamiltonian_Error_Series} has lowest order term $O(s^{p})$, meaning that such an extrapolation scheme is canceling error terms that are already trivial for powers below $p$. We will use our above lemmas to write specific asymptotic results using the extrapolation scheme of \cite{low2019Multiproduct} which shares this property. However, we can always find alternative specialized extrapolation schemes which also cancel error terms starting from $O(s^{p})$, which we detail below. This result will be used in our numerical experiments. 

\begin{lemma}[Error for user-specified schedules]\label{lem:specialized-Richardson-error}
    Let $\PC$ be an $\Upsilon$-staged $p$-th order product formula of symmetry class $\sigma$. Take any $m$-term extrapolation schedule $R=\{(b_k, q_k)\}_{k=1}^m$ of inverse Trotter step sizes $s_k=s/q_k$, which cancels the powers $\{{p}, {p +\sigma}, \ldots, {p+(m-2)\sigma}\}$. For a target evolution time $T$ denote
    \begin{equation}
        \cP^{(R)}_{p,m}(T):= \sum_{k=1}^m b_k \cP^{1/s_k}(s_kT)
    \end{equation}
    as the extrapolation of product formulae according to this schedule. We consider a Hamiltonian which assumes the conditions of Lemma \ref{lem:time-evol-error-series}. 
    The error in the extrapolated time signal satisfies
    \begin{equation}
    \left\| \cP^{(R)}_{p,m}(T) - e^{-iHT} \right\| \leq \varepsilon\,,
    \end{equation}
    for any $\varepsilon\leq 4$, by taking maximum number of Trotter steps
    \begin{equation}
        r_{\max} = q_m \left\lceil \frac{1}{q_1} \max \Big\{ 1,\left(a_{\max } \Upsilon \lcomm T\right)^{1+1/p} \Big\}\left(\frac{4\|\vec{b}\|_1}{\varepsilon}\right)^{\frac{1}{\sigma (m-1) + p}}\right\rceil\,.
    \end{equation}
    where we have denoted $\|\vec{b}\|_1=\sum_k\left|b_k\right|$.
\end{lemma}

\begin{proof}
    The Taylor remainder for inverse Trotter step $s_k$ has size 
\begin{align}
    \left\|R_{\sigma(m-1)}(T, s_k)\right\| &\leq  \sum_{\substack{j \in \sigma \mathbb{Z}_{+} \\ j \geq  \sigma(m-1)+p}} s_k^j\left\|\tilde{E}_{j+1}(T)\right\|  \\
    & \leq \sum_{\substack{j \in \sigma \mathbb{Z}_{+} \\ j \geq \sigma(m-1)+p}} s_k^j \sum_{l=1}^{\lfloor j/p\rfloor} \frac{\left(a_{\max } \Upsilon \lcomm T\right)^{j+l}}{l!} \\
    & \leq (e-1) \sum_{\substack{j \in \sigma \mathbb{Z}_{+} \\ j \geq \sigma(m-1)+p}} s_k^j \left(a_{\max } \Upsilon \lcomm T \right)^j \eta^{ j/p} \,. 
\end{align}
where the first two inequalities come from Lemma \ref{lem:time-evol-error-series}, and the final inequality was established in the proof of Lemma \ref{lem:Effective_Hamiltonian_Error_Series}. Using this, the size of the error can be bounded as 
\begin{align}
    \|\cP^{(R)}_{p,m}(T&)- e^{-iHT}\| \leq \\ 
    &\leq \|\vec{b}\|_1 \cdot \left\|R_{\sigma(m-1)}(T, s_1)\right\| \\ 
    &\leq \|\vec{b}\|_1 \cdot (e-1)  \sum_{\substack{j \in \sigma \mathbb{Z}_{+} \\ j \geq \sigma(m-1)+p}} s_1^j \eta^{ j/p} \left(a_{\max } \Upsilon \lcomm T \right)^j \\
    &\leq \|\vec{b}\|_1 \cdot (e-1)  \left(s_1 a_{\max } \Upsilon \lcomm T \right)^{\sigma(m-1)+p} \eta^{ \sigma (m-1)/p + 1} \sum_{\substack{j \in \sigma \mathbb{Z}_{+} \\ j \geq 0}} s_1^j \left(a_{\max } \Upsilon \lcomm T \right)^j \eta^{ j/p} \\
    &\leq \|\vec{b}\|_1 \cdot (e-1)  \left(s_1 a_{\max } \Upsilon \lcomm T \right)^{\sigma(m-1)+p} \eta^{ \sigma (m-1)/p + 1}  \sum_{\substack{j \in \sigma \mathbb{Z}_{+} \\ j \geq 0}}  \left(\frac{1}{2} \right)^j \\
    &\leq 4 \|\vec{b}\|_1 \eta^{ \sigma (m-1)/p + 1}  \left(s_1 a_{\max } \Upsilon \lcomm T \right)^{\sigma (m-1)+p}\,.
\end{align}
where we have denoted $\vec{b}:=(b_1,..,b_m)$ as the vector of extrapolation coefficients, and the penultimate line is true if $s_1 (a_{\max } \Upsilon \lcomm T)^{1+1/p}\leq 1/2$. Our final requirement on $s_1$ will be more stringent. Namely, the error is less than $\varepsilon$ if we take reference Trotter step number equal to
\begin{equation}
    r = 1/s = \left\lceil \frac{1}{q_1} (a_{\max } \Upsilon \lcomm T)^{1+1/p} \left(\frac{4\|\vec{b}\|}{\varepsilon}\right)^{\frac{1}{\sigma(m-1)+p}} \right\rceil\,.
\end{equation}

\end{proof}

\noindent We see that there is a stronger exponent on $(1/\varepsilon)$ compared to Lemma \ref{lem:generic-Richardson-error-2}. Thus, this lemma is advantageous and should be used in place of Lemma \ref{lem:generic-Richardson-error-2} for any satisfying extrapolation scheme. As extrapolation schemes that cancel powers up to arbitrary order can be found rather easily numerically, we use this Lemma as the basis for Theorem \ref{thm:time-signal-generic} when discussing user-defined extrapolation schedules.

\subsection{Well-conditioned extrapolation strategy}
In order to make asymptotic statements independent of extrapolation parameters we now adopt a particular Richardson extrapolation strategy found in \cite{low2019Multiproduct} -- this allows an approximation with error \(\widetilde{O}(s^{\sigma m} )\), using coefficients whose condition number grows only as \(O(\log m)\). This will satisfy the extrapolation properties of Lemmas \ref{lem:generic-Richardson-error} and \ref{lem:generic-Richardson-error-2}, but not Lemma \ref{lem:specialized-Richardson-error}.

\begin{lemma}(LKW Well-conditioned Richardson extrapolation \cite{low2019Multiproduct, watkins2022TimeDependentClockSimulation} (Lemma 5 in \cite{watson2024exponentiallyReducedTrotter}))\label{lemma:well-conditioned-Richardson}
Let $f$  be an even, real-valued function of $s$, and let $p_j$ and $r_j$ be the degree $j$ Taylor polynomial and Taylor remainder, respectively, such that $f(s)=p_j(s)+r_j(s)$. Let   
\begin{align}
    F^{(m)}(s)=\sum_{k=1}^m b_k f(s_k)
\end{align}
be the unique Richardson extrapolation $R=\{b_k, q_k\}$ of $f$ at points $s_1, s_2, \ldots s_m$ given by
\begin{equation}
    s_k=\frac{s}{q_k}\,;\qquad  q_k=q_{\mathrm{scale }}\left\lceil\frac{\sqrt{8}m/\pi}{\sin (\pi(2 k-1) / 8 m)}\right\rceil, \quad k \in\{1, \ldots, m\}\,,
\end{equation}
where $q_{\mathrm{scale}} \in \mathbb{Z}_+$ satisfies $m \leq q_k/q_{\mathrm{scale}}\leq 3m^2$ \cite{low2019Multiproduct}. Let $b_k$ be the entries of the solution vector of the $m \times m$ linear system $V\vec{b}=(1,0,\cdots,0)$ for $V_{jk} = \sigma_k^{2(j-1)}$. Then, the extrapolated function satisfies
\begin{align}
    F^{(m)}(s)=f(0)+\sum_{k=1}^m b_k r_{2 m}(s_k)\,,
\end{align}
and $\|\vec{b}\|_1=O(\log m)$.
\end{lemma}

We can use this extrapolation strategy to write down the asymptotic scaling we use throughout our main results.

\begin{lemma}[Time evolution gate complexity with well-conditioned extrapolation]\label{lem:richardson-well-conditioned}
Suppose that $\cP^{(R)}_{p,m}(T)$ is constructed using the well-conditioned extrapolation strategy of Lemma \ref{lemma:well-conditioned-Richardson}. We attain extrapolation error $\| e^{-iHT} -  \cP^{(R)}_{p,m}(T)\| \leq \varepsilon$  
using $m=O(\log(1/\varepsilon))$  extrapolation circuits, each with at most
$O\!\left(\left(a_{\max } \Upsilon \lcomm T\right)^{(1+1 / p)} \log (1 / \varepsilon)\right)$ Trotter steps.
\end{lemma}

\begin{proof}
    We see that the extrapolation scheme of \cite{low2019Multiproduct} has properties that $q_{\max}/q_{\min} \leq 3m$. We can directly substitute this into Eq.~\eqref{eq:generic-Richardson-error-2} with choice $m=\lceil\log(1/\varepsilon) \rceil$. We thus have
    \begin{equation}
        r_{\max} = O\left(  \log(1/\varepsilon) \left(a_{\max } \Upsilon \lcomm T\right)^{1+1/p}\left(\frac{4\|\vec{b}\|_1}{\varepsilon}\right)^{\frac{1}{\sigma \lceil\log(1/\varepsilon) \rceil}}\right)\,.
    \end{equation}
    As $(1/\varepsilon)^{1/\log(1/\varepsilon)} = O(1)$ and $\|\vec{b}\|_1 = O(\log m)$ we thus have the largest number of Trotter steps satisfying
    \begin{equation}
        r_{\max} =O\!\left(\left(a_{\max } \Upsilon \lcomm T\right)^{(1+1 / p)} \log (1 / \varepsilon)\right)\,.
    \end{equation}
\end{proof}

\section{Randomized algorithm framework for matrix functions}

So far we have seen how to approximate the time evolution operator mathematically as a linear combination of $m$ instances of a $p$-th order product formulae, specifically seeking conditions when $\| e^{-iHT} -  \cP^{(R)}_{p,m}(T)\| \leq \varepsilon_R$. We will hereon refer to $\varepsilon_R$ as the (Richardson) extrapolation error. In this section we will demonstrate how to compile the time evolution operator as part of an algorithm, which will allow us to instantiate Tasks \hyperref[eq:task-1]{1} and \hyperref[eq:task-2]{2} as laid out in the introduction.

\subsection{A sketch for the warm-up problem}

To give intuition for how our algorithms work, we first sketch how to give an algorithm for our warm-up task: approximating $\tr[\rho e^{-iHT}]$ to additive error $\varepsilon$. Later, our warm-up task will be formally instantiable as a special case of our algorithm for Task \hyperref[eq:task-1]{1}.

Lemma \ref{lem:richardson-well-conditioned} allows us to write

\begin{equation}
    \tr[\rho \cP^{(R)}_{p,m}(T)] = \sum_{k=1}^m b_k \tr[\rho \cP^{1/s_k}(s_kT)] \overset{\varepsilon_R}{\approx} \tr[\rho e^{-iHT}]  \,,
\end{equation}
where the approximation guarantee is possible due to Holder's inequality ($\|\rho \|_1 = 1$). Note that $\tr[\rho \cP^{(R)}_{p,m}(T)]$ can be simply framed as a quasiprobability distribution over quantities of the form $\tr[\rho U_k]$, where $U_k$ is an efficiently implementable unitary (product formula). Each $\tr[\rho U_k]$ can be simply obtained via a Hadamard test -- specifically, it is the expectation value of a Hadamard test protocol.

Our algorithm then becomes apparent: sample according to a probability distribution with probability masses proportional to $\{b_k \}_k$, obtaining index $j$; run one instance of a Hadamard test corresponding to the quantity $\tr[\rho \cP^{1/s_k}(s_kT)]$; record the measurement statistic multiplied by the weighted phase $b_j/\|\vec{b}\|_1$; repeat a sufficient number of times to constrain statistical error and take the mean.

As found in Lemma \ref{lem:richardson-well-conditioned}, the worst case gate depth is $O(\Gamma\Upsilon \left(a_{\max } \Upsilon \lcomm T\right)^{(1+1 / p)} \log (1 / \varepsilon_R))$. We can place guarantees on the statistical error using Hoeffding's inequality as our random variable has bounded size $O(1)$ -- namely, it is sufficient to take $O(\|\vec{b}\|^2_1/\varepsilon_s^2)$ samples to obtain additive approximation error $\varepsilon_S$ with arbitrary constant success probability. A simple choice of $\varepsilon_s, \varepsilon_R = O(\varepsilon)$ gives the desired error guarantee. We provide a formal statement of algorithm complexities, and note this comes as a corollary of Theorem \ref{thm:overlap} to come later.

\begin{theorem}[Time signal compilation]
\label{thm:time-signal}
Consider a matrix $H\in \mathbb{C}^{2^n \times 2^n}$ of the form in Eq.~\eqref{eq:matrix}.
Take $R(\varepsilon/2)$ as the Richardson extrapolation scheme of Lemma \ref{lem:richardson-well-conditioned} with extrapolation error $\varepsilon/2$. Take any even $p$-th order product formula $\PC$. Using this choice of $R(\varepsilon/2)$ and $\PC$, Algorithm~\ref{algo:overlap} returns an estimate of \( \operatorname{Tr}[\rho e^{-iHT}] \) to additive error at most $\varepsilon$ with success probability at least \( (1 - \delta) \) using:

\begin{itemize}
    \item \textbf{Gate depth (per sample):}
    \begin{equation}
    C_{\mathrm{gate}} = O\!\left( \Gamma \Upsilon \cdot \left( a_{\max} \Upsilon \lambda_{\mathrm{comm}} \, T \right)^{1 + 1/p} \log(1/\varepsilon) \right),
    \end{equation}

    \item \textbf{Sample complexity:}
    \begin{equation}
    C_{\mathrm{sample}} = O\!\left( \frac{  (\log \log(1/ \varepsilon))^2}{\varepsilon^2} \cdot \log\left( \frac{1}{\delta} \right) \right),
    \end{equation}

    \item \textbf{Classical preprocessing time:}
    \begin{equation}
    C_{\mathrm{pre}} = O\big(\log(1/\varepsilon)\big),
    \end{equation}
\end{itemize}
where each circuit uses $(n+1)$ qubits.
\end{theorem}

Further, rather than using the LKW well-conditioned extrapolation scheme, we can also explicitly write out the resource costs of using a generic extrapolation scheme as in Lemma \ref{lem:generic-Richardson-error-2}.

\begin{theorem}[Time signal compilation -- user defined extrapolation]
\label{thm:time-signal-generic}
Consider a matrix $H\in \mathbb{C}^{2^n \times 2^n}$ of the form in Eq.~\eqref{eq:matrix}.
Let $\PC$ be an $\Upsilon$-staged $p$-th order product formula of symmetry class $\sigma$. Take any $m$-term extrapolation schedule $R=\{(b_k, q_k)\}_{k=1}^m$ of inverse Trotter step sizes $s_k=s/q_k$, which cancels the powers $\{{p}, {p +\sigma}, \ldots, {p+(m-2)\sigma}\}$. Algorithm~\ref{algo:overlap} returns an estimate of \( \operatorname{Tr}[\rho e^{-iHT}] \) to additive error at most $\varepsilon$ with success probability at least \( (1 - \delta) \) using:

\begin{itemize}
    \item \textbf{Gate depth (per sample):}
    \begin{equation}
    C_{\mathrm{gate}} \leq  \Gamma \Upsilon \cdot q_m \left\lceil \frac{1}{q_1} \max \Big\{ 1,\left(a_{\max } \Upsilon \lcomm T\right)^{1+1/p} \Big\}\left(\frac{4\|\vec{b}\|_1}{\varepsilon}\right)^{\frac{1}{\sigma (m-1) + p}} \right\rceil \,,
    \end{equation}

    \item \textbf{Sample complexity:}
    \begin{equation}
    C_{\mathrm{sample}} \leq \frac{  2\|\vec{b}\|_1^2 }{\varepsilon^2} \log\left( \frac{2}{\delta} \right) \,,
    \end{equation}

    \item \textbf{Classical preprocessing time:}
    \begin{equation}
    C_{\mathrm{pre}} = O\big(\log(1/\varepsilon)\big)\,,
    \end{equation}
\end{itemize}
where each circuit uses $(n+1)$ qubits.
\end{theorem}

\subsection{Algorithms via Fourier approximation}\label{sec:fourier-algos}

Our central idea to construct simple algorithms for general matrix functions \( f(H) \) is to approximate our desired function using a truncated discrete Fourier series:
\begin{equation}\label{eq:fourier-series}
    f(H) \approx f_K(H) := \sum_{k=1}^K c_k e^{-i H t_k},
\end{equation}
where the coefficients \( c_k \in \mathbb{C} \) and time parameters \( t_k \in \mathbb{R} \) depend on the desired approximation accuracy, specifically on the Fourier truncation error $\varepsilon_{\mathrm{F}}$, which we denote as any known bound
\begin{equation}\label{eq:fourier-series-approximation} 
    \left\| f(H) - f_F(H) \right\| \leq \varepsilon_{\mathrm{F}}.
\end{equation}
We note the this Fourier series $F(\varepsilon_F)=\{c_k(\varepsilon_F),t_k(\varepsilon_F)\}_{k=1}^K$ can be any "out-of-the-box" scalar Fourier series which approximates the scalar function $f(\cdot)$ to approximation error $\varepsilon_{\mathrm{F}}$ on any domain which subsumes the spectrum of $H$. From hereon we drop explicit dependence on $\varepsilon_F$ in notation unless emphasis is needed.

We will refer often to the maximal time parameter which we denote as 
\begin{equation}
    T := \max_{1 \leq k \leq K} |t_k|\,,
\end{equation}
and the vector of Fourier coefficients as \(\vec{c} = (c_1, \ldots, c_K)\) whose \( \ell_1 \)-norm which we denote as
\begin{equation}
c := \sum_{k=1}^K |c_k|\,.
\end{equation}

As we found in the previous section, each exponential \( e^{-i H t_k} \) can then be approximated using an \( m \)-term Richardson extrapolation schedule $R=\{b_j, q_j \}_{j=1}^m$ based on a \( p \)-th order product formula \( \PC \):
\begin{equation}
    e^{-i H t_k} \overset{\varepsilon_R}{\approx}  \sum_{j=1}^m b_j \, \PC^{1/s_j}(s_j t_k) =: \cP^{(R)}_{p,m}(t_k), \label{eq:function-quasiprob}
\end{equation}
where \( s_j = s/q_j \) are scaled Trotter step sizes, and the approximation guarantee is possible if the conditions of Lemma \ref{lem:specialized-Richardson-error} or Lemma \ref{lemma:well-conditioned-Richardson} are satisfied. We denote this Fourier--Richardson approximation as
\begin{equation}\label{eq:fourier-richardson-approx}
f_{F,R}(H) := \sum_{k=1}^K \sum_{j=1}^m c_k b_j \, \PC^{1/s_j}(s_j t_k),
\end{equation}
and we can see that $f_{F,R}(H) \overset{\varepsilon_F + c\cdot \varepsilon_R}{\approx} f(H)$. 

\begin{lemma}[Fourier-Richardson approximation]\label{lem:fourier-richardson}
    Take the Fourier series approximation $F$ in Eq.~\eqref{eq:fourier-series-approximation} and take the conditions of Lemma \ref{lem:specialized-Richardson-error} or Lemma \ref{lemma:well-conditioned-Richardson} with $\varepsilon = \varepsilon_R$ and $T = \max_k t_k$. We have 
    \begin{equation}
        \left\|f_{F,R}(H) - f(H) \right\| \leq \varepsilon_F + c\cdot \varepsilon_R\,.
    \end{equation}
\end{lemma}

\begin{proof}
    Using the triangle inequality, we have 
    \begin{align}
        \left\|f_{F,R}(H) - f(H) \right\| &\leq \left\|f_{F,R}(H) - f_F(H) \right\| + \left\|f(H) - f_F(H) \right\| \\
        &\leq \sum_{k=1}^K |c_k| \left\| e^{-iAt_k} - \cP^{(R)}_{p,m}(T) \right\| + \varepsilon_F\,,
    \end{align}
    where the second inequality follows by a triangle inequality on the definitions of $f_{F,R}(H)$ and $f_{F}(H)$, and the proof follows by directly imposing the conditions of Lemma \ref{lem:specialized-Richardson-error} or Lemma \ref{lemma:well-conditioned-Richardson}.
\end{proof}

For convenience, we denote $S := \sum_{k=1}^K \sum_{j=1}^m |c_k b_j|$.
This quantity bounds the negativity of the quasiprobability distribution in Eq.~\eqref{eq:fourier-richardson-approx} and will characterize the sample overhead of our algorithms.
When $R$ is the Low-Kliuchnikov-Wiebe extrapolation procedure of Lemma \ref{lemma:well-conditioned-Richardson} this is simply bounded as
\begin{equation}\label{eq:normalization-bound}
S = \left( \sum_{k=1}^K |c_k| \right) \left( \sum_{j=1}^m |b_j| \right) = O(c \log m)\,.
\end{equation}

The rest of this section will discuss how to instantiate $f_{F,R}(H)$ in each of the contexts of Tasks \hyperref[eq:task-1]{1} and \hyperref[eq:task-2]{2}, while controlling the approximation error to our true desired quantity $f(H)$.

\subsection{Circuit primitives}

As a preamble to discussing our full algorithms, we prove a basic result about the circuits we will use. All of our algorithms use a Hadamard test or a generalization thereof, which we draw in Figure \ref{fig:circuits}.

\begin{figure}[h]
  \centering
    \includegraphics[width=\columnwidth]{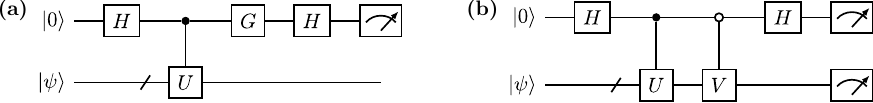}
  \caption{\textbf{(a): Hadamard test circuit}. Measuring the expectation value of $Z$ on the first register returns \textrm{Re}($\bra{\psi}U\ket{\psi}$) and \textrm{Im}($\bra{\psi}U\ket{\psi})$ for choices of $G = I$ and $G = S^{\dag}:=\ketbra{0}-i\ketbra{1}$ respectively. \textbf{(b): Generalized Hadamard test circuit}. Applying controlled-$U$ and anticontrolled-$V$, followed by measurement corresponding to the observable $Z \otimes O$ returns $\frac{1}{2}\left(\bra{\psi}U^{\dag}OV\ket{\psi}+ \bra{\psi}V^{\dag}OU\ket{\psi} \right)$ in expectation.\label{fig:circuits}}
\end{figure}

The output of a Hadamard test is a standard result. For any $n$-qubit unitary $U$, taking the $n+1$ qubit circuit in Figure \ref{fig:circuits} and measuring the expectation value of $Z$ on the first register returns \textrm{Re}($\bra{\psi}U\ket{\psi}$) when the gate $G$ is chosen to be the identity, and \textrm{Im}($\bra{\psi}U\ket{\psi})$ when $G = S^{\dag}:=\ketbra{0}-i\ketbra{1}$ is a phase gate. Here we prove a similar result for the generalized Hadamard test.

\begin{lemma}[Expectation value of generalized Hadamard test]
\label{lemma:hadamard-test-2}
Let \( \rho \) be a quantum state, and let \( U, V\) be unitaries and $O$ an observable. Then, the generalized Hadamard test (see Figure \hyperref[fig:circuits]{1(b)}) that
\begin{itemize}
    \item Initializes the first control qubit in \( \ket{+} \) with a Hadamard gate,
    \item Applies a controlled-\( U \),
    \item Applies an anti-controlled-\(V \) (gate implemented when control is in state \( \ket{0} \)),
    \item Measures corresponding to the observable \( X \otimes O \) 
\end{itemize}
returns the quantity $\frac{1}{2} \left( \mathrm{Tr}[O V \rho U^\dagger] + \mathrm{Tr}[O U \rho V^\dagger] \right)$
in expectation value.
\end{lemma}

\begin{proof}
We trace the circuit after each gate to verify its outcome:
\begin{align}
    \ketbra{+} \otimes \rho 
    &= \frac{1}{2} \left( \ketbra{0} \otimes \rho + \ketbra{1} \otimes \rho + \ket{0}\!\!\bra{1} \otimes \rho + \ket{1}\!\!\bra{0} \otimes \rho \right) \tag{initial state} \\
    &\mapsto \frac{1}{2} \left( 
        \ketbra{0} \otimes \rho + 
        \ket{1}\bra{1} \otimes U \rho U^\dagger +
        \ket{0}\!\!\bra{1} \otimes \rho U^\dagger + 
        \ket{1}\!\!\bra{0} \otimes U \rho
    \right) \tag{controlled-$U$} \\
    &\mapsto \frac{1}{2} \left( 
        \ketbra{0} \otimes V \rho V^\dagger + 
        \ket{1}\bra{1} \otimes U \rho U^\dagger +
        \ket{0}\!\!\bra{1} \otimes V \rho U^\dagger + 
        \ket{1}\!\!\bra{0} \otimes U \rho V^\dagger
    \right) \tag{anti controlled-$V$} \\
    &\mapsto \mathrm{Tr}\left[
        (X \otimes O) \cdot \frac{1}{2} \left( 
            \ket{0}\!\!\bra{1} \otimes V \rho U^\dagger + 
            \ket{1}\!\!\bra{0} \otimes U \rho V^\dagger 
        \right)
    \right] \tag{expectation value for $X \otimes O$} \\
    &= \frac{1}{2} \left( \mathrm{Tr}[O V \rho U^\dagger] + \mathrm{Tr}[O U \rho V^\dagger] \right)\,.  
\end{align}
\end{proof}

\subsection{Estimating $\tr[\rho f(H)]$}

We now present our full algorithm for Task (\hyperref[eq:task-1]{1(i)}) for estimating quantities of the form $\tr[\rho U f(H)]$.  

\begin{algorithm}[h]
\caption{Randomized algorithm for estimating $\operatorname{Tr}[\rho U f(H)]$}
\label{algo:overlap}
\begin{algorithmic}[1]
\Require Fourier series $F=\{c_k, t_k\}_{k=1}^K$; Richardson extrapolation schedule $R=\{b_j,q_j\}_{j=1}^m$; base Trotter step number $1/s$; input state $\rho$ and unitary $U$ with efficient preparations; matrix with decomposition $H=\sum_{\ell} H_{\ell}$ and product formula $\PC$.
\State \textbf{Classical preprocessing:} Compute the probability distribution:
\[
\left\{ \frac{ \left| c_k b_j \right| }{ S } \right\}_{k,j} \quad \text{for } k = 1,\ldots,K\,;\ j = 1,\ldots,m\,;\quad \text{with } S = \|\vec{c}\|_1 \|\vec{b}\|_1\,.
\]
\For{$i = 1$ to $C_{\mathrm{sample}}$}:
    \State Sample an index pair $(k', j')$ from the probability distribution
    \State Take $s_{j'} \leftarrow s/q_{j'}$ 
    \State Prepare Hadamard test circuits corresponding to 
    \[\operatorname{Re} \left[ \operatorname{Tr}\left( \rho U \PC^{1/s_{j'}}(s_{j'} t_{k'}) \right) \right] \quad \text{and} \quad  \operatorname{Im} \left[ \operatorname{Tr}\left( \rho U \PC^{1/s_{j'}}(s_{j'} t_{k'}) \right) \right]\,,\]
    \hspace{1.25em} collect one measurement outcome (single-shot statistic) from each circuit $(X^{(k',j')}_{\text{Re}}, X^{(k',j')}_{\text{Im}})$
    \State Take $\mu_{(k',j')} \leftarrow S \cdot \operatorname{sign}(c_{k'} b_{j'}) \cdot (X^{(k',j')}_{\text{Re}} + iX^{(k',j')}_{\text{Im}})$ and store value
\EndFor
\State \textbf{Return:} $\widehat{\mu} \leftarrow$ the sample mean over all $C_{\mathrm{sample}}$ values
\end{algorithmic}
\end{algorithm}

We first show that \Cref{algo:overlap} is approximately correct. Specifically, we show the expectation value of the estimator over the probability distribution defined in the algorithm exactly equals the Fourier-Richardson approximation to our desired quantity $\operatorname{Tr}[\rho U f_{K,R}(H)]$.

\begin{lemma}[Correctness for overlap estimator]
\label{lem:overlap-correctness}
Let $(k,j)$ be sampled with probability \( \Pr(k,j) = \frac{|c_k b_j|}{S} \) for some Fourier series $F=\{c_k, t_k\}_{k=1}^K$ and Richardson extrapolation schedule $R=\{b_j,q_j\}_{j=1}^m$, and let
\begin{equation}
    \mu_{(k,j)} \leftarrow S \cdot \operatorname{sign}(c_{k} b_{j}) \cdot (X^{(k,j)}_{\text{Re}} + iX^{(k,j)}_{\text{Im}})
\end{equation}
be the random variable that Algorithm \ref{algo:overlap} instantiates which generates one outcome each iteration of the for loop, where $\operatorname{sign}(c_{k} b_{j}) = c_{k} b_{j}/|c_{k} b_{j}|$ denotes the complex phase.
Then, \( \mu_{(k,j)} \) is an unbiased estimator of \( \operatorname{Tr}[\rho f_{F,R}(H)] \) with $f_{F,R}(H)$ as defined in Eq.~\eqref{eq:fourier-richardson-approx}, i.e.,
\begin{equation}
    \mathbb{E}_{k,j,C}[\mu_{(k,j)}] = \operatorname{Tr}[\rho U f_{F,R}(H)]\,,
\end{equation}
where $\mathbb{E}_{k,j,C}$ denotes an expectation value over indices $k,j$, and over quantum circuit outputs.
\end{lemma}

\begin{proof}
We first prove a fact that we will use throughout this section, which is that Algorithm \ref{algo:overlap} instantiates a random operator whose expectation value is $f_{F,R}(H)$. This is simply seen by noting that
\begin{align}
    \mathbb{E}_{k,j}\big[S \cdot \operatorname{sign}(c_{k} b_{j}) \cdot U \PC^{1/s_{j}}(s_{j} t_{k})\big] &= \sum_{k,j} \frac{|c_k b_j|}{S} \cdot S \cdot \operatorname{sign}(c_{k} b_{j}) \cdot U \PC^{1/s_{j'}}(s_{j'} t_{k'}) \\
    & = \sum_{k,j} {c_k b_j} U \PC^{1/s_{j}}(s_{j} t_{k})
\end{align}
which is exactly $U f_{F,R}(H)$ as defined in Eq.~\eqref{eq:fourier-richardson-approx}. Further, we note that the expectation over outputs of Hadamard tests yields
\begin{equation}
    \mathbb{E}_{C}\big[ X^{(k,j)}_{\text{Re}} + iX^{(k,j)}_{\text{Im}} \big] = \tr[\rho U \PC^{1/s_{j}}(s_{j} t_{k})]\,.
\end{equation}
The proof then follows due to linearity of expectation values.
\end{proof}

\begin{theorem}[Overlap estimation]
\label{thm:overlap}
Consider a matrix $H\in \mathbb{C}^{2^n \times 2^n}$ of the form in Eq.~\eqref{eq:matrix}.
Take any Fourier series $F(\varepsilon/3)=\{c_k(\varepsilon/3), t_k(\varepsilon/3)\}_{k=1}^K$ approximation of function $f$ with Fourier approximation error $\varepsilon/3$ as defined in Eq.~\eqref{eq:fourier-series-approximation}, and take $R\big(\varepsilon/(3c(\varepsilon/3))\big)$ as the Richardson extrapolation scheme of Lemma \ref{lem:richardson-well-conditioned} with extrapolation error $\varepsilon/(3c(\varepsilon/3))$. Take any even $p$-th order product formula $\PC$. Using this choice of $F(\varepsilon/3)$, $R\big(\varepsilon/(3c(\varepsilon/3))\big)$ and $\PC$, Algorithm~\ref{algo:overlap} returns an estimate of \( \operatorname{Tr}[\rho U f(H)] \) to additive error at most $\varepsilon$ with success probability at least \( (1 - \delta) \) using:

\begin{itemize}
    \item \textbf{Gate depth (per sample):}
    \begin{equation}
    C_{\mathrm{gate}} = O\!\left( \Gamma \Upsilon \cdot \log\left( \frac{c(\varepsilon/3)}{\varepsilon} \right) \cdot \left( a_{\max} \Upsilon \lambda_{\mathrm{comm}} \, T(\varepsilon/3) \right)^{1 + \frac{1}{p}} \right),
    \end{equation}

    \item \textbf{Sample complexity:}
    \begin{equation}
    C_{\mathrm{sample}} = O\!\left( \frac{ c(\varepsilon/3)^2 \cdot \big(\log \log\big(\frac{c(\varepsilon/3)}{\varepsilon}\big)\big)^2}{\varepsilon^2} \cdot \log\left( \frac{1}{\delta} \right) \right),
    \end{equation}

    \item \textbf{Classical preprocessing time:}
    \begin{equation}
    C_{\mathrm{pre}} = O\left(K(\varepsilon/3) \cdot \log\bigg(\frac{c(\varepsilon/3)}{\varepsilon}\bigg)\right),
    \end{equation}
\end{itemize}
where \( T(\varepsilon/3) := \max_k |t_k(\varepsilon/3)| \), \( c(\varepsilon/3) = \sum_k |c_k(\varepsilon/3)| \), and each circuit uses $(n+1)$ qubits.
\end{theorem}

\begin{proof}
Let \( \widehat{\mu}_{M} = \frac{1}{M} \sum_{\ell=1}^M X^{(\ell)} \) be the empirical mean over $M$ samples of the unbiased estimator defined in Lemma~\ref{lem:overlap-correctness}. We can bound the total error into a contribution of two pieces using the triangle inequality:
\begin{equation}
\left| \operatorname{Tr}[\rho U f(H)] - \widehat{\mu}_M \right| \leq \left| \operatorname{Tr}[\rho U f_{K,R}(H)] - \operatorname{Tr}[\rho U f(H)] \right| + \left| \widehat{\mu}_M -  \operatorname{Tr}[\rho U f_{K,R}(H)] \right| \,,
\end{equation}
where the first part characterizes the Fourier-Richardson approximation, and the second part characterizes the statistical error. In Lemma \ref{lem:fourier-richardson} we bounded the first term by $\varepsilon_F + c(\varepsilon_F) \cdot \varepsilon_R$, where $\varepsilon_F$ is the Fourier approximation error controlled completely by choice of Fourier parameters, and $\varepsilon_R$ is the Richardson extrapolation error. We can thus bound the first term by $2\varepsilon/3$ by setting the Fourier approximation error to $\varepsilon_F =\varepsilon/3$ and the Richardson approximation error to $\varepsilon_R = \varepsilon/(3c(\varepsilon/3))$. We use the extrapolation schedule of Lemma \ref{lem:richardson-well-conditioned}, which requires product formulae of worst-case Trotter step number $O\!\left( \log\left( \frac{c(\varepsilon/3)}{\varepsilon} \right) \cdot \left( a_{\max} \Upsilon \lambda_{\mathrm{comm}} \, T(\varepsilon/3) \right)^{1 + \frac{1}{p}} \right)$, from which our stated worst-case gate depth follows.

Now we derive the sample complexity. Each sample outcome recorded in the algorithm is a complex number with real and imaginary parts in $[-S,S]$. By Hoeffding's inequality (and a union bound to account for real and imaginary parts), for any \( \varepsilon_S > 0 \),
\begin{equation}
\mathbb{P}\left[ \left| \widehat{\mu}_M - \mathbb{E}[\widehat{\mu}_M] \right| \geq \varepsilon_s \right] 
\leq 4 \exp\left( - \frac{M \varepsilon_s^2}{4 S^2} \right)\,,
\end{equation}
and we recall that Lemma \ref{lem:overlap-correctness} demonstrates that $\E[\hat{\mu}_M]= \operatorname{Tr}[\rho U f_{K,R}(H)]$. To ensure that this failure probability is at most \( \delta \), it suffices to choose any
\begin{equation}
M \geq \frac{4S^2}{\varepsilon_s^2} \cdot \log\left( \frac{4}{\delta} \right).
\end{equation}
To obtain the stated sample complexity we set the statistical error to $\varepsilon_s = \varepsilon/3$, recall (see equation~\eqref{eq:normalization-bound}) that $S = \sum_{k,j} |c_k b_j| = O(c \log m)$ and take the fact that $ m = O(\log(1/\varepsilon_R))$ is sufficient to maintain the extrapolation error using Lemma \ref{lem:richardson-well-conditioned}. This guarantees that the total error is bounded by $2\varepsilon/3 + \varepsilon/3 = \varepsilon$ as required.

Finally, we see the classical complexity of forming the probability distribution $\{|c_k b_j / S \}_{k,j}^{K,m}$ is simply $O(K(\varepsilon_F) \cdot m(\varepsilon_R))$.
\end{proof}

\subsection{Estimating $\text{Tr}[ f(H)\rho f(H)^{\dagger} O]$}

Similar to before, we approximate our desired quantity with a quasiprobability distribution over implementable circuits with product formulae, using $f_{F,R}$ in  Eq.~\eqref{eq:fourier-richardson-approx} to approximate $f(H)$. We present the algorithm in Algorithm \ref{algo:observable}. We first show that the algorithm returns $\text{Tr}[ f_{F,R}(H)\rho f_{F,R}(H)^{\dagger} O]$ if given unlimited samples.

\begin{lemma}[Correctness for observable estimator]
\label{lem:observable-correctness}
Let $(k,j)$ be sampled with probability \( \Pr(k,j) = \frac{|c_k b_j|}{S} \) for some Fourier series $F=\{c_k, t_k\}_{k=1}^K$ and Richardson extrapolation schedule $R=\{b_j,q_j\}_{j=1}^m$, and let
\begin{equation}
    \mu_{(k,k',j,j')} \leftarrow S^2 \cdot \operatorname{sign}(c_{k} c^*_{k'} b_{j} b_{j'}) \cdot X^{(k,k',j,j')}
\end{equation}
be the random variable that Algorithm \ref{algo:observable} instantiates which generates one outcome each iteration of the for loop.
Then, \( \mu_{(k,k',j,j')} \) is an unbiased estimator of \( \tr[f_{F,R}(H) \rho f_{F,R}(H)^{\dag} O] \), i.e.,
\begin{equation}
    \mathbb{E}_{k,k',j,j',C}[\mu_{(k,k',j,j')}] = \tr[f_{F,R}(H) \rho f_{F,R}(H)^{\dag} O]\,,
\end{equation}
with $f_{F,R}(H)$ as defined in Eq.~\eqref{eq:fourier-richardson-approx}, where $\mathbb{E}_{k,k',j,j',C}$ denotes an expectation value over indices $k,k',j,j'$, and over quantum circuit outputs.
\end{lemma}

\begin{proof}
In the proof of Lemma \ref{lem:overlap-correctness} we showed that $\mathbb{E}_{k,j}\big[S \cdot \operatorname{sign}(c_{k} b_{j}) \PC^{1/s_{j}}(s_{j} t_{k})\big] = f_{F,R}(H)$.
From this it follows simply that 
\begin{align}
    \mathbb{E}_{k,k',j,j'}\big[S^2 \cdot \operatorname{sign}(c_{k}c_{k'} b_{j} b_{j'})  \PC^{1/s_{j}}(s_{j} t_{k}) \rho \PC^{1/s_{j}}(s_{j} t_{k})^{\dag} \big] &= f_{F,R}(H) \rho f_{F,R}(H)^{\dag}\,.
\end{align}
Further, from Lemma \ref{lemma:hadamard-test-2} we note that the expectation over outputs of the generalized Hadamard test yields
\begin{equation}
    \mathbb{E}_{C}\big[ X^{(k,k',j,j')} \big] = \frac{1}{2} \left( \operatorname{Tr}\left[ \PC^{1/s_{j}}(s_{j} t_{k})  \rho (\PC^{1/s_{j'}}(s_{j'} t_{k'}))^{\dag} O \right] + \operatorname{Tr}\left[ \PC^{1/s_{j'}}(s_{j'} t_{k'})  \rho (\PC^{1/s_{j}}(s_{j} t_{k}))^{\dag} O \right] \right)\,.
\end{equation}
The proof then follows again due to linearity of expectation values.
\end{proof}

\begin{theorem}[Observable estimation]
\label{thm:overlap-with-observable}
Consider a matrix $H\in \mathbb{C}^{2^n \times 2^n}$ of the form in Eq.~\eqref{eq:matrix}. Denote $\tilde{f} = \max(1, \|f(H)\|)$.
Take any Fourier series  $F(\varepsilon_F)=\{c_k(\varepsilon_F), t_k(\varepsilon_F)\}_{k=1}^K$ approximation of function $f$ (defined in Eq.~\eqref{eq:fourier-series-approximation}) with approximation error $\varepsilon_F = \varepsilon/9\tilde{f}$, and take $R(\varepsilon_R)$ as the Richardson extrapolation scheme of Lemma \ref{lem:richardson-well-conditioned} with extrapolation error $\varepsilon_R = \varepsilon / \big(9\tilde{f}\cdot c(\varepsilon/9\tilde{f})\big)$. Take any order $p$ product formula $\PC$. Using this choice of $F(\varepsilon_F)$, $R(\varepsilon_R)$ and $\PC$, Algorithm~\ref{algo:observable} returns an estimate of \( \operatorname{Tr}[f(H)\rho f(H)^{\dagger}O] \), for $\|O\| \leq  1$, to additive error at most $\varepsilon \leq 1$ with success probability at least \( (1 - \delta) \) using:

\begin{itemize}
    \item \textbf{Gate depth (per sample):}
    \begin{equation}
    C_{\mathrm{gate}} = 
    O\!\left( \Gamma \Upsilon \cdot \log\left( \frac{9\tilde{f}\cdot c(\varepsilon/9\tilde{f})}{\varepsilon} \right) \cdot \left( a_{\max} \Upsilon \lambda_{\mathrm{comm}} \, T(\varepsilon/9\tilde{f}) \right)^{1 + \frac{1}{p}} \right),
    \end{equation}

    \item \textbf{Sample complexity:}
    \begin{equation}
    C_{\mathrm{sample}} = O\!\left( \frac{c(\varepsilon/9\tilde{f})^4 \cdot \Big(\log \log \Big(\frac{9\tilde{f}\cdot c(\varepsilon/9\tilde{f})}{\varepsilon} \Big) \Big)^4}{\varepsilon^2} \cdot \log\left(\frac{1}{\delta}\right) \right).
    \end{equation}

    \item \textbf{Classical preprocessing time:}
    \begin{equation}
    C_{\mathrm{pre}} = O\left( \Big(K(\varepsilon/9\tilde{f}) \cdot \log \Big(\frac{9\tilde{f}\cdot c(\varepsilon/9\tilde{f})}{\varepsilon}\Big) \Big)^2 \right),
    \end{equation}
\end{itemize}
where \( T(\cdot) := \max_k |t_k(\cdot)| \), \( c(\cdot) = \sum_k |c_k(\cdot)| \), and each circuit uses $(n+1)$ qubits.
\end{theorem}

\begin{proof}
Let \( \widehat{\mu}_M := \frac{1}{M} \sum_{\ell=1}^M \widehat{X}^{(\ell)} \) denote the empirical mean of the unbiased estimator from Algorithm~\ref{algo:observable} over $M$ samples. Similar to the proof of Theorem \ref{thm:overlap}, we can start by bounding the total error by two contributions (Fourier-Richardson approximation, and statistical):

\begin{align}
\left| \operatorname{Tr}[f(H)\rho f(H)^{\dagger} O] - \widehat{\mu}_M \right| &\leq 
\left| \operatorname{Tr}[f(H)\rho f(H)^{\dagger}O] - \operatorname{Tr}[f_{F,R}(H)\rho f_{F,R}(H)^{\dagger}O] \right| \nonumber \\
&\quad + 
\left| \operatorname{Tr}[f_{K,R}(H)\rho f_{K,R}(H)^{\dagger}O] - \widehat{\mu}_M \right|.
\end{align}

Let us investigate the first term, denoting it as $\Delta_{F,R}$. We have 

\begin{align}
    \Delta_{F,R} &\leq \left| \operatorname{Tr}[f(H)\rho (f(H)-f_{F,R}(H))^{\dagger}O] \right | + \left| \operatorname{Tr}[(f_{F,R}(H)-f(H))\rho f_{F,R}(H)^{\dagger}O] \right| \\
    &\leq \|f(H)\| \cdot \|f(H)-f_{F,R}(H) \| + \|f_{F,R}(H)\| \cdot \|f(H)-f_{F,R}(H) \| \\
    &\leq \tilde{f} (\varepsilon_F + c\cdot \varepsilon_R) + \tilde{f} (\varepsilon_F + c\cdot \varepsilon_R) + (\varepsilon_F + c\cdot \varepsilon_R)^2 \\
    &\leq 3\tilde{f}(\varepsilon_F + c\cdot \varepsilon_R)\,,
\end{align}
where the first inequality is a triangle inequality, the second inequality comes from repeated uses of H\"older's inequality and using the fact that $\|\rho\|_1 = 1$ and $\|O\|\leq 1$, the third inequality is due to Lemma \ref{lem:fourier-richardson}, and the final inequality is true if $\varepsilon_F + c\cdot \varepsilon_R \leq 1$. We can ensure this error is bounded by $2\varepsilon/3$ (and thus also subsume the previous condition in the fact that $\varepsilon\leq 1$) if we set the Fourier series error as $\varepsilon_F = \varepsilon/(9\tilde{f})$ and the Richardson extrapolation error as $\varepsilon_R = \varepsilon/(9\tilde{f}\cdot c(\varepsilon_F))$. The stated gate complexity comes by recalling that a sufficient maximum Trotter step number specified by Lemma \ref{lem:richardson-well-conditioned} is $O\!\left( \log\left( 1/\varepsilon_R \right) \cdot \left( a_{\max} \Upsilon \lambda_{\mathrm{comm}} \, T(\varepsilon_F) \right)^{1 + \frac{1}{p}} \right)$.

The statistical error can be upper bounded by $\varepsilon_s$ by choice of 

\begin{equation}
    M \geq \frac{2S^4}{\varepsilon_s^2} \cdot \log\left( \frac{2}{\delta} \right)\,,
\end{equation}
where similarly to our proof of overlap estimation have used Hoeffding's inequality and the fact that the random variable has magnitude no larger than $S^2$. As in the proof of Theorem \ref{thm:overlap} we use the fact that $S = O(c(\varepsilon_F)\log\log(1/\varepsilon_R))$ and set the statistical error to $\varepsilon_s=\varepsilon/3$ to obtain our stated sample complexity. This also guarantees that the total error is bounded by $\varepsilon$ as required.

The classical preprocessing cost of forming the probability distribution is $O(K^2 m^2)$ where $m=O(\log(1/\varepsilon_R))$ for the extrapolation schedule of Lemma \ref{lem:richardson-well-conditioned}.
\end{proof}

\begin{algorithm}[h]
\caption{Randomized algorithm for estimating $\operatorname{Tr}[f(H) \rho f(H)^\dag O]$}
\label{algo:observable}
\begin{algorithmic}[1]
\Require Fourier series $F=\{c_k, t_k\}_{k=1}^K$; Richardson extrapolation schedule $R=\{b_j,q_j\}_{j=1}^m$; base Trotter step number $1/s$; input state $\rho$ and observable $O$ with efficient preparations; matrix with decomposition $H=\sum_{\ell} H_{\ell}$ and product formula $\PC$.
\State \textbf{Classical preprocessing:} Compute the probability distribution:
\[
\left\{ \frac{ \left| c_k c_{k'} b_j b_{j'} \right| }{ S^2 } \right\}_{k,j} \quad \text{for } k,k' = 1,\ldots,K;\ j,j' = 1,\ldots,m
\]
\For{$i = 1$ to $C_{\mathrm{sample}}$}:
    \State Sample two index pairs $(k'', j''), (k''', j''')$ from the probability distribution
    \State Take $s_{j''} \leftarrow s/q_{j''}$, $s_{j'''} \leftarrow s/q_{j'''}$ 
    \State Prepare a generalized Hadamard test circuit corresponding to 
    \[ 
    \frac{1}{2} \left( \operatorname{Tr}\left[ \PC^{1/s_{j''}}(s_{j''} t_{k''})  \rho (\PC^{1/s_{j'''}}(s_{j'''} t_{k'''}))^{\dag} O \right] + \operatorname{Tr}\left[ \PC^{1/s_{j'''}}(s_{j'''} t_{k'''})  \rho (\PC^{1/s_{j''}}(s_{j''} t_{k''}))^{\dag} O \right] \right) \,,\]
    \hspace{1.25em} collect one measurement outcome (single-shot statistic)  $X^{(k'',k''',j'',j''')}$
    \State Take $\mu_{(k'',k''',j'',j''')} \leftarrow S^2 \cdot \operatorname{sign}(c_{k''} c_{k'''} b_{j''} b_{j'''}) \cdot X^{(k'',k''',j'',j''')}$ and store value
\EndFor
\State \textbf{Return:} $\widehat{\mu} \leftarrow$ the sample mean over all $C_{\mathrm{sample}}$ values
\end{algorithmic}
\end{algorithm}

\subsection{Estimating distributions}

Our final task is to estimate the measurement distribution corresponding to the (non-normalized) state $f(H)\ket{\psi}$. Specifically, we wish to statistically approximate the vector 
\begin{equation}\label{eq:f-distribution}
    \vec{f}:= \sum_{\vec{z}_{n} \in (0,1)^n} \big| \bra{\vec{z}_{n}} U f(H) \ket{\psi} \big|^2 \ket{\vec{z}_{n}}
\end{equation}
where $U$ is some efficiently implementable unitary. We remark that this is proportional to the probability distribution that one obtains from sampling the state $f(H) \ket{\psi}/\|f(H) \ket{\psi} \|_2$ in a basis defined by $U$, and is exactly the probability distribution for normalized $f(H)$.

In Appendix A6 of \cite{wang2023qubitefflinalg} the following Lemma is essentially given.

\begin{lemma}[Distribution estimation -- adapted from \cite{wang2023qubitefflinalg}]\label{lem:distribution-estimation}
    Suppose an operator $G$ can be decomposed into a linear combination of unitaries as $G=\sum_i g_i U_i$. Let $U$ be an implementable unitary. We have an algorithm to return a vector $\vec{v}$ that, with probability at least $(1-\delta)$, satisfies
    \begin{equation}
        \| \vec{v} - \vec{G} \|_2 \leq \varepsilon 
    \end{equation}
    for $\vec{G}:= \sum_{\vec{z}_{n} \in (0,1)^n} \big| \bra{\vec{z}_{n}} U G \ket{\psi} \big|^2 \ket{\vec{z}_{n}}$, using $O \big( \frac{g^4}{\varepsilon^2} \log\big(\frac{1}{\delta} \big) \big)$ circuit samples, where $g:= \sum_i |g_i|$ and each circuit consists of generalized Hadamard tests of $U U_i$.
\end{lemma}

The algorithm can be simply stated and is a small detour from Algorithm \ref{algo:observable}. First, we classically construct the probability distribution  $p_i := |g_i|/g$. Then, sample two indices $(i,j)$ from this distribution and implement the generalized Hadamard test in Figure \hyperref[fig:circuits]{1(b)} with $U_i$ and $U_j$. Next, perform a measurement on all qubits in the computational basis, obtaining an $n$-qubit string which we denote $(z, \vec{z}_n) \in \{ 0,1\}^{n+1}$. Based on this outcome output the vector $g^2 \operatorname{sign}(g_ig_j) (-1)^{z} \ket{\vec{z}_n}$ where we denote $\operatorname{sign}(g_ig_j)=g_ig_j/|g_ig_j|$. Repeat the sampling procedure and take the sample mean of all outcomes.

Note that in the previous section we already found a operator $G$ in this form that approximates $f(H)$ to additive error in operator norm. This means that the probability distributions corresponding to $f(H)\ket{\psi}$ and $G\ket{\psi}$ are separated by essentially the same error in total variation distance. This allows us to write the following result calling Algorithm \ref{algo:distribution}, and we provide the formal proof below.

\begin{algorithm}[ht]
\caption{Randomized algorithm for estimating distributions}
\label{algo:distribution}
\begin{algorithmic}[1]
\Require Fourier series $F=\{c_k, t_k\}_{k=1}^K$; Richardson extrapolation schedule $R=\{b_j,q_j\}_{j=1}^m$; base Trotter step number $1/s$; input state $\ketbra{\psi}$ and unitary $U$ with efficient preparations; matrix with decomposition $H=\sum_{\ell} H_{\ell}$ and product formula $\PC$.
\State \textbf{Classical preprocessing:} Compute the probability distribution:
\[\left\{ \frac{ \left| c_k c_{k'} b_j b_{j'} \right| }{ S^2 } \right\}_{k,k',j,j'} \quad \text{for } k,k' = 1,\ldots,K;\ j,j' = 1,\ldots,m\,.\]
\For{$i = 1$ to $C_{\mathrm{sample}}$}:
    \State Sample two index pairs $(k'', j''), (k''', j''')$ from the probability distribution.
    \State Take $s_{j''} \leftarrow s/q_{j''}$, $s_{j'''} \leftarrow s/q_{j'''}$ .
    \State Prepare a generalized Hadamard test circuit corresponding to unitaries 
    \[ 
    U \PC^{1/s_{j''}}(s_{j''} t_{k''})\,,\quad U \PC^{1/s_{j'''}}(s_{j'''} t_{k'''}) \,,\]
    \hspace{1.25em} and  measure all qubits to obtain  $(z_0,\vec{z}_n) \in \{0,1\}^{n+1}$.
    \State Take $\vec{f}_{(k'',k''',j'',j''')} \leftarrow S^2 \cdot \operatorname{sign}(c_{k''} c_{k'''} b_{j''} b_{j'''}) \cdot (-1)^{z_0} \cdot \vec{z}_n$.
\EndFor
\State \textbf{Return:} $\widehat{\mu} \leftarrow$ the sample mean over all $C_{\mathrm{sample}}$ values.
\end{algorithmic}
\end{algorithm}

\begin{theorem}[Distribution estimation]
\label{thm:distribution-recovery}
Consider a matrix $H\in \mathbb{C}^{2^n \times 2^n}$ of the form in Eq.~\eqref{eq:matrix}. Denote $\hat{f} =: \max(1,\|f(H)\ket{\psi}\|_2)$. Take any Fourier series  $F(\varepsilon_F)=\{c_k(\varepsilon_F), t_k(\varepsilon_F)\}_{k=1}^K$ approximation of function $f$ (defined in Eq.~\eqref{eq:fourier-series-approximation}) with approximation error $\varepsilon_F = \varepsilon/9\hat{f}$, and take $R(\varepsilon_R)$ as the Richardson extrapolation scheme of Lemma \ref{lem:richardson-well-conditioned} with extrapolation error $\varepsilon_R = \varepsilon / \big(9\hat{f}\cdot c(\varepsilon/9\hat{f})\big)$. Take any $p$-th order product formula $\PC$. Using this choice of $F(\varepsilon_F)$, $R(\varepsilon_R)$ and $\PC$, Algorithm~\ref{algo:distribution} returns an estimate of \( \vec{f} \) (defined in Eq.~\eqref{eq:f-distribution}) to additive error at most $\varepsilon \leq 1$ with success probability at least \( (1 - \delta) \) using:

\begin{itemize}
    \item \textbf{Gate depth (per sample):}
    \begin{equation}
    C_{\mathrm{gate}} = 
    O\!\left( \Gamma \Upsilon \cdot \log\left( \frac{9\hat{f}\cdot c(\varepsilon/9\hat{f})}{\varepsilon} \right) \cdot \left( a_{\max} \Upsilon \lambda_{\mathrm{comm}} \, T(\varepsilon/9\hat{f}) \right)^{1 + \frac{1}{p}} \right),
    \end{equation}

    \item \textbf{Sample complexity:}
    \begin{equation}
    C_{\mathrm{sample}} = O\!\left( \frac{c(\varepsilon/9\hat{f})^4 \cdot \Big(\log \log \Big(\frac{9\hat{f}\cdot c(\varepsilon/9\hat{f})}{\varepsilon}\Big)\Big)^4}{\varepsilon^2} \cdot \log\left(\frac{1}{\delta}\right) \right),
    \end{equation}

    \item \textbf{Classical preprocessing time:}
    \begin{equation}
    C_{\mathrm{pre}} = O(\Big(K(\varepsilon/9\hat{f})\cdot \log\Big(\frac{9\hat{f}\cdot c(\varepsilon/9\hat{f})}{\varepsilon}\Big)\Big)^2)\,,
    \end{equation}
\end{itemize}
where \( T(\varepsilon_F) := \max_k |t_k(\varepsilon_F)| \) is the largest time parameter in the Fourier approximation with error \( \varepsilon/9\hat{f} \), and \( c(\varepsilon_F) := \sum_k |c_k(\varepsilon_F) | \) is the \( \ell_1 \)-norm of the Fourier coefficients.
\end{theorem}

\begin{proof}
    Recall as before from Lemma \ref{lem:fourier-richardson} we have an operator $f_{F,R}(H)$ which satisfies $\|f_{F,R}(H) - f(H)\| \leq \varepsilon_F + c(\varepsilon_F) \cdot \varepsilon_R$
    where $f_{K,R}(H) = \sum_{k=1}^K \sum_{j=1}^m c_k b_j \PC^{1/s_j}(s_j t_k)$ is a linear combination of product formulae, $\varepsilon_F$ is an error controlled purely by Fourier coefficients (defined in Eq.~\eqref{eq:fourier-series-approximation}), and $\varepsilon_R$ is an error purely controlled by extrapolation parameters (c.f.~Section \ref{sec:extrapolation-error}). Thus, $f_{F,R}$ takes the form of $G$ in Lemma \ref{lem:distribution-estimation}, with weight of linear combination 
    $S = O(c(\varepsilon_F) \log \log (1/\varepsilon_R))$, where $c:= \sum_k |c_k|$.
    
    
    Recall further from Lemma \ref{lem:richardson-well-conditioned} that the maximum Trotter step number of the product formulae is $O( \log\left(c(\varepsilon_F)/\varepsilon_R \right) \cdot ( a_{\max} \Upsilon \lambda_{\mathrm{comm}} t_{\max}(\varepsilon_F) )^{1 + 1/p} )$.

    Denote $\vec{f}_{K,m}(H)$ as the vector corresponding to the (non-normalized probability) distribution with $i$th entry $|\bra{i}{f}_{K,m}(H)\ket{\psi}|^2$ for all $i$, and likewise $\vec{f}(H)$ as the distribution with $i$th entry $|\bra{i}{f}(H)\ket{\psi}|^2$. We can check that the $\ell_1$ distance between these distribution satsifies
    \begin{align}
        \| \vec{f}_{K,m}(H) - \vec{f}(H) \|_1 &= \sum_i \left| \bra{i} \left({f}_{K,m}(H)\ketbra{\psi}{f}^{\dag}_{K,m}(H) - {f}(H)\ketbra{\psi}{f}^{\dag}(H) \right) \ket{i} \right| \\
        &\leq  \tr\left[ \left| {f}_{K,m}(H)\ketbra{\psi}{f}^{\dag}_{K,m}(H) - {f}(H)\ketbra{\psi}{f}^{\dag}(H)  \right| \right] \\
        &= \left\| {f}_{K,m}(H)\ketbra{\psi}{f}^{\dag}_{K,m}(H) - {f}(H)\ketbra{\psi}{f}^{\dag}(H)  \right\|_1 \\
        &\leq \big\| ({f}_{K,m}(H) - {f}(H))\ketbra{\psi}{f}_{K,m}^{\dag}(H) \big\|_1 + \big\| {f}(H)\ketbra{\psi}({f}_{K,m}^{\dag}(H) - {f}^{\dag}(H)) \big\|_1  \\
        &\leq \left\|({f}_{K,m}(H) - {f}(H))\ket{\psi} \right\|_2 \cdot  \big( \|{f}(H)\ket{\psi} \|_2 + \|{f}_{K,m}\ket{\psi} \|_2\big) \\
        &\leq 2\|{f}_{K,m}(H) - {f}(H) \| \cdot \|f\ket{\psi}\|_2  + \|{f}_{K,m}(H) - {f}(H) \|^2 \\
        &\leq 3 \hat{f} (\varepsilon_F + c(\varepsilon_F) \cdot \varepsilon_R) \qquad (\text{if}\ \varepsilon_F + c(\varepsilon_F) \cdot \varepsilon_R \leq 1)\,, \label{eq:subsumed-error-inequality}
    \end{align}
    where the first equality is an unravelling of definitions, the first inequality is due to the fact that any operator can be written as the difference of two positive operators $X_+$ and $X_-$ and that $|\bra{i} X_+ - X_- \ket{i}| \leq \bra{i} X_+ \ket{i} + \bra{i} X_- \ket{i} = \bra{i} | X_+ - X_-| \ket{i}$, the following equality introduces the Schatten 1-norm, the second inequality is obtained by adding and subtracting terms followed with a triangle inequality, the third inequality is due to the tracial matrix H\"older inequality, and the final inequality follows by definition of the operator norm and denoting $\hat{f} := \max(1,\|f\ket{\psi}\|_2)$.
    
    We now have all the tools to put together the theorem statement. Lemma \ref{lem:distribution-estimation} demonstrates that Algorithm \ref{algo:distribution} returns a vector $\vec{v}$ can be obtained which has statistical error $\varepsilon_s$ characterized as
    \begin{equation}
        \|\vec{v} - \vec{f}_{K,m}(H)\|_2 \leq \varepsilon_s\,,
    \end{equation}
    with probability at least $(1-\delta)$ using $O\!\left( \frac{S^4}{\varepsilon_s^2} \log\left(\frac{1}{\delta} \right) \right) = O\!\left( \frac{(c(\varepsilon_F) \log \log (1/\varepsilon_R))^4}{\varepsilon_s^2} \log\left(\frac{1}{\delta} \right) \right)$ samples. As
    \begin{equation}
        \|\vec{v} - \vec{f}(H)\|_2 \leq \|\vec{v} - \vec{f}_{K,m}(H)\|_2 + \| \vec{f}_{K,m}(H) - \vec{f}(H) \|_1,
    \end{equation}
    (due to triangle inequality and monotonicity of $\ell_p$ norms) we can thus guarantee $\|\vec{v} - \vec{f}(H)\|_2 \leq \varepsilon \leq 1$ by setting $\varepsilon_s = \varepsilon/3$, $\varepsilon_F =  \varepsilon/(9\hat{f}), \varepsilon_R = \varepsilon/(9\hat{f}c(\varepsilon_F))$. Further, we note that $\varepsilon\leq 1$ subsumes the condition in Eq.~\eqref{eq:subsumed-error-inequality}. Substituting these values in, we obtain the stated quantum complexities. The classical preprocessing cost to compute the probability distribution is $O((K(\varepsilon_F)\log(\varepsilon_R))^2)$ as in the proof of Theorem \ref{thm:overlap}.
\end{proof}

\section{$k$-local systems}\label{sec:k-local}

In this section we discuss how to obtain commutator scaling for $k$-local systems for our algorithms, as well as other product formula extrapolation algorithms. The key technical tool is based on an insight of Mizuta to exploit the fact that the BCH formula can be tightly truncated whilst maintaining good approximation of Trotterized time evolution \cite{mizuta2026commutator}.

\begin{lemma}[Truncated BCH formula \cite{mizuta2026commutator}]
    For a $k$-local system on $n$ qubits with maximum energy per site $g$, define the truncated effective Trotter Hamiltonian
    \begin{equation}
        \theff(t,\varepsilon):= H + \sum_{j = 1}^{p_0(\varepsilon)} E_{j+1} t^j\,.
    \end{equation}
    If we set the truncation order as $p_0(\varepsilon) = \lceil \log (2n/\varepsilon) \rceil$ for some parameter $\varepsilon \in (0,1)$, then we have 
    \begin{equation}
        \| e^{-i\Heff(t) t} - e^{-i\theff(t,\varepsilon) t}  \| \leq \varepsilon\,.
    \end{equation}
\end{lemma}

The above lemma implies that for an extrapolation schedule $\{s_i \}_{i}$ with coefficients $\{b_i \}_{i}$, the approximation error can be bounded as 

\begin{equation}
    \big\| \sum_{i} b_i e^{-i\Heff(s_iT) T} - \sum_{i} b_i e^{-i\theff(s_iT, \varepsilon_{\mathrm{tr}}\hat{s}/\|\vec{b}\|_1 ) T}  \big\| \leq \varepsilon_{\mathrm{tr}}\,,
\end{equation}
where $\hat{s}$ is any given lower bound on all $s_i$,
by using a triangle inequality and bounding the telescoping sum. That is, we can control the error by setting a truncation parameter $p_0(\varepsilon_{\mathrm{tr}}\hat{s}/\|\vec{b}\|_1) = \lceil \log(2n\|\vec{b}\|_1/\hat{s}\varepsilon_{\mathrm{tr}})\rceil$. This will be useful to us as it means we only require knowledge of error operators to bounded order $p_0(\varepsilon_{\mathrm{tr}}\hat{s}/\|\vec{b}\|_1)$, and allows us to avoid divergences in the extrapolated commutator factor. Specifically, for any error parameter $\varepsilon$ we can simply repeat the same analysis of Lemma \ref{lem:time-evol-error-series} for $\theff(t,\varepsilon)$ and obtain an error series  

\begin{equation}\label{eq:error-series-k-local}
       e^{-i\theff(sT, \varepsilon)T} - e^{-iHT}= \sum_{j\in \sym\mathbb{Z}_+\geq p} s^j \tilde{E}_{j+1}(T,\varepsilon)\,,
\end{equation}
with error operators now of the form
\begin{equation}
        \tilde{E}_{j+1}(T,\varepsilon) \coloneqq
        \sum_{l=1}^{ \lfloor j/p \rfloor} T^{j+l} \int^1_0 ds_1 \int^{s_1}_0 ds_2 \dots \int^{s_{l-1}}_0 ds_l \sum_{\substack{ p\leq j_1\dots j_l \leq p_0(\varepsilon) \\ \in \sym\mathbb{Z}_+ \\ j_1+\dots+j_l=j}}  \left(\prod_{\kappa=l}^1 e^{-i(s_{\kappa-1}-s_\kappa)TH} i{E_{j_\kappa+1}}\right) e^{-is_l T H}\,,
    \end{equation}
    \begin{align}
        \norm{\tilde{E}_{j+1}(T,\varepsilon)} &\leq   (a_\mathrm{max}\Upsilon T)^j \sum_{l=1}^{ \lfloor j/p \rfloor  }  \frac{(a_\mathrm{max}\Upsilon T)^l}{l!} \sum_{\substack{ p\leq j_1\dots j_l \leq p_0(\varepsilon) \\ \in \sym\mathbb{Z}_+ \\ j_1+\dots+j_l=j}}  \left(\prod_{\kappa=1}^l  \frac{\acomm^{(j_\kappa + 1)}}{(j_\kappa + 1)^2}\right) = \sum_{l=1}^{\lfloor j/p \rfloor  } \frac{(a_\mathrm{max}\Upsilon T \lcomm^{(s\varepsilon_{\mathrm{tr}})} )^{j+l}}{l!}\,,
    \end{align}    
where $\lcomm^{(\varepsilon)}$ is defined as 
\begin{equation}
    \lcomm^{(\varepsilon)}:= \sup_{\substack{j\in \sigma \mathbb{Z}_{+} \geq \sigma m \\ 1 \leq l \leq \lfloor j/p \rfloor}} \Bigg( \sum_{\substack{ p\leq j_1\dots j_l \leq p_0(\varepsilon) \\ \in \sym\mathbb{Z}_+ \\ j_1+\dots+j_l=j}}  \left(\prod_{\kappa=1}^l  \frac{\acomm^{(j_\kappa + 1)}}{(j_\kappa + 1)^2}\right) \Bigg)^{\frac{1}{j+l}}\,.
\end{equation}

\begin{lemma}
    For a $k$-local Hamiltonian $H$ on $n$ qubits with maximum on site energy $g$, we have
    \begin{equation}
        \lcomm^{(\varepsilon)} = O\!\left( kg (p(\Lambda/g)^{1/(p+1)} + \log(n/\varepsilon)) \right)\,.
    \end{equation}
\end{lemma}

\begin{proof}
The analysis is identical to that found in \cite[Lemma 5]{mizuta2026commutator}, and we provide a proof here for completeness. For $k$-local systems, we have $\acomm^{(j)} \leq (j-1)! (2kg)^{j-1} \Lambda$ \cite{childs2021TheoryTrotter}. Thus, we can write 
\begin{align}
    \sum_{\substack{ p\leq j_1\dots j_l \leq p_0(\varepsilon) \\ \in \sym\mathbb{Z}_+ \\ j_1+\dots+j_l=j}}  \left(\prod_{\kappa=1}^l  \frac{\acomm^{(j_\kappa + 1)}}{(j_\kappa + 1)^2}\right) 
    &\leq  \sum_{\substack{ p\leq j_1\dots j_l \leq p_0(\varepsilon) \\ \in \sym\mathbb{Z}_+ \\ j_1+\dots+j_l=j}}  \left(\prod_{\kappa=1}^l  \frac{j!}{(j_{\kappa}+1)^2} (2kg)^{j_{\kappa}} \Lambda \right) \\
    &\leq \sum_{\substack{ p\leq j_1\dots j_l \leq p_0(\varepsilon) \\ \in \sym\mathbb{Z}_+ \\ j_1+\dots+j_l=j}} \prod_{\kappa=1}^l \left( 2kg (j_{\kappa}+1) (\Lambda/g)^{1/(j_{\kappa}+1)} \right)^{(j_{\kappa}+1)} \\
    &\leq \left( \max_{p+1 \leq j' \leq p_0(\varepsilon)+1}\left( 2kg j' (\Lambda/g)^{1/j'} \right) \right)^{j+l} \sum_{\substack{ j_1\dots j_l  \geq 0 \\ j_1+\dots+j_l=j}}  1 \\
    &\leq (4kg)^{j+l} \left( \max_{p+1\leq j' \leq p_0(\varepsilon)+1}\left( j' (\Lambda/g)^{1/j'} \right) \right)^{j+l}\,,
\end{align}
where in the final line we have used the stars-and-bars formula to bound the sum by $\binom{j+l-1}{l-1} \leq 2^{j+l}$. The function $x (\Lambda/g)^{1/x}$ monotonically decreases in $x$ for all $0< x < \log (\Lambda/g)$ and increasing for $x \geq \log (\Lambda/g)$. Thus, we have 

\begin{align}
    \max_{p+1\leq j' \leq p_0(\varepsilon)+1}\left( j' (\Lambda/g)^{1/j'} \right) &\leq \max\left\{ (p+1)(\Lambda/g)^{1/(p+1)},  (p_0(\varepsilon)+1) (\Lambda/g)^{1/(p_0(\varepsilon)+1)}\right\} \\
    &= O\!\left(p(\Lambda/g)^{1/(p+1)} + p_0(\varepsilon)\right)\,,
\end{align}
due to the fact that $\Lambda \leq ng$ and $x^{1/\log (x)} = O(1)$. Putting everything together, we have 
\begin{equation}
    \lcomm^{(\varepsilon)} = O\left(kg (p(\Lambda/g)^{1/(p+1)} + p_0(\varepsilon)) \right)\,,
\end{equation}
as required.
\end{proof}

\begin{lemma}[Generic Richardson extrapolation error for time signal, $k$-local]\label{lemma:generic-Richardson-error-2}
    Let $\PC$ be an $\Upsilon$-staged $p$-th order product formula of symmetry class $\sigma$, where $\sigma=2$ if $\PC$ is symmetric, 1 otherwise. For a target evolution time $T$ let
    \begin{equation}
        \cP^{(R)}_{p,m}(T):= \sum_{k=1}^m b_i \cP^{1/s_i}(s_iT) = \sum_{i=1}^m b_i e^{-i\Heff(s_iT) T}
    \end{equation}
    denote an $m$-term Richardson extrapolation, with ascending sequence of Trotter steps $r_k=1 / s_k \in \mathbb{Z}_{+}$, which cancel the powers $s^\sigma, s^{2\sigma}, \ldots, s^{\sigma(m-1)}$. We consider a Hamiltonian which assumes the conditions of Lemma \ref{lem:time-evol-error-series}. Then, the error in the extrapolation of a Trotterized time signal, as compared to an exact time signal, satisfies
    \begin{equation}
    \big|\tr[ \rho(\cP^{(R)}_{p,m}(T) - e^{-iHT}) ]\big| \leq 4\|\vec{b}\|_1 \eta^{\lfloor \sigma m/p \rfloor} \left(s_1 a_{\max } \Upsilon \lcomm^{(\hat{s}\varepsilon/2\|\vec{b}\|_1)} T\right)^{\sigma m} + \varepsilon/2 \,.
    \end{equation}
    for any quantum state $\rho$, where $\|\vec{b}\|_1=\sum_k\left|b_k\right|$, and $\hat{s}$ is a lower bound on all $s_k$. 
\end{lemma}

\begin{proof}
Set the truncation parameter as $p_0(\varepsilon')= p_0(\varepsilon \hat{s}/2\|\vec{b}\|_1) = \lceil \log(4n\|\vec{b}\|_1/\hat{s}\varepsilon)\rceil$. The extrapolation error takes the form 
\begin{equation}
    \left| \tr[\rho ( \cP^{(R)}_{p,m}(T) - e^{-iHT}) ] \right| \leq \sum_{k=1}^m |b_k| \cdot \|  R_{\sigma(m-1)}(T, s_k, \varepsilon')\|.
\end{equation}
where $R_{\sigma(m-1)}(T, s_k)$ denotes the Taylor remainder of degree $\sigma(m-1)$ which takes the form
\begin{equation}
    R_{\sigma(m-1)}(T, s_k):=\sum_{\substack{j \in \sigma \mathbb{Z}_{+} \\ j \geq \sigma m}} s_k^j \tilde{E}_{j+1}(T,\varepsilon') \,,
\end{equation}
for each inverse Trotter step $s_k$. The size of this is bounded as
\begin{align}
    \left\|R_{\sigma(m-1)}(T, s_k,\varepsilon')\right\| &\leq  \sum_{\substack{j \in \sigma \mathbb{Z}_{+} \\ j \geq \sigma m}} s_k^j\left\|\tilde{E}_{j+1}(T,\varepsilon')\right\|  \\
    & \leq \sum_{\substack{j \in \sigma \mathbb{Z}_{+} \\ j \geq \sigma m}} s_k^j \sum_{l=1}^{\lfloor j/p\rfloor} \frac{\left(a_{\max } \Upsilon \lcomm^{(\varepsilon')} T\right)^{j+l}}{l!} \,,
\end{align}
where as before we have used Lemma \ref{lem:time-evol-error-series} and denoted $\lambda_{j, l}$ as in the statement of this Lemma. 
As $\lcomm^{(\varepsilon')}$ is independent of the extrapolation index $i$, we can proceed as before and bound the extrapolation error simply as 
\begin{align}
    \big|\tr[ \rho(\cP^{(R)}_{p,m}(T) - e^{-iHT}) ]\big| &\leq \|\vec{b}\|_1 \cdot \left\|R_{\sigma(m-1)}(T, s_1)\right\| \leq 4 \|\vec{b}\|_1 \eta^{\lfloor \sigma m/p\rfloor}  \left(s_1 a_{\max } \Upsilon \lcomm^{(\varepsilon')} T \right)^{\sigma m}\,.
\end{align}
on the condition that $s_1 (a_{\max } \Upsilon \lcomm T)^{1+1/p}\leq 1/2$. Finally, one can check that setting $\varepsilon'= \varepsilon \hat{s}/2\|\vec{b}\|_1$ guarantees that the truncation error is bounded as 
\begin{equation}
    \big\| \cP^{(R)}_{p,m}(T) - \sum_{i} b_i e^{-i\theff(s_iT, \varepsilon\hat{s}/2\|\vec{b}\|_1 ) T}  \big\| \leq \varepsilon/2\,,
\end{equation}
and the statement of the lemma follows by a triangle inequality.
\end{proof}

The above lemma allows us to repeat prior analysis with $\lcomm^{(\hat{s}\varepsilon/2\|\vec{b}\|_1)}$ for $k$-local systems in place of $\lcomm$. In the case of compiling time evolution, we have $1/\hat{s} = \mathrm{poly}(n,g,T,\varepsilon^{-1})$ (note that $\Lambda \leq ng$).

Moreover, we have $\lcomm^{(\hat{s}\varepsilon/2\|\vec{b}\|_1)} \leq \lcommklocal$, where 
    \begin{equation}
        \lcommklocal  = O\!\left( kg (p\Lambda^{1/(p+1)} + g \log(n g T/\varepsilon)) \right) \,.
    \end{equation}

Further, when compiling other functions $T$ is simply taken to be the maximum time parameter of the Fourier series, and an additive term of $g\log\log(c(\varepsilon/3))$ is inherited for the Fourier weight $c(\varepsilon/3)$.

\section{Exploiting Hamiltonian symmetries}\label{sec:Fermionic}

\subsection{General result}

In this section we demonstrate that our algorithms (as well as prior art that uses extrapolated product formulae \cite{watson2024exponentiallyReducedTrotter, chakraborty2025quantum}) inherit the ability of product formulae to exploit Hamiltonian symmetries in commutator scaling.

\begin{theorem}[Hamiltonian symmetries]\label{thm:symmetries}
    Consider a Hamiltonian $H = \sum_{\gamma=1}^{\Gamma} H_{\gamma}$ with symmetry $S$ such that $[H_{\gamma},S]=0$ for all $\gamma$. Let $\Pi_i$ be the projector onto the $i$th symmetry sector corresponding to $S$. Then, given any input state that belongs to that symmetry sector $\rho = \Pi_i\rho \Pi_i$, Algorithms \ref{algo:overlap}, \ref{algo:observable}, \ref{algo:distribution} all succeed at their stated tasks with complexities given in Theorems \ref{thm:overlap}, \ref{thm:overlap-with-observable}, \ref{thm:distribution-recovery} respectively, with a refined extrapolated commutator factor $\lcomm \rightarrow \lcommi$, where 
    \begin{equation}
        \lcommi := \sup _{\substack{j \in \sigma \mathbb{Z}_{+} \geq \sigma m \\ 1 \leq l \leq \lfloor j/p \rfloor}} \bigg(\sum_{\substack{j_1 \ldots j_1 \in \sigma \mathbb{Z}_{+} \geq p \\ j_1+\cdots+j_l=j}} \prod_{\kappa=1}^l \frac{\alpha_{\text{comm},i}^{\left(j_\kappa+1\right)}}{\left(j_\kappa+1\right)^2}\bigg)^{1 /(j+l)}\,,
    \end{equation}
    where we have denoted  
    \begin{equation}
        \alpha_{\text{comm},i}^{(j)}:= \sum_{\gamma_1, \dots, \gamma_j = 1}^{\Gamma} \left\|
        \bigl[ H_{\gamma_1}, \dots,  H_{\gamma_j} \bigr] \Pi_i \right\|\,.
    \end{equation}
\end{theorem}

We remark that $\alpha_{\text{comm},i}^{(j)}$ (c.f.~Definition \ref{def:acomm}) is exactly the commutator factor that characterizes the complexity of product formulae given an input state that belongs to the $i$th symmetry sector for a symmetric Hamiltonian.

\begin{proof}
    Under the stated conditions, we demonstrate a modified version of Lemmas \ref{lem:time-evol-error-series} and \ref{lem:richardson-well-conditioned}. These lemmas can be propagated through to Theorems \ref{thm:overlap}, \ref{thm:overlap-with-observable} and \ref{thm:distribution-recovery}. First, we note that time evolution unitaries of $H$ commute with $S$ and $\Pi_i$. Second, any product formula constructed from the decomposition given in the Theorem statement $H = \sum_{\gamma=1}^{\Gamma} H_{\gamma}$ commutes with $S$ and $\Pi_i$, as it consists of a product of time evolution operators of the form $\exp(-iH_{\gamma} \tau)$. Third, the Hamiltonian error operators $E_j$ (defined in Lemma \ref{lem:Effective_Hamiltonian_Error_Series}) also commute with $S$ and $\Pi_i$, as they consist of a linear combination of nested commutators of $H_{\gamma}$. Finally, the time evolution error operators $\tilde{E}_{j}$ (defined in Eq.~\eqref{eq:tilde_E_explicit}) also commute with $S$ and $\Pi_i$, as they are formed from products of time evolution unitaries and $E_j$.

    The first two statements allow us to write a projected error series (modifying Lemma \ref{lem:time-evol-error-series}):
    \begin{equation} 
       \Pi_i (\cP^{1/s}(sT) - e^{-iHT}) =  (\cP^{1/s}(sT) - e^{-iHT})\Pi_i  = \sum_{j\in \sym\mathbb{Z}_+\geq p} s^j  \tilde{E}_{j+1}(T) \Pi_i\,.
    \end{equation}
    Further, the first and fourth statements allow us to denote $\tilde{E}_{j+1}(T) \Pi_i  = \tilde{E}_{j+1}^{(i)}(T)$ where
    \begin{align}
        \norm{\tilde{E}^{(i)}_{j+1}(T)} &\leq   (a_\mathrm{max}\Upsilon T)^j \sum_{l=1}^{ \lfloor j/p \rfloor  }  \frac{(a_\mathrm{max}\Upsilon T)^l}{l!} \sum_{\substack{j_1\dots j_l\in \sym\mathbb{Z}_+\geq p \\ j_1+\dots+j_l=j}}  \left(\prod_{\kappa=1}^l  \frac{\alpha_{\textrm{comm},i}^{(j_\kappa + 1)}}{(j_\kappa + 1)^2}\right)\\
        &= \sum_{l=1}^{\lfloor j/p \rfloor  } \frac{(a_\mathrm{max}\Upsilon T \lcommi)^{j+l}}{l!}\,,
    \end{align}
    with $\alpha_{\textrm{comm},i}^{(j_\kappa + 1)}$ and $\lcommi$ defined in the theorem statement, where we have used the fact that $\Pi_i = \Pi_i^2$. From here we can modify Lemma \ref{lem:richardson-well-conditioned} to bound the projected error: we can attain extrapolation error $\| (e^{-iHT} -  \cP^{(R)}_{p,m}(T)) \Pi_i\| = \| \Pi_i(e^{-iHT} -  \cP^{(R)}_{p,m}(T)) \| \leq \varepsilon$  
    using $m=O(\log(1/\varepsilon))$  extrapolation circuits, each with at most
    $O\!\left(\left(a_{\max } \Upsilon \lcommi T\right)^{(1+1 / p)} \log (1 / \varepsilon)\right)$ Trotter steps. We note this is an identical result to Lemma \ref{lem:richardson-well-conditioned} except with $\lcomm \rightarrow \lcommi$.

    Let us now inspect the three quantities required of Tasks (\hyperref[eq:task-1]{1(i)}) (\hyperref[eq:task-1]{1(ii)}) and \hyperref[eq:task-2]{2}. In each case, we are either approximating the operator $\rho f(H)$ or $f(H) \rho f(H)$, where $f(H)$ is decomposed into a linear combination of time evolution operators. As $\rho = \Pi_i\rho \Pi_i$, we can recover all our results of Theorems \ref{thm:overlap}, \ref{thm:overlap-with-observable} and \ref{thm:distribution-recovery} using our new projected error bounds, attaining the same complexities except with $\lcomm \rightarrow \lcommi$.
\end{proof}

\subsection{Input states with well-defined Fermion number}

In the remainder of this section we discuss implications for Fermionic systems as a specific example.

In many quantum simulation settings (such as electronic structure problems and lattice fermion models) the dynamics are number preserving. That is, the Hamiltonians and observables of interest commute with the number operator $\hat{n}_{i} = \hat{a}_{i}^{\dag} \hat{a}_{i}$, where we denote $ \hat{a}_{i}$ and $ \hat{a}_{i}^{\dag}$ as the Fermionic annihilation and creation operators respectively for the $i$th mode. Under such dynamics, if the input state has a well-defined Fermion number, tighter Trotter error bounds can be established, by considering the size of the error operator under the so-called Fermionic seminorm \cite{mcardle2022ExploitingFermionNumber, su2021NearlyTightTrottInerElect, low2023complexityTrotter}. Specifically, the complexity depends on  where $\alpha_{\eta\text{-}\mathrm{comm}} 
\ll \alpha_{\mathrm{comm}}$ if the number of Fermions is much smaller than the number of basis states (such as orbitals).


We work in a Fermionic Fock space with fixed particle number sectors (symmetry sectors) labeled by $\eta$. For any operator $X$ that preserves particle number, we define the Fermionic $\eta$-seminorm:
\begin{equation}
\|X\|_\eta := \max_{|\psi_\eta\rangle,\, |\phi_\eta\rangle} \frac{|\langle \phi_\eta | X | \psi_\eta \rangle|}{\| \psi_\eta \| \cdot \| \phi_\eta \|},
\end{equation}
where $\ket{\psi_\eta}$ and $\ket{\phi_\eta}$ are $\eta$-electron states. Importantly, this seminorm quantifies the action of $X$ restricted to the $\eta$-particle subspace.

\begin{lemma}[Fermionic seminorm as a projected spectral norm \cite{su2021NearlyTightTrottInerElect}]
\label{prop:Fermionic_seminorm}
For any number-preserving operator $X$, the Fermionic $\eta$-seminorm satisfies
\begin{equation}
\|X\|_{\eta} = \max_{|\psi_{\eta}\rangle, |\phi_{\eta}\rangle} |\langle\phi_{\eta}|X|\psi_{\eta}\rangle| = \|X\Pi_{\eta}\|,
\end{equation}
where $\Pi_\eta$ is the projector onto the $\eta$-Fermion subspace.
\end{lemma}

We can thus use Theorem \ref{thm:symmetries} to write a refined statement for our algorithms when the input state has a well-defined Fermion number.

\begin{theorem}[Exploiting low Fermion number]\label{thm:fermions}
    Consider a number preserving Hamiltonian $H = \sum_{\gamma=1}^{\Gamma} H_{\gamma}$ such that $[H_{\gamma},\hat{n}]=0$ for all $\gamma$. Then, given any input state with well-defined Fermion number $\eta$, Algorithms \ref{algo:overlap}, \ref{algo:observable}, \ref{algo:distribution} all succeed at their stated tasks with complexities given in Theorems \ref{thm:overlap}, \ref{thm:overlap-with-observable}, \ref{thm:distribution-recovery} respectively, with a refined extrapolated commutator factor $\lcomm \rightarrow \lcomm^{(\eta)}$, where 
    \begin{equation}
        \lcomm^{(\eta)} := \sup _{\substack{j \in \sigma \mathbb{Z}_{+} \geq \sigma m \\ 1 \leq l \leq \lfloor j/p \rfloor}} \bigg(\sum_{\substack{j_1 \ldots j_1 \in \sigma \mathbb{Z}_{+} \geq p \\ j_1+\cdots+j_l=j}} \prod_{\kappa=1}^l \frac{\alpha_{\eta\text{-}\mathrm{comm}}^{(j_\kappa + 1)}}{\left(j_\kappa+1\right)^2}\bigg)^{1 /(j+l)}\,,
    \end{equation}
    where we have denoted  
    \begin{equation}
      \alpha_{\eta\text{-}\mathrm{comm}}^{(j)}:= \sum_{\gamma_1, \dots, \gamma_j = 1}^{\Gamma} \left\| \bigl[ H_{\gamma_1}, \dots,  H_{\gamma_j} \bigr] \right\|_{\eta}\,.
    \end{equation}
\end{theorem}

In \cite{su2021NearlyTightTrottInerElect, mcardle2022ExploitingFermionNumber, low2023complexityTrotter, watson2023QuantumAlgorithmNuclearEFT}, the Trotter error with respect to the Fermionic seminorm has been carefully studied -- more specifically, bounds on $\alpha_{\eta\text{-}\mathrm{comm}}^{(j)}$ are known. In what follows, we give an example of how this materially allows the same scaling to be translated to our extrapolation algorithms. We quote a theorem of \cite{low2023complexityTrotter}.

\begin{theorem}[Trotter error with Fermionic induced 1-norm scaling (\cite{low2023complexityTrotter}, Theorem 4)]\label{thm:fermion-low}
Let $H = T + V := \sum_{j,k} \tau_{j,k} \hat{a}_j^\dagger \hat{a}_k + \sum_{l,m} \nu_{l,m} \hat{n}_l \hat{n}_m$ be an interacting-electronic Hamiltonian, where the sums are over n  orbitals and $\hat{n}_{i}$ is the orbital occupation number. Let $\PC$ denote a $p$-th-order product formula based on this decomposition. Then,
\begin{equation}
(\alpha_{\eta\text{-}\mathrm{comm}}^{(p+1)})^{1/p} = O \!\left( \left( \vertiii{\tau }_1 + \vertiii{\nu}_{1,\eta} \right)^{1-1/p} \left(\vertiii{\tau }_1 \vertiii{\nu}_{1,\eta} \, \eta\right)^{1/p} \right)\,,
\end{equation}
where we define norms on Hamiltonian parameters as
\begin{equation}
\vertiii{\tau }_1 = \max_j \sum_k |\tau_{j,k}|, \quad
\vertiii{\nu}_{1,\eta} = \max_j \max_{k_1 < \cdots < k_\eta} \left( |\nu_{j,k_1}| + \cdots + |\nu_{j,k_\eta}| \right).
\end{equation}
\end{theorem}

The above bound (and similar results in \cite{su2021NearlyTightTrottInerElect, mcardle2022ExploitingFermionNumber}) can be used to give refined gate complexities for simulating systems of Fermions when the input state has well-defined Fermion number which is small. For instance, for the uniform electron gas (Jellium) in a small Fermion number regime $\eta = O(n)$, Theorem \ref{thm:fermion-low} leads to a gate complexity $O(n^{5/3 +o(1)})$, which can be compared to other state-of-art techniques which do not exploit any prior on the input state, which attain $O(n^{2+o(1)})$ gate complexity for Trotter methods or   $\widetilde{O}(n^2)$ for methods which use more ancillary qubits.

\begin{theorem}[Extrapolated error with Fermionic induced 1-norm scaling]\label{thm:lcomm-fermionic}
    Subsume the conditions of Theorem \ref{thm:fermion-low}, and suppose that 
    $\eta = \Omega (\max( \vertiii{\tau }_1 /\vertiii{\nu}_{1,\eta} , \vertiii{\nu}_{1,\eta}/\vertiii{\tau }_1)) $. We have 
    \begin{equation}
\lcomm^{1+1/p} = O \!\left( \left( \vertiii{\tau }_1 + \vertiii{\nu}_{1,\eta} \right)^{1-1/p} \left(\vertiii{\tau }_1 \vertiii{\nu}_{1,\eta} \, \eta\right)^{1/p} \right)\,,
\end{equation}
\end{theorem}

\begin{proof}
Lemma \ref{lem:lcomm-tight-bound} gives a bound for $\lcomm$ for any bound on $\acomm^{(j)}$ of the form $\acomm^{(j)} \leq A n^j$. Asserting the condition  $\vertiii{\tau }_1 \vertiii{\nu}_{1,\eta} \, \eta  \geq ( \vertiii{\tau }_1 + \vertiii{\nu}_{1,\eta} )^2$ gives $A\geq 1$ and this can be substituted into Lemma \ref{lem:lcomm-tight-bound} to obtain the stated result. 
\end{proof}

To illustrate an example of when the conditions of Theorem \ref{thm:lcomm-fermionic} are valid, take all $|\nu_{j,k}|, |\tau_{j,k}| = \Theta(1)$. Then, the theorem conditions are satisfied for any systems with Fermion number $\eta = \Omega(\sqrt{n})$.

\section{Gate complexities without extrapolated commutator factor}\label{sec:lcomm-fixed-order}

In our derivations on the error series for time evolution operators, we ultimately establish a bound on the size of error operators as

\begin{equation}
    \|\widetilde{E}_{j+1}(T) \| = O( (a_{\max} \Upsilon \lcomm T)^{j (1+1/p)} )\,,
\end{equation}

(see Lemmas \ref{lem:time-evol-error-series} and \ref{lem:generic-Richardson-error}). In this section we consider an alternative bound which is wholly independent of $\lcomm$ and can be propagated through to produce alternative gate complexities for all of our core algorithms. This enables refined analytical guarantees for when bounds on $\lcomm$ are not tight, yet non-trivial bounds on $\acomm^{(p+1)}$ are known. 

\begin{lemma}\label{lemma:error-operator-bound-acomm}
Consider the time evolution operator error series in Lemma \ref{lem:time-evol-error-series}. We have
    \begin{equation}
    \|\widetilde{E}_{j+1}(T) \| \leq \begin{cases}
         (e-1)\left(4 a_{\max } \Upsilon \Lambda T\right)^j \quad \qquad \text{if}\;\  \acomm^{(p+1)} T \leq \frac{(4 \Lambda)^p}{2 a_{\max } \Upsilon} \\
        (e-1) \left( (2 \acomm^{p+1})^{1/p}(a_{\max } \Upsilon T)^{1+1/p}) \right)^j \qquad  \text{else}
    \end{cases}
\end{equation}
\end{lemma}

\begin{proof}
    In \cite{watson2024exponentiallyReducedTrotterb} the nesting property 
    \begin{equation}
        \alpha_{\mathrm{comm}}^{(j)} \leq(2 \Lambda)^{j-k} \alpha_{\mathrm{comm}}^{(k)}
    \end{equation}
    is demonstrated for any $k,j$ satisfying $2 \leq k \leq j$. From there the authors show that   
    \begin{equation}
        \left\|\tilde{E}_{j+1}(T)\right\| \leq\left(4 a_{\max } \Upsilon \Lambda T\right)^j \sum_{l=1}^{\lfloor j / p\rfloor} \frac{1}{l!}\left(\frac{2 a_{\max } \Upsilon \acomm^{(p+1)} T}{(4 \Lambda)^p}\right)^l\,.
    \end{equation}
    Now we can treat our two cases separately, based on the value of $\left(\frac{2 a_{\max } \Upsilon \acomm^{(p+1)} T}{(4 \Lambda)^p}\right)$. If this is $\leq 1$ then we can bound the sum by the first term as
    \begin{equation}
        \left\|\tilde{E}_{j+1}(T)\right\| \leq (e-1)\left(4 a_{\max } \Upsilon \Lambda T\right)^j  \left(\frac{2 a_{\max } \Upsilon \acomm^{(p+1)} T}{(4 \Lambda)^p}\right) \leq (e-1)\left(4 a_{\max } \Upsilon \Lambda T\right)^j\,,
    \end{equation}
    where we have used the fact that $\sum_{l=1}^{\infty} \frac{1}{l!} \leq (e-1)$.
    If the term is $\geq 1$ then we can bound the sum by the largest term and we have 
    \begin{align}
        \left\|\tilde{E}_{j+1}(T)\right\| \leq (e-1)\left(4 a_{\max } \Upsilon \Lambda T\right)^j  \left(\frac{2 a_{\max } \Upsilon \acomm^{(p+1)} T}{(4 \Lambda)^p}\right)^{j / p} \\
        \leq (e-1) \left( (2 \acomm^{(p+1)})^{1/p}(a_{\max } \Upsilon T)^{1+1/p}) \right)^j \,.
    \end{align}
\end{proof}

The alternative bound we establish above can be propagated through to give alternative gate complexities for our three core algorithms. Using the bound on the size of error operators in Lemma \ref{lemma:error-operator-bound-acomm}, the extrapolation error found in Lemma \ref{lem:generic-Richardson-error} is modified as 
\begin{equation}
    \left\| \cP^{(R)}_{p,m}(T) - e^{iHT} \right\| \leq 4 \|\vec{b}\|_1 s^{\sigma m} \cdot \max\big\{1,  (\acomm^{(p+1)})^{\lfloor \sigma m/p\rfloor}\left(a_{\max } \Upsilon T \right)^{\sigma m + \lfloor \sigma m/p\rfloor}, \left(a_{\max } \Upsilon \Lambda T \right)^{\sigma m} \big\}\,.
\end{equation}

Following prior proof steps, in order to obtain extrapolation error in the time evolution $\| \cP^{(R)}_{p,m}(T) - e^{iHT} \| \leq \varepsilon$, it is sufficient to take number of Trotter steps 
    \begin{equation}
        r_{\max} = q_m \left\lceil \frac{1}{q_1} \max \Big\{ 1,\left(a_{\max } \Upsilon \Lambda T\right), (\acomm^{(p+1)})^{1 / p}\left(a_{\max } \Upsilon T\right)^{1+1/p} \Big\} \left(\frac{4\|\vec{b}\|_1}{\varepsilon}\right)^{\frac{1}{\sigma m}}\right\rceil\,.
    \end{equation}
in the setting of Lemma \ref{lem:generic-Richardson-error}
and
    \begin{equation}
        r_{\max} = q_m \left\lceil \frac{1}{q_1} \max \Big\{ 1,\left(a_{\max } \Upsilon \Lambda T\right), (\acomm^{(p+1)})^{1 / p}\left(a_{\max } \Upsilon T\right)^{1+1/p} \Big\} \left(\frac{4\|\vec{b}\|_1}{\varepsilon}\right)^{\frac{1}{\sigma (m-1) + p}}\right\rceil\,.
    \end{equation}
in the setting of Lemma \ref{lem:generic-Richardson-error-2}. Finally, taking the extrapolation schedule of \cite{low2019Multiproduct} and $m=O(\log(1/\varepsilon))$ as before we can recover prior results with modified gate complexities, which we encapsulate with the following theorem.

\begin{theorem}[Gate complexities without extrapolated commutator factor]
Algorithms \ref{algo:overlap}, \ref{algo:observable}, \ref{algo:distribution} all succeed at their stated tasks with complexities given in Theorems \ref{thm:overlap}, \ref{thm:overlap-with-observable}, \ref{thm:distribution-recovery} respectively with the replacement
\begin{equation}
    (\lcomm T(\varepsilon))^{1+1/p} \rightarrow  (\acomm^{(p+1)})^{1/p} (T(\varepsilon) )^{1+1/p} + \Lambda T(\varepsilon)
\end{equation}
for each respective prescribed $\varepsilon$.
\end{theorem}

\section{Doubly Randomized Extrapolation Trotterization}
\label{sec:partial-randomization}

In this section we demonstrate how to integrate our Randomized Extrapolation Trotterization scheme with the partial randomization ideas introduced in~\cite{gunther2025phase}. This ultimately allows us to construct a gate depth that interpolates between the gate depth of Randomized Extrapolation Trotterization and the randomized compilation of \cite{wan2021RandPhaseEst}.

\subsection{Matrix decompositions with long tails}

As in \cite{gunther2025phase}, we presume ability to adopt a Hamiltonian decomposition
\begin{equation}
\label{eq:partial-randomization-decomposition}
H = 
\underbrace{\sum_{l=1}^{\Gamma_A} H_l}_{H_A} + 
\underbrace{\sum_{m=1}^{\Gamma_B} h_m P_m}_{H_B}\,,
\end{equation}
where $H_l$ are generic Hermitian operators (each assumed as before to be simulable in $O(1)$ depth), and $P_m$ are Pauli operators (or more generally, Hermitian unitaries).
The idea will be to compile \( H_A \) via Randomized Extrapolation Trotterization and \( H_B \) is simulated via the randomized compiler of \cite{wan2021RandPhaseEst}.

This decomposition will be especially useful if it is the case that
\begin{equation}
    \Lambda_B := \sum_{m=1}^M |h_m| \ll \Lambda, 
    \quad \text{and} \quad \Gamma_A \ll L,
\end{equation}
where \( \Lambda = \sum_{l=1}^{\Gamma_A} \|H_l\| + \sum_{m=1}^{M} |h_m| \) is the total 1-norm of \( H \) and $L=\Gamma_A+\Gamma_B$. This is the case when the majority of the weight of the matrix is concentrated on very few terms in $H_A$, i.e.~if $H$ has a long tail. We remark that this feature of matrices is already exploited by the randomized compiler of \cite{wan2021RandPhaseEst} which avoids dependence on $L$ (see Table \ref{tab:time-signal}) -- partial randomization will allow us to exploit long-tailed matrices even more effectively.

Before we introduce how we use the core idea from partial randomization to add an extra randomization loop to Randomized Extrapolation Trotterization, we share a useful lemma that allows us to estimate time evolution operator as a linear combination of random unitaries. This is the core lemma that the partial randomization approach exploits.

\begin{lemma}[Random Compiler Lemma, adapted from \cite{wan2021RandPhaseEst} Lemma 2]
\label{lem:random-compiler-lemma}
Let $H_B = \sum_{k} h_k P_k$ be a Hermitian operator, where $P_k^2 = I$ and $h_k \in 
\mathbb{R}$ for all $k$. Define $\Lambda_B := \sum_k |h_k|$. Then, for any time $t > 0$ there exists a decomposition
\begin{equation}
    e^{-iH_Bt} = \beta(t,d) \sum_{j \in \mathbb{Z}_+} p_j(t,d)\, U_j(t,d),
\end{equation}
where:
\begin{itemize}
\item $\{p_j(t,d)\}_j$ is an (efficiently computable) probability distribution.
\item Each $U_j(t,d)$ is a sequence of quantum gates with tuneable non-Clifford depth $d$.
\item The normalization factor $\beta(t,d)$ is efficiently computable and satisfies $\beta(t,d)=O(\exp(\Lambda^2 t^2/ d))$.
\end{itemize}
\end{lemma}

Each $U_j(t,d)$ contains $O(d)$ Pauli rotations interleaved with layers of $O(j)$ Pauli matrices. In practice, the Pauli matrices can be classically compiled together in $O(jn)$ time, and so the quantum gate depth can be taken to be $O(d)$. In order to deal with the infinite sum, Wan et al.~demonstrate that the sum can be truncated with error $\varepsilon$ for truncation order scaling logarithmically in $1/\varepsilon$ and other problem parameters. Thus, the classical complexity remains efficient for a small controllable error, which we hereon disregard for simplicity. 

\subsection{Algorithm and analysis}

We now describe how to integrate Randomized Extrapolation Trotterization with partially randomized product formulae, and show that similar to standalone Randomized Extrapolation Trotterization, this improves the error scaling from \( \varepsilon^{-1/p} \) to \( \log(1/\varepsilon) \), whilst giving a new potentially advantageous scaling with respect to Hamiltonian parameters.

We denote the partially randomized product formula $\mathcal{S}_{p}(t)$  corresponding to the $\Gamma_A + 1$ term decomposition of the Hamiltonian $H$ into $\{H_l\}_{l=1}^{L_{A}} \cup \{H_B\}$. For instance, we can write the first-order formula as
\begin{equation}
\mathcal{S}_{1}(t) = \left( \prod_{l=1}^{\Gamma_A} e^{-i H_l t} \right) e^{-i H_B t}\,,
\end{equation}
for Trotter time step $t$.
This will be a helpful mathematical construction and we remark so far we are lacking a full prescription of how to implement this unitary as $H_B$ contains many terms. 
Similar to our constructions in Section \ref{sec:extrapolation-error}, we also define the $m$-term Richardson extrapolated approximation of the full evolution as
\begin{equation}
\mathcal{S}_{p,m}^{(R)}(T) := \sum_{k=1}^{m} b_k \, \mathcal{S}_p^{1/s_k}(s_k T),
\end{equation}
for any extrapolation schedule $R = \{b_k, q_k\}_{k=1}^m$.

Our strategy will be to on one side use Randomized Extrapolation Trotterization to control the error associated with the decomposition leading to product formula \( \mathcal{S}_{p}(t) \) leading to some contribution of extrapolated Trotter error plus statistical error. On the other side, we compile each instance of the operator $e^{-i H_B t}$ using the random compiler lemma (Lemma \ref{lem:random-compiler-lemma}) incurring only statistical error. By performing one randomized compilation all at once, the statistical error is essentially shared. 

We present an explicit algorithm for our warm-up problem of estimating time signals in Algorithm \ref{algo:PRPF-Richardson}, which we focus on to highlight the new mechanism. Algorithms for all our other tasks and ensuing corresponding analysis follow in the same way.

\begin{algorithm}[h]
\caption{\ Doubly Randomized Extrapolation Trotterization for time signal $\tr[\rho e^{-iHT}]$}
\label{algo:PRPF-Richardson}
\begin{algorithmic}[1]

\Require Hamiltonian $H = H_A + H_B$ of the form in Eq.~\eqref{eq:partial-randomization-decomposition}, Richardson extrapolation schedule $R=\{b_k,q_k\}_{k=1}^m$; base Trotter step number $1/s$; input state $\rho$ with efficient preparation; product formula $\mathcal{S}(t)$ with $\Upsilon$ stages consisting of evolution of $H_B$ for times $\{a_{\nu} t\}_{\nu}^{\Upsilon}$.

\State \textbf{Classical preprocessing:} Compute the  $m\Upsilon+1$ probability distributions:
\[
\left\{  \left| b_k  \right| / \|\vec{b}\|_1 \right\}_{k}\,,\ \big\{ \{p_j(a_\nu s_k T, \Lambda_B^2 T^2 s_{k})\}_j \big\}_{k,\nu} \quad \text{for }\ k = 1,\ldots,m\,,\quad \nu = 1, \ldots, \Upsilon
\]

\For{$i = 1$ to $M$}
    \State Sample index $k'$ from the distribution $\{ {|b_k|}/{\|\vec{b}\|_1} \}_k$\,.
    \State Take ${s_{k'}} \gets s/q_{k'}$, $d_{k'} \gets {\Lambda_B^2 T^2}{s_{k'}}$\,.
    
    \For{$u = 1$ to $1/s_{k'}$}
        \State Sample index $j_{\nu,u}$ from distribution $\{p_j(a_\nu s_{k'}T,d_{k'})\}_j$ for each $\nu \in [\Upsilon]$\,.
        \State {Take $\mathcal{S}_{u}(s_{k'}T) \gets$  $(\mathcal{S}(s_{k'}T)$ s.t.~each $\exp(-iHa_{\nu}t_{k'})$ is replaced with $U_{j_{\nu,u}}(a_{\nu}s_{k'}T,d_{k'}) )$}\,.
    \EndFor
    \State Prepare Hadamard test circuits corresponding to:
    \[
        \operatorname{Re}\left[\operatorname{Tr}[\rho \Pi_{u=1}^{1/s_{k'}}\mathcal{S}_u(s_{k'} T)]\right]\,,\quad \operatorname{Im}\left[\operatorname{Tr}[\rho \Pi_{u=1}^{1/s_{k'}}\mathcal{S}_{u}(s_{k'}T)]\right]\,,
    \]
    \hspace{1.25em} collect one measurement outcome (single-shot statistic) from each circuit $\left(X_{\operatorname{Re}}^{\left(i\right)}, X_{\operatorname{Im}}^{\left(i\right)}\right)$\,.
    \State Take $X^{(i)} \gets X_{\operatorname{Re}}^{(i)} + i \cdot X_{\operatorname{Im}}^{(i)}$\,.
    \State Take $Y^{(i)} \gets X^{(i)} \cdot \operatorname{sign}(b_{k'}) \cdot \|\vec{b}\|_1 \cdot \prod_{u=1}^{1/s_{k'}} \prod_{\nu=1}^{\Upsilon} \beta(a_{\nu}s_{k'}T, d_{k'})$\,.
\EndFor
\State \textbf{Return:} $\hat{Y}_M \leftarrow$ the mean of $\{Y^{(i)} \}_i$\,.

\end{algorithmic}
\end{algorithm}

Let us first take a moment to parse what Algorithm \ref{algo:PRPF-Richardson} is doing. We randomly compile according to Randomized Extrapolation Trotterization as before (first $\texttt{for}$ loop) -- that is, we collect Hadamard test data from randomly selected product formulae -- with one modification. Now, when an operator $e^{-i H_B t}$ appears in the product formula, we randomly select a $U_j(t,d)$ drawn from the probability distribution in Lemma \ref{lem:random-compiler-lemma}. This is done independently for each instance of an operator of the form $e^{-i H_B t}$ appearing in the product formula (second $\texttt{for}$ loop), and facilitates random compilation of each $e^{-i H_B t}$. In the following theorem we demonstrate the required complexities for the algorithm to succeed.

\begin{theorem}[Doubly Randomized Extrapolation Trotterization for time signals]
Consider a matrix $H \in \mathbb{C}^{2^n \times 2^n}$ decomposed as $H = H_A + H_B$ as in Eq.~\eqref{eq:partial-randomization-decomposition}. Let $\tilde{\lambda}_{\text{comm}}$ be the extrapolated commutator factor (as in Eq.~\eqref{eq:lcomm}) corresponding to the decomposition of $H$ into $\{H_{l}\}_{l=1}^{L_{A}} \cup H_{B}$. Call the LKW well-conditioned extrapolation strategy of Lemma~\ref{lemma:well-conditioned-Richardson}. 
Then, using Algorithm~\ref{algo:PRPF-Richardson} we can obtain an $\varepsilon$-additive approximation to $\mathrm{Tr}[\rho e^{-iHT}]$ with the following resources:
\begin{itemize}
    \item \textbf{Circuit samples:}
    \begin{equation}
    C_{\text{sample}} = O\!\left(\frac{(\log \log(1/\varepsilon))^2}{\varepsilon^{2}} \log\left(\frac{1}{\delta}\right)\right)\,,
    \end{equation}
    \item \textbf{Gate depth (per circuit):}
    \begin{equation}
    C_{\text{gate}} = O\!\left(\Gamma_A (a_{\max} \Upsilon \tilde{\lambda}_{\mathrm{comm}} T)^{1 + \frac{1}{p}} \log\left(1/\varepsilon\right) + \Lambda_{B}^{2}T^{2}\right) \,.
    \end{equation}
    \item \textbf{Classical preprocessing time:}
    \begin{equation} C_{\text{pre}} = O(\polylog(1/\varepsilon))\,,
    \end{equation}
\end{itemize}
and each circuit uses $n+1$ qubits.
\label{theorem:extrapolated-partial-randomization}
\end{theorem}
\begin{proof}
Let \( \hat{Y}_{M}\) denote the empirical estimate returned by Algorithm~\ref{algo:PRPF-Richardson} after \( M \) samples. We start by verifying the expectation value of this estimator corresponds to the Richardson extrapolated product formula:
    \begin{align}
        \mathbb{E}[\hat{Y}_{M}] = \mathbb{E}[Y^{(i)}] 
        &= \sum_{k=1}^{m}\frac{|b_{k}|}{\|\vec{b}\|_1} \mathbb{E}\Big[ \prod_{u=1}^{1/s_{k}} \prod_{\nu=1}^{\Upsilon} \beta(a_{\nu}s_{k}T, d_{k})\ \text{sign}(b_{k}) \cdot \|\vec{b}\|_1 \cdot X^{(i)} \Big] \\
        &= \sum_{k=1}^{m} b_{k} \mathbb{E}\Big[ \prod_{u=1}^{1/s_{k}} \prod_{\nu=1}^{\Upsilon} \beta(a_{\nu}s_{k}T, d_{k})\ X^{(i)}\Big] \\
        &= \sum_{k=1}^{m} b_{k} \prod_{u=1}^{1/s_{k}} \prod_{\nu=1}^{\Upsilon} \beta(a_{\nu}s_{k}T, d_{k})\ \mathbb{E}\big[X^{(i)}]  \\
        &= \sum_{k=1}^{m} b_{k} \prod_{u=1}^{1/s_{k}} \prod_{\nu=1}^{\Upsilon}  \beta(a_{\nu}s_{k}T, d_{k})\  \mathbb{E}\big[\operatorname{Tr}[\rho \Pi_{u=1}^{1/s_{k}}\mathcal{S}_{u}(s_{k}T)] \big] \\
        &=\sum_{k=1}^{m} b_{k} \mathrm{Tr}\big[\rho \, \mathcal{S}(s_kT)^{1/s_k}\big] \\
        &=\mathrm{Tr}\big[\rho \, \mathcal{S}_{p,m}^{(R)}(T)\big]
    \end{align}
where in the first line we have explicitly written out the expectation value over product formulae, in the third line we use the fact that each draw in the second \texttt{for} loop in Algorithm \ref{algo:PRPF-Richardson} is independent, in the fourth line we have taken the expectation value over circuit measurement outcomes, in the fifth line we have used Lemma \ref{lem:random-compiler-lemma}, and in the final line follows from the definition of the extrapolated product formula $\mathcal{S}_{p,m}^{(R)}(T)$.

As in prior analyses, we decompose the total error contributions into statistical error $\varepsilon_{S}$ and Richardson extrapolation error $\varepsilon_{R}$:
\begin{equation}
\left| \hat{Y}_{M} - \mathrm{Tr}[\rho e^{-iHT}] \right| \leq \underbrace{\left|\hat{Y}_{M} - \mathrm{Tr}\left[\rho \, \mathcal{S}_{p,m}^{(R)}(T)\right]\right|}_{\varepsilon_{\text{S}}} + \underbrace{\left|\mathrm{Tr}\left[\rho \, \mathcal{S}_{p,m}^{(R)}(T)\right]-\mathrm{Tr}\left[\rho e^{-iHT}\right]\right|}_{\varepsilon_{\text{R}}}.
\end{equation}
To guarantee overall error at most \( \varepsilon \), it is sufficient to ask for 
$\varepsilon_{\text{S}}, \varepsilon_{\text{R}} \leq {\varepsilon}/{2}$.

\textit{Extrapolation error.}  
By Lemma~\ref{lem:richardson-well-conditioned}, the extrapolation can be constrained as $\varepsilon_{\text{R}} \leq {\varepsilon}/{2}$ if  $s$ is chosen such that each constituent product formula in $\mathcal{S}_{p,m}^{(R)}(T)$
uses \( r_k \leq r_{\max} \) Trotter steps, where
\begin{equation}
r_{\max} = O\!\left((a_{\max} \Upsilon \tilde{\lambda}_{\text{comm}} T)^{1 + \frac{1}{p}} \log\left(1/\varepsilon\right)\right).
\end{equation}
Given a $p=O(1)$ product formula, each Trotter step consists of $O(\Gamma_A)$ terms, of which $O(1)$ terms are operators of the form $O(e^{-iH_B t_k})$ where $t_k = O(T/r_k)$, and all remaining terms are each implementable in $O(1)$ time based on our starting assumptions (see Definition \ref{def:matrix}).

\textit{Sample Complexity.}  
Now let us consider the set of parameters specified in Algorithm \ref{algo:PRPF-Richardson}. We see that each instance of our estimator has size 
\begin{equation}
    |Y^{(i)}| \leq 2 \|\vec{b}\|_1 \cdot \prod_{u=1}^{1/s_{k'}} \prod_{\nu=1}^{\Upsilon} \beta(a_{\nu}s_{k'}T)
\end{equation}
where $\beta(a_{\nu}s_{k'}T, d_k) \leq \exp\left(\frac{\Lambda_B^2 (a_{\nu}s_{k'}T)^2}{d_k} \right) = \exp(a_{\nu}^2 s_{k'})$, recalling that Algorithm \ref{algo:PRPF-Richardson} sets $d_{k} = {\Lambda_B^2 T^2}{s_{k}}$. Thus, $|Y^{(i)}| = O(\|\vec{b}\|_1)$ as $s_{k'} \leq 1$ and we assume that $p,\Upsilon, a_{\max}=O(1)$. Our random variable has bounded size, and the stated sample complexity is given by Hoeffding's inequality, where we have used the fact from Lemma \ref{lemma:well-conditioned-Richardson} that $\|\vec{b}\|_1 = O(\log\log(1/\varepsilon_R)) = O(\log\log(1/\varepsilon))$.

\textit{Gate complexity.}  
For a given Trotter step size $s_kT$ chosen in one iteration of the Randomized Extrapolation Trotterization, the gate depth per Trotter step implemented is $O(\Gamma_A + d_k) = O(\Gamma_A + \Lambda_B^2 T^2 s_k)$, and thus the total gate depth is $O((\Gamma_A + d_k)/s_k) = O(\Gamma_A r_k + \Lambda_B^2T^2)$. The worst-case gate depth we state comes from substituting in our expression for $r_{\max}$. 

\textit{Classical preprocessing.} We have $O(m\Upsilon)$ probability distributions to precompute, each of dimension $\polylog(1/\varepsilon)$. The stated complexity follows by noting that $m=\log(1/\varepsilon)$ is sufficient for the well-conditioned extrapolation scheme of Lemma \ref{lemma:well-conditioned-Richardson}.
\end{proof}

We observe that the sample complexity and classical preprocessing cost are (asymptotically) unchanged from our core Randomized Extrapolation Trotterization.

Let us now focus on the new extrapolated commutator factor $\tilde{\lambda}_{\text{comm}}$. Note that in this setting we are explicitly presuming that we have chosen a distinct matrix decomposition and accompanying product formula $\mathcal{S}_p(t)$ we chose has a different commutator factor $\tilde{\alpha}_{\text{comm}}$, compared to a conventional choice (denoted as $\PC_p$ and $\alpha_{\text{comm}}$ elsewhere in this manuscript) where each element of the matrix decomposition is efficiently evolvable. This introduces a new $\tilde{\lambda}_{\text{comm}}$ which characterizes the complexity of our doubly Randomized Extrapolation Trotterization. For certain problems it may be that $\tilde{\alpha}_{\text{comm}}$ can be efficiently and strongly bounded, and indeed one may be able to make a smart choice of $H_B$ to exploit such a bound and obtain refined complexity guarantees. Alternatively, the fact that $H_B$ appears in $\tilde{\alpha}_{\text{comm}}$ makes it different to analyze. In the following lemma we demonstrate a worst-case bound which is always available, relating $\tilde{\lambda}_{\text{comm}}$ to the extrapolated commutator factor of a conventional decomposition $\tilde{\lambda}_{\text{comm}}$.

\begin{lemma}[Relation between extrapolated commutator factors]\label{lem:PRPF-lcomm}
Consider a partition $H = H_A + H_B$, where $H_A = \sum_{\gamma \in S_A} H_{\gamma}$ and $H_B = \sum_{\gamma \in S_B} H_{\gamma}$ is any decomposition of $H_B$. Let $\acomm$ be the commutator factor corresponding to decomposition $H = \sum_{{\gamma} \in S_A \cup S_B} H_{\gamma}$, and let $\tilde{\alpha}_{\text{comm}}$ be the commutator factor corresponding to decomposition $H = \sum_{i \in S_A \cup \{B\}} H_{\gamma}$. Finally, let $\lcomm$ and $\tilde{\lambda}_{\text{comm}}$ be the corresponding respective extrapolated commutator factors constructed as in Eq.~\eqref{eq:lcomm}. We have
\begin{equation}
\tilde{\lambda}_{\text{comm}} \leq \lambda_{\text{comm}}.
\end{equation}
\end{lemma}

\begin{proof}
We compare the commutator factors $\tilde{\alpha}_{\text{comm}}^{(j)}$ and $\alpha_{\text{comm}}^{(j)}$ for fixed order $j$. We have
\begin{equation}
\tilde{\alpha}_{\text{comm}}^{(j)} := \sum_{\gamma_1, \dots, \gamma_j \in S_A \cup B} \left\| [H_{\gamma_1}, , \dots, H_{\gamma_j}]  \right\|.
\end{equation}
We now consider any arbitrary term in the above sum that contains $H_B$. We have 
\begin{align}
    \left\| [H_{\gamma_1}, , \dots, H_B, \dots, H_{\gamma_j}]  \right\| &=  \Big\| \big[H_{\gamma_1}, , \dots, \sum_{\gamma' \in S_B} H_{\gamma'}, \dots,  H_{\gamma_j}\big]  \Big\| \\
    &= \Big\| \sum_{\gamma' \in S_B} \big[H_{\gamma_1}, , \dots,  H_{\gamma'}, \dots,  H_{\gamma_j}\big]  \Big\| \\
    &\leq \sum_{\gamma' \in S_B} \Big\|  \big[H_{\gamma_1}, , \dots,  H_{\gamma'}, \dots,  H_{\gamma_j}\big]  \Big\|\,,
\end{align}
where the second equality is due to the bilinearity of the commutator, and the inequality is an application of the triangle inequality.

From this, all instances of $H_B$ can be unfurled and we have 
\begin{equation}
    \tilde{\alpha}_{\text{comm}}^{(j)} \leq \sum_{\gamma_1, \dots, \gamma_j \in S_A \cup S_B} \left\| [H_{\gamma_1}, , \dots, H_{\gamma_j}]  \right\| = \alpha_{\text{comm}}^{(j)}.
\end{equation}
As $\tilde{\lambda}_{\text{comm}}$ is an increasing function in all commutator factors $\tilde{\alpha}_{\text{comm}}^{(j)}$, it follows that $
\tilde{\lambda}_{\text{comm}} \leq \lambda_{\text{comm}}$ as required.
\end{proof}

We remark that the above lemma is agnostic to the choice of decomposition of $H_B$. Thus, whilst the Pauli decomposition in Eq.~\eqref{eq:partial-randomization-decomposition} is chosen for algorithmic implementation, any other decomposition can be used for analytical analysis.

We can put together the same methodology as in Algorithm \ref{algo:PRPF-Richardson} and analysis in Theorem \ref{theorem:extrapolated-partial-randomization} and Lemma \ref{lem:PRPF-lcomm} for all of Tasks \hyperref[eq:task-1]{1-2} to obtain refined complexities.

\begin{theorem}[Doubly Randomized Extrapolation Trotterization]
For a matrix decomposed as in Eq.~\eqref{eq:partial-randomization-decomposition},  Algorithms \ref{algo:overlap}, \ref{algo:observable}, \ref{algo:distribution} all succeed at their stated tasks with complexities given in Theorems \ref{thm:overlap}, \ref{thm:overlap-with-observable}, \ref{thm:distribution-recovery} respectively with the replacement to the gate depth
\begin{equation}
    L(\lcomm T(\varepsilon))^{1+1/p}\log(1/\varepsilon') \rightarrow  \Gamma_A  (\lcomm T(\varepsilon))^{1+1/p}\log(1/\varepsilon') + \Lambda_B^2 T^2(\varepsilon)
\end{equation}
for each respective prescribed $\varepsilon, \varepsilon'$ in each theorem.
\end{theorem}

\section{Application: ground state energy estimation}\label{sec:phase-estimation}

In this section, we explore a method for ground state energy estimation (sometimes also referred to as phase estimation) that involves computing the cumulative distribution function (CDF) of the eigenvalue spectrum of a Hamiltonian \( H \). This core idea was first introduced by Lin and Tong in \cite{lin2022HeisenbergLimited} and studied further in \cite{wan2021RandPhaseEst}. We will see that this procedure is an instance of Task \hyperref[eq:task-1]{1(i)}. We start this section by providing a self-contained exposition of the CDF approach to ground state energy estimation. Readers interested in skipping ahead to algorithm complexities can go to Theorem \ref{thm:phase-estimation}, and we recall a complexity comparison is given in Table \ref{tab:phase-estimation}.

Our primary quantity of interest, the eigenvalue CDF associated with \(H\), can be formulated using the Heaviside step function which we approximate via a Fourier series expansion over the interval \([- \pi, \pi]\). Consequently, it is necessary to rescale the Hamiltonian such that its spectrum lies within this interval. We define the normalized Hamiltonian as \(\hat{H} := \kappa H\), where the scaling factor \(\kappa\) will be specified later. Throughout this section, we assume the existence of a known upper bound \(\hat{K}\) on the spectral norm of \(H\), meaning \(\hat{K} \geq \|H\|\). The CDF can be expressed as
\begin{equation}\label{eq:CDF}
C(x) = (\hat{p} * \Theta)(x) = \operatorname{Tr}[\rho \, \Theta(x I - \kappa H)] = \sum_{x \geq \kappa E_i} \operatorname{Tr}[\rho \Pi_i] \,,
\end{equation}
where \(\Theta(x)\) is the Heaviside step function, \(\hat{p}(x) = \sum_i \operatorname{Tr}[\rho \Pi_i] \, \delta(x - \kappa E_i)\) represents the probability density function corresponding to \(\rho\) and the normalized Hamiltonian \(\hat{H}\), and \( * \) denotes convolution. Here, \(\Pi_i\) is the projector onto the eigenspace associated with the \(i\)-th eigenvalue of \(H\), arranged in increasing order. We also define \(\eta = \operatorname{Tr}[\rho \Pi_0]\) as the overlap between the ground state and the ansatz state. 

When generalized to generic functions, we see that Eq.~\eqref{eq:CDF} gives a useful perspective on the state overlap (Task \hyperref[eq:task-1]{1(i)}) -- the state overlap with $f(H)$ can be thought of as a value of the convolution of the probability density function of $\rho$ with $f$ -- it can be used to extract spectral properties of $H$ with reference to a probe state $\rho$. Further, when we consider Fourier series approximations, we will see that the offset of the spectrum by $x$ can be inconsequentially treated as a phase shift to be kept track of classically.


\subsection{Heaviside approximation and approximate CDF construction}

In order to enable us to use Algorithm \ref{algo:overlap}, we first present a useful Fourier series approximation to the Heaviside function valid in the domain \([- \pi, \pi]\).

\begin{lemma}[Heaviside function  Fourier approximation (\cite{wan2021RandPhaseEst}, Lemma 1)]\label{lemma-heaviside-fourier}
    There exists a Fourier series $\widetilde{\Theta}(x)=\sum_{j \in S} \widetilde{\Theta}_j e^{i j x}$ with $S:=\{0\} \cup\{ \pm(2 j+1)\}_{j=0}^d$, maximum ``time'' parameter $d={O}(u^{-1} \log \left(\varepsilon_F^{-1}\right))$, Fourier approximation parameter $\varepsilon_{F}$, and $u \in [0,\pi/2]$ serves as a resolution parameter. This Fourier series satisfies
\begin{enumerate}
    \item Approximation guarantee: $|\widetilde{\Theta}(x)-\Theta(x)| \leq \varepsilon_F \quad \forall x \in[-\pi+u,-u] \cup[u, \pi-u]$,  
    \item Bounded maximum value: $-\varepsilon_F \leq\widetilde{\Theta}(x) \leq 1+\varepsilon_F \quad \forall x \in \mathbb{R}$,
    \item Bounded size of coefficients: $\|\widetilde{\Theta}\|_1:=\sum_{j \in S_1} |\widetilde{\Theta}_j|=O(\log d)$\,.
\end{enumerate}
\end{lemma}
\noindent We use this Fourier series to construct an approximate CDF which we denote as
\begin{align}
    \widetilde{C}(x) = (\hat{p} * \widetilde{\Theta} )(x) = \sum_{j \in S} \widetilde{\Theta}_j \int_{-\pi}^\pi \hat{p}(y) e^{i j(x-y)} \mathrm{d}y = \sum_{j \in S} \widetilde{\Theta}_j e^{i j x} \operatorname{Tr}\left[\rho e^{-i j \kappa H}\right] =   \tr[\rho \widetilde{\Theta}(xI-\kappa H)] \,.
\end{align}

The following lemma ensures that this Fourier approximation gives a reasonable approximation to our desired cumulative distribution function. 
\begin{lemma}[Approximation to CDF from Fourier series, adapted from \cite{wan2021RandPhaseEst} Proposition 12]\label{lemma:CDF-approximation}
Take $\widetilde{\Theta}(x)$ from Lemma \ref{lemma-heaviside-fourier}, and set $\kappa = \frac{\pi-u}{2\hat{K}}$
. Then, the quantity $\widetilde{C}(x) = \tr[\rho \widetilde{\Theta}(xI-\kappa H)]$ satisfies
\begin{equation}\label{eq:aCDF-to-CDF}
    C(x-u)-\varepsilon_F \leq \widetilde{C}(x) \leq C(x+u)+\varepsilon_F\,.
\end{equation}    
\end{lemma}
We observe that the probability density function is defined over an interval contained within
$[-\kappa \hat{K}, \kappa \hat{K}]$, which for the sake of the above lemma is $[-\frac{1}{2}(\pi - u), \frac{1}{2}(\pi - u)]$. This choice of support is necessary to guarantee the error bounds established in Eq.~\eqref{eq:aCDF-to-CDF}. From this point onward, we fix the normalization factor as $\kappa = \frac{\pi - u}{2\hat{K}}$ as above to ensure that this lemma remains applicable throughout our analysis.

Now, we demonstrate that the ground state energy can be characterized by a condition on the approximate CDF, allowing us to use the approximate CDF as a tool to determine the ground state energy. The key idea is that the exact CDF $C(x)$ marks the ground state energy by the smallest value of $x$ for which the CDF is non-zero (at the ground state energy, the CDF takes value $\geq \eta$). The story is similar with the approximate CDF, within the resolution of the approximation. The contents of the following lemma are essentially identical to the analysis in \cite{lin2022HeisenbergLimited}.

\begin{lemma}\label{lemma:CDF-from-aCDF}
    Suppose we have an algorithm to approximate $\widetilde{C}(x)$ to precision $\eta/8$ for some $0< u < \pi/2$, $\varepsilon_F = \eta/8$ and for any $x \in [-\pi/2,\pi/2]$. Then we may determine whether    \begin{equation}\label{eq:CDF-criterion}
        C(x+u)>\eta / 2 \quad \text{or}\quad C(x-u)<\eta \,.
    \end{equation} 
    Suppose further that we locate an $x^*$ that simultaneously satisfies both conditions. Then, we have that 
    \begin{equation}
        \left| x^*/\kappa - E_0 \right| \leq u/\kappa \,,
    \end{equation}
    i.e., we have an additive estimate to the ground state energy $E_0$.
\end{lemma}

\begin{proof}
We use the algorithm to determine whether $\widetilde{C}(x) > 3\eta/4$ or not. Within the specified precision, we can guarantee either
\begin{equation}
    \widetilde{C}(x)>(5 / 8) \eta \quad \textit{ or } \quad \widetilde{C}(x)<(7 / 8) \eta\,,
\end{equation}
which respectively implies either $C(x+u)>\eta / 2 $ {or} $C(x-u)<\eta$ respectively by use of Lemma \ref{lemma:CDF-approximation}, and we have satisfied the first claim. Now we suppose both conditions are simultaneously satisfied for some $x^*$. We recall that $\eta \leq \tr[\Pi_0 \rho]$ lower bounds the ground space overlap. We additionally note that the exact CDF $C(x)$ cannot take values in $(0,\tr[\Pi_0 \rho]) \supseteq (0,\eta)$ and so we have that 
\begin{align}
    C(x^{*}+u)\geq \eta\,,& \quad C(x^{*}-u)=0\,, \\
    \implies 
    x^* + u \geq \kappa E_0\,,& \quad x^*- u < \kappa E_0\,.
\end{align}
Thus, $|x^* - \kappa E_0| \leq u$.
\end{proof}

\subsection{Binary search for the ground state energy}

Next we present an algorithm and lemma showing that the search algorithm can locate \( x^* \) using only a logarithmic number of queries to the approximate CDF \( \widetilde{C}(x) \).

\begin{lemma}[Binary search with approximate CDF]\label{lemma:binary-search}
    Suppose that we have an algorithm ${A}(x,\widetilde{u},\widetilde{\varepsilon},\widetilde{\delta})$ which evaluates $\widetilde{C}(x)$ for any $x$ to additive error $\widetilde{\varepsilon}$ and success probability at least $(1-\widetilde{\delta})$, for $\varepsilon_F = \eta/8$ and some resolution parameter $\widetilde{u} \in [0,\pi/2]$. Then, the ground state energy can be found to additive precision $u/\kappa$ and success probability at least $(1-\delta)$ by running ${A}(x,0.9u,\eta/8,\delta/L)$ at $L=O(\log(1/u))$ different values of $x$.
\end{lemma}

\begin{proof}
    The contents of \cite[Section 5]{lin2022HeisenbergLimited} prove this statement, but we present it here for completeness. We see from Lemma \ref{lemma:CDF-from-aCDF} that setting $\widetilde{\varepsilon} = \eta/8$ allows us to determine the criterion \eqref{eq:CDF-criterion} (with some success probability to be later addressed and some resolution parameter). We first show how to use this to search for $x^*$, before discussing how to tune the success probability.
    The search procedure operates by finding successively tighter upper and lower bounds to $x^*$. Set
    \begin{equation}
        x_{0,0} = -\pi/2, \quad x_{1,0} = \pi/2\,,
    \end{equation}
    where we have that $C(x_{0,0})<\eta$ and $C(x_{1,0})>\eta / 2$. We now specify an update rule to generate $(C(x_{0,\ell}), C(x_{1,\ell}))$ that have decreasing separation for increasing $\ell$, but still satisfy $C(x_{0,\ell})<\eta$ and $C(x_{1,\ell})>\eta / 2$ (recall that as $C(x)$ cannot take values in $(0,\eta)$,  we have $x_{0,\ell} \leq x^*\leq x_{1,\ell}$). For all $\ell\geq 0$ until termination, first construct $x_{\ell} = (x_{0,\ell} + x_{1,\ell})/2$. Second, run ${A}(x_{\ell},0.9u,\eta/8,\widetilde{\delta})$ (note the prefactor $0.9$ for the resolution parameter can be arbitrarily chosen to any number smaller than $1$). Due to \eqref{eq:CDF-criterion} this determines whether (i): $C(x_{\ell} + 0.9u)> \eta/2$ or (ii): $C(x_{\ell} - 0.9u)< \eta$. We now apply the following update rule based on the outcome
    \begin{align}
        \text{(i)}&: \qquad x_{0,\ell+1} = x_{0,\ell}\,,\qquad x_{1,\ell+1} = x_{\ell} + 0.9u\,, \\
        \text{(ii)}&: \qquad x_{0,\ell+1} = x_{\ell}- 0.9u\,,\qquad x_{1,\ell+1} = x_{1,\ell} \,, 
    \end{align}
    and it is simple to check that the conditions $C(x_{0,\ell})<\eta$ and $C(x_{1,\ell})>\eta / 2$ are satisfied, and that the separation $x_{1,\ell}-x_{0,\ell}$ is decreasing with increasing $\ell$ so long as $x_{1,\ell}-x_{0,\ell} \geq 1.8u$. 
    We terminate the procedure  at step $L$ when $x_{1,\ell}-x_{0,\ell} \leq 2u$, as this implies that $|x_L-\kappa E_0|\leq u$. Then, $x_L/\kappa$ satisfies our desired approximation of the ground state energy $E_0$. One can check that the separation satisfies
    \begin{equation}
        x_{1,\ell}-x_{0,\ell} = \frac{\pi - 1.8u}{2^{\ell}} + 1.8u\,,
    \end{equation}
    and thus $L=O(\log(1/u))$ evaluations of ${A}(x_{\ell},0.9u,\eta/8,\widetilde{\delta})$ are sufficient.

    Finally, let us determine a sufficient success probability $\widetilde{\delta}$ for the algorithm that evaluates $\widetilde{C}(x)$. If each use of ${A}(x_{\ell},0.9u,\eta/8,\widetilde{\delta})$ fails with probability at most $\widetilde{\delta}$, then in $L$ uses the failure probability is at most $L\widetilde{\delta}$. Our stated claim follows by requiring that this equals our desired overall failure probability $L\widetilde{\delta} = \delta$.
\end{proof}

Our last piece to construct a full phase estimation algorithm is to specify the algorithm ${A}(x,\widetilde{u},\widetilde{\varepsilon},\widetilde{\delta})$ to evaluate $\widetilde{C}(x)$. We observe that $\widetilde{C}(x)$ can be estimated statistically via our Algorithm~\ref{algo:overlap}, where the relevant function is the Heaviside function $\Theta$. For this algorithm we use the LKW extrapolation strategy of Lemma \ref{lemma:well-conditioned-Richardson}, and thus are able to import the analysis of Theorem \ref{thm:overlap}. Whilst the translation to the Heaviside function is direct, we display the full algorithm pseudocode here for clarity.

\begin{algorithm}[H]
\caption{\ Estimation of approximate CDF via Fourier–Heaviside Expansion}
\label{algo:approx-CDF}
\begin{algorithmic}[1]
\State\textbf{Classical preprocessing:} Compute distribution \( \left\{ \frac{|\widetilde{\Theta}_j b_k|}{\mathcal{G}} \right\}_{j,k} \) where \( \mathcal{G} = \sum_{j,k} |\widetilde{\Theta}_j b_k| \)
\For{$i = 1$ to $\cC_{\mathrm{sample}}$}
    \State Sample \( (j', k') \) from distribution, \(T \gets \kappa \cdot j' \)
    \State Prepare Hadamard tests corresponding to:
    \[\operatorname{Re}[\operatorname{Tr}( \rho \PC^{1/s_{k'}}(s_{k'} T))]\,\quad \operatorname{Im}[\operatorname{Tr}( \rho \PC^{1/s_{k'}}(s_{k'} T))] \]
    \State $(X^{(k')}_{\text{Re}}, X^{(k')}_{\text{Im}}) \gets$ one measurement statistic from each circuit; set \( Z^{(i)} = X^{(k')}_{\text{Re}} + i X^{(k')}_{\text{Im}} \)
    \State \( C^{(i)} \gets \|\vec{b}\|_1 \cdot \|\widetilde{\Theta}\|_1 \cdot \operatorname{sign}(b_{k'} \widetilde{\Theta}_{j'}) \cdot e^{i j' x} \cdot Z^{(i)} \)
\EndFor
\State \textbf{Return:} Mean of $C^{(i)}$  over  $\cC_{\mathrm{sample}}$ values.
\end{algorithmic}
\end{algorithm}

As before with Algorithm~\ref{algo:overlap}, the repeating subroutine of this algorithm is structured into two main parts: a quantum component (Step 4) followed by classical post-processing (Step 6). While we may be tempted to rerun Algorithm \ref{algo:approx-CDF} in its entirety to instantiate ${A}(x, \widetilde{u}, \widetilde{\varepsilon}, \widetilde{\delta})$ at each desired value of $x$, it is useful to note that the quantum step does not need to be repeated for each $x$. Instead, we can run the quantum step once, record its outcomes, and then apply different classical post-processing (Step 6) corresponding to each desired $x$. As a result, obtaining outputs for $L$ different values of $x$ requires only a single round of quantum measurements, with all additional computation handled classically. Let us now analyze the complexity of running Algorithm~\ref{algo:overlap} for our purposes.

\begin{lemma}[Approximate CDF estimation]\label{lemma:aCDF-algorithm}
Algorithm~\ref{algo:approx-CDF} implements the procedure $A(x,\widetilde{u},\widetilde{\varepsilon},\widetilde{\delta})$, preparing a $\widetilde{\varepsilon}$-additive approximation to the approximate CDF $\widetilde{C}(x)$, with Fourier parameters $u = \widetilde{u}$ and $\varepsilon_F = \eta/8$. The algorithm has success probability at least $1 - \widetilde{\delta}$ using the following resources:
\begin{align}
    \cC_{\text{gate}} &= O\!\left( \Gamma \left(a_{\max} \Upsilon \lambda_{\mathrm{comm}} \kappa \widetilde{u}^{-1} \log(\eta^{-1})\right)^{1 + 1/p} \cdot \log\left( \widetilde{\varepsilon}^{-1} \log\big( \widetilde{u}^{-1} \log(\eta^{-1}) \big) \right) \right), \\
    \cC_{\text{sample}} &= O\!\left( \frac{ \left( \log \log(\widetilde{\varepsilon}^{-1}) \cdot \log(\widetilde{u}^{-1} \log(\eta^{-1})) \right)^2 }{ \widetilde{\varepsilon}^2 } \cdot \log\left( \frac{1}{\widetilde{\delta}} \right) \right)\,,\\
    \cC_{\text{pre}} &= O\!\left(  \widetilde{u}^{-1} \log(\eta^{-1}) \log\left(\frac{\log(\widetilde{u}^{-1} \log(1/\eta))}{\widetilde{\varepsilon}} \right)\right)\,.
\end{align}
where $\cC_{\text{sample}}$ denotes the number of quantum circuit repetitions, $\cC_{\text{gate}}$ is the maximum depth of any individual circuit, and $\cC_{\text{pre}}$ is the classical preprocessing cost.
\end{lemma}

\begin{proof}
Algorithm~\ref{algo:approx-CDF} is a direct instance of Algorithm~\ref{algo:overlap} applied to the approximate Heaviside function $\widetilde{\Theta}$. We call Algorithm~\ref{algo:overlap} using the following parameters:

\begin{itemize}
  \item Heaviside Fourier approximation parameters $u=\widetilde{u}$, $\varepsilon_F = \eta/8$, and $d = O(\widetilde{u}^{-1} \log(1/\eta))$
  \item Target precision $\widetilde{\varepsilon}$
  \item Failure probability upper bound $\widetilde{\delta}$
\end{itemize}
Taking the notation of Section \ref{sec:fourier-algos}, the size of the Fourier coefficients, maximum time parameter, and number of Fourier modes are bounded as (taking Fourier approximation error $\varepsilon_F = \eta/8$)
\begin{align}
&c = \|\widetilde{\Theta}\|_1 = O(\log(\widetilde{u}^{-1} \log(1/\eta)))\,,\\
&T = \kappa d = O(\kappa \widetilde{u}^{-1} \log(1/\eta))\,,  \\
&K = 4d+3 = O(\widetilde{u}^{-1} \log(1/\eta))\,.
\end{align}
These can be directly substituted into Theorem~\ref{thm:overlap} to yield the stated results.
\end{proof}

Now we have established the resource costs of the approximate CDF estimation procedure implemented by Algorithm~\ref{algo:approx-CDF}, we can write down the complexity of our full algorithm for ground state energy estimation.

\begin{theorem}[Ground state energy estimation]\label{thm:phase-estimation} Take any $n$-qubit Hamiltonian $H$ of the form in Defintion \ref{def:matrix} with spectral norm upper bound $\|H\|\leq \hat{K}$. Take any even $p$-th order product formula, and the LKW extrapolation scheme of Lemma \ref{lemma:well-conditioned-Richardson}. Suppose we have a trial state $\rho$ who has overlap with the ground energy space $\eta = \tr[\rho \Pi_0]$.
We give an algorithm to estimate the ground state energy of $H$ to additive precision $\varepsilon$ and with success probability at least $(1-\delta)$ using resources 
\begin{itemize}
    \item \textbf{Gate depth (per sample):}
    \begin{equation}
        \cC_{\text{gate}} = \widetilde{O}\!\left(\Gamma \Upsilon \left(\frac{a_{\max } \Upsilon \lcomm }{\varepsilon} \right)^{1+1 / p} \Big(\log \frac{1}{\eta}\Big)^{2+1/p}\right)\,,
    \end{equation}
    
    \item \textbf{Sample complexity:}
    \begin{equation}
        \cC_{\text{sample}} = \widetilde{O}\!\left(\frac{1}{\eta^2} \Big(\log\Big(\frac{\hat{K}}{\varepsilon}\Big)\Big)^2 \log\Big(\frac{1}{\delta}\Big)  \right)\,,
    \end{equation}

    \item \textbf{Classical preprocessing time:}
    \begin{equation}
        \cC_{\text{pre}} = \widetilde{O}\!\left( \frac{\hat{K}}{\varepsilon} \log^2\frac{1}{\eta}\right)\,,
    \end{equation}
\end{itemize}
where $\cC_{\textit{sample}}$ is the number of measurement samples required from $(n+1)$-qubit circuits consisting on Hadamard tests of the product formula, each with gate depth at most $\cC_{\textit{gate}}$.
\end{theorem}

\begin{proof}
    We use the search algorithm outlined in Lemma \ref{lemma:binary-search} which makes calls to an algorithm ${A}(x,\widetilde{u},\widetilde{\varepsilon},\widetilde{\delta})$ which approximately returns the approximate CDF, and use Lemma \ref{lemma:aCDF-algorithm} to instantiate ${A}(x,\widetilde{u},\widetilde{\varepsilon},\widetilde{\delta})$. Set $\widetilde{\varepsilon} = \eta/8$, $\widetilde{u}=0.9u$, $\widetilde{\delta}=\delta/L$ with $L=O(\log(1/u))$, as discussed in Lemma \ref{lemma:binary-search}.  Further, recall from Lemma \ref{lemma:binary-search} that to reach additive precision $\varepsilon$ in the ground state energy we set $u=\kappa \varepsilon$, and that our CDF is defined such that  $\kappa = O(\hat{K}^{-1})$, where $\hat{K}$ is an upper bound on the spectral norm of $H$ (see Lemma \ref{lemma:CDF-approximation}). Substituting all these parameters into Lemma \ref{lemma:aCDF-algorithm} yields 
    \begin{align}
    \cC_{\text{sample}} &= O\!\left(\frac{\big(\log\log(\eta^{-1})\cdot \log(\hat{K} \varepsilon^{-1}\log(\eta^{-1})\big)^2}{\eta^2} \log\Big(\frac{\log(\hat{K}\varepsilon^{-1})}{\delta}\Big)  \right)\,, \\\cC_{\text{gate}} &= O\!\left(\Gamma \left(a_{\max } \Upsilon \lcomm \varepsilon^{-1}\log(\eta^{-1}))\right)^{(1+1 / p)} \log (\eta^{-1}\log(\hat{K} \varepsilon^{-1}\log(\eta^{-1}))\right)\,,\\
    \cC_{\text{pre}} &= O\!\left( \log\big( \hat{K} \varepsilon^{-1} \log(\eta^{-1}) \big) \log\left(\frac{\log(\hat{K} \varepsilon^{-1} \log(1/\eta))}{\eta} \right)\right)
\end{align}
and we simplify the expressions in $\widetilde{O}$-notation in the theorem statement.
\end{proof}

\section{Application: Green's functions}\label{sec:greens}

In this section we apply our algorithm framework to the estimation of Green's functions in the context of many-body physics. Green's functions capture key dynamical and spectral properties of quantum systems, and their accurate estimation provides insights into physical properties such as particle propagation, kinetic energy, and spectral densities. 

We define the advanced and retarded Green's function in the frequency domain (denoted as $G^{(+)}(\omega)$ and $G^{(-)}(\omega)$ respectively) as the matrix-valued functions with matrix elements
\begin{align}
    G_{ij}^{(+)}(\omega) &:= \bra{E_0}\hat{a}_i \left(\omega - (H-E_0) + i\eta_{\textrm{broad}} \right)^{-1} \hat{a}^{\dag}_j \ket{E_0}\,,\label{eq:greens1} \\
    G_{ij}^{(-)}(\omega) &:= \bra{E_0}\hat{a}^{\dag}_i \left(\omega - (H-E_0) - i\eta_{\textrm{broad}} \right)^{-1} \hat{a}_j \ket{E_0}\,, \label{eq:greens2}
\end{align}
where $E_0$ is the ground state energy of $H$, $\eta_{\textrm{broad}}$ is a broadening parameter that determines the resolution of the Green's function, and $\hat{a}^{\dag}_i,\hat{a}_i$ are Fermionic single-particle creation and annihilation operators. At its core, we see that there is a matrix function, sometimes called the resolvent operator 
\begin{equation}
    R(\omega \pm i \eta_{\text{broad}}, {H}) = (\omega \pm i\eta_{\textrm{broad}} - \hat{H})^{-1}\,,
\end{equation}
which amounts to a shifted matrix inverse (we denote $\hat{H} = (H-E_0)$). In the rest of this section we will assume the ground state (or an approximation thereof) is given to us for simplicity of exposition -- this can be instantiated via a separate algorithm, or even within our matrix function framework by considering a function that (approximately) projects some input state to the ground state. The creation and annihilation operators can be mapped to qubit systems using standard mappings such as the Jordan-Wigner transformation, and lead to $O(1)$-depth unitaries. Thus, we see that  estimation of the quantities in Eqs.~\eqref{eq:greens1} and \eqref{eq:greens2} correspond exactly to Task \hyperref[eq:task-1]{1(i)}.

In order to estimate Green's functions, we need to approximate the resolvent operator with a Fourier series. Such approximations have conveniently already been found in prior art, and we quote one such approximation here, originally used in the context of LCU algorithms \cite{keen2021quantumGSGreens}. We consider a Fourier series of the form 

\begin{equation}
    h(\omega + i \eta_{\text{broad}}, \hat{H}) = -i \sum_{k=0}^{N_c} \Delta t \, e^{i(\omega + i \eta_{\text{broad}} - \hat{H}) k \Delta t}.
\end{equation}

\begin{lemma}[Fourier approximation of the resolvent](\cite[Theorem 3]{keen2021quantumGSGreens}) 
\label{lemma:resolvent_lcu}
Let $\hat{H}$ have spectrum normalized $\lambda(\hat{H}) \subseteq [0,1]$. For any $\omega \in \sigma(\hat{H})$ and broadening parameter $\eta_{\text{broad}}>0$ we have
\begin{equation}
\| R(\omega + i \eta_{\text{broad}}, \hat{H}) - h(\omega + i \eta_{\text{broad}}, \hat{H}) \| \le \varepsilon
\end{equation}
provided that
\begin{equation}
N_c = O\Big(\frac{1}{\eta_{\text{broad}} \varepsilon} \log \frac{1}{\eta_{\text{broad}} \varepsilon} \Big), \qquad
\Delta t = \min\Big\{\frac{\varepsilon}{2}, \frac{3}{\|\hat{H}\|} \Big\} = O(\varepsilon).
\end{equation}
\end{lemma}

We are now able to write down an algorithm for estimating Green's functions using Randomized Extrapolation Trotterization.


\begin{theorem}
[Green's functions]
\label{thm:resolvent_estimate}
Consider a Hamiltonian of the form in Def.~\ref{def:matrix} with spectrum in $[0,1]$. Assuming access to the ground state $\ket{E_0}$, the matrix entries to the retarded and advanced Green's functions in Eq.~\eqref{eq:greens1} and \eqref{eq:greens2} can be estimated to additive error $\varepsilon$ with success probability at least $(1-\delta)$ using Algorithm~\ref{algo:overlap} with the following complexities:
\begin{itemize}
    \item \textbf{Gate depth (per sample):}
    \begin{equation}
        \cC_{\text{gate}} = \widetilde{O}\Bigg(\Gamma \Upsilon \cdot \Big( a_{\text{max}} \, \Upsilon \, \lambda_{\text{comm}} \cdot \frac{1}{\eta_{\text{broad}}} \Big)^{1 + \frac{1}{p}} \cdot \left(\log \frac{1}{ \varepsilon}\right)^{2+\frac{1}{p}} \Bigg)\,,\,,
    \end{equation}
    
    \item \textbf{Sample complexity:}
    \begin{equation}
        \cC_{\text{sample}} = \widetilde{O}\left(\frac{1}{ \eta_{\text{broad}}^2 \varepsilon^2} \cdot \log \frac{1}{\delta}\right)\,,
    \end{equation}

    \item \textbf{Classical preprocessing time:}
    \begin{equation}
        \cC_{\text{pre}} = \widetilde{O}\left( \frac{1}{ \eta_{\text{broad}}^2 \varepsilon^2} \right)\,.
    \end{equation}
\end{itemize}
Each quantum circuit consists of an $n+1$ qubit Hadamard test on the product formula.
\end{theorem}

\begin{proof}
We use the Fourier approximation of Lemma~\ref{lemma:resolvent_lcu}. In the notation of Section \ref{sec:fourier-algos}, the Fourier parameters for Fourier error $\varepsilon_F$ are
\begin{align}
    &c(\varepsilon_F) = \sum_{k=0}^{N_c} \Delta t = O(N_c \Delta t) = O\Big(\frac{1}{\eta_{\text{broad}}} \log \frac{1}{\eta_{\text{broad}} \varepsilon_F} \Big)\,,\\
    &t_{\max}(\varepsilon_F) = N_c \Delta t = O\Big(\frac{1}{\eta_{\text{broad}}} \log \frac{1}{\eta_{\text{broad}} \varepsilon_F} \Big)\,,\\
    &K(\varepsilon_F) = N_c = O\Big(\frac{1}{\eta_{\text{broad}} \varepsilon_F} \log \frac{1}{\eta_{\text{broad}} \varepsilon_F} \Big)\,.
\end{align}
We can very simply substitute these parameters into Theorem~\ref{thm:overlap} to yield:
\begin{align}
\cC_{\text{sample}} &= O\Bigg( \frac{1}{\varepsilon^2}
    \left( \frac{1}{\eta_{\text{broad}}} \log \frac{1}{\eta_{\text{broad}} \varepsilon} \right)^2 \left(\log \log 
    \left( \frac{1}{\eta_{\text{broad}} \varepsilon} \log \frac{1}{\eta_{\text{broad}} \varepsilon}\right) \right)^2
    \cdot \log \frac{1}{\delta}
\Bigg)\,,\\
\cC_{\text{gate}} &= O\Bigg( \Gamma \Upsilon \cdot
    \Big( a_{\text{max}} \, \Upsilon \, \lambda_{\text{comm}} \cdot \frac{1}{\eta_{\text{broad}}} \log \frac{1}{\eta_{\text{broad}} \varepsilon} \Big)^{1 + \frac{1}{p}}   \cdot  \log \left( \frac{1}{\eta_{\text{broad}} \varepsilon} \log \frac{1}{\eta_{\text{broad}} \varepsilon} \right)
\Bigg)\,, \\
\cC_{\text{pre}} &= O\left( \frac{1}{\eta_{\text{broad}} \varepsilon} \log \frac{1}{\eta_{\text{broad}} \varepsilon} \log\Big(\frac{1}{\eta_{\text{broad}} \varepsilon} \log \frac{1}{\eta_{\text{broad}} \varepsilon} \Big) \right)\,,
\end{align}
and we simplify these expressions in $\widetilde{O}$ notation in the theorem statement.
\end{proof}

\section{Application: Time-Evolved States}\label{sec:time-evol-distribution}

Theorem \ref{thm:distribution-recovery} directly gives an application for probing the distributions of time-evolved states for some time of interest $t$, that is, by considering the function $f(H) = \exp(-iHt)$. Here the Fourier parameters are trivial ($c=1, T=t)$ and the result can be directly stated.

\begin{theorem}[Distribution recovery for time-evolved states]
\label{thm:distribution-time-evol}
Consider an $n$-qubit Hamiltonian of the form in Definition \ref{def:matrix} and denote its time-evolved distribution on accessible input state $\ket{\psi}$ as the vector $\vec{p}$ with entries $p_j := |\bra{j}e^{-iHt}\ket{\psi}|^2$.
We give an algorithm to return $\vec{v}$ such that $\|\vec{v} -\vec{p} \|_2 \leq \varepsilon$ with success probability at least \( (1 - \delta) \) and

\begin{itemize}
    \item \textbf{Gate depth (per sample):}
    \begin{equation}
    C_{\mathrm{gate}} = 
    O\!\left( \Gamma \Upsilon \left( a_{\max} \Upsilon \lambda_{\mathrm{comm}} \, t \right)^{1 + \frac{1}{p}} \log(1/\varepsilon) \right),
    \end{equation}

    \item \textbf{Sample complexity:}
    \begin{equation}
    C_{\mathrm{sample}} = O\!\left( \frac{(\log \log (1/\varepsilon))^4}{\varepsilon^2} \cdot \log\left(\frac{1}{\delta}\right) \right),
    \end{equation}

    \item \textbf{Classical preprocessing time:}
    \begin{equation}
    C_{\mathrm{pre}} = O(\log(1/\varepsilon))\,.
    \end{equation}
\end{itemize}
Each quantum circuit consists of an $n+1$ qubit Hadamard test on the product formula.
\end{theorem}

Let us discuss now what quality of approximation is available if one uses regular Trotter formulae or any other conventional quantum algorithm. Conventionally, one asks for a unitary $U$ which approximates the unitary $e^{-iHt}$ to operator norm $\varepsilon$, which corresponds to a state whose associated probability distribution approximates the true probability distribution $\ell_1$-norm error $O(\varepsilon)$. However, in practice one must collect samples to estimate the probability distribution corresponding to $U$. This incurs statistical noise, and it is not possible to obtain an $\varepsilon$-additive $\ell_1$-norm error without incurring dimension dependence in the sample complexity. Instead, one can appeal to vector Bernstein inequalities, which state that one can incur an $\varepsilon$-additive $\ell_2$-norm error by spending $O(1/\varepsilon^2)$ shots. Thus, one can directly compare the above result to algorithms for time evolution such as regular Trotterization, where the gate depth per sample is $O\Big( \Gamma \Upsilon (\acomm^{(p+1)})^{1/p} \left( \Upsilon \, t \right)^{1 + \frac{1}{p}} (1/\varepsilon)^{1/p} \Big)$, and the number of samples is $O(\log(1/\delta)/\varepsilon^2)$.

\section{Numerical studies}
\label{sec:numerical-comparison}

In this section we carry out some numerical studies to give some insight into how well Randomized Extrapolation Trotterization (or Trotter extrapolation in general) may perform in practice.

We start by explicitly plotting the error bounds we derive for an electronic structure problem. We plot analytical bounds using both the well-conditioned LKW extrapolation schedule (as defined in \cref{lemma:well-conditioned-Richardson}), as well as extrapolation schedules found by simple brute-force search. We see that under the assumptions made, Randomized Extrapolation Trotterization starts to beat regular Trotter (in upper bounds) at a mild target precision. Moreover, extrapolation schedules generated by simple brute-force search improve on the asymptotically well-conditioned LKW extrapolation schedule. 

\subsection{Constant factor comparison and heuristic optimization}\label{sec:trotter-constants}

We analyze resource costs for the Hamiltonian electronic structure of plane-waves, for which it is known $\acomm^{(j)} = O(n^j)$, where $n$ is the number of basis functions. For simplicity, we consider the task of compiling a time signal. In order to generate resource estimates, we pull out constant factors for Randomized Extrapolation Trotterization and regular Trotterised time evolution.

\medskip

\textbf{Constant factors for regular product formulae.} The number of Trotter steps for a $p$-th-order product formula to attain additive error $\varepsilon$ in the time signal is shown in \cite[Appendix C]{childs2021TheoryTrotter} to be bounded as
\begin{equation}\label{eq:trotter-steps}
r_{\mathrm{Trotter}} \leq \left(\frac{2}{1+p}\right)^{\frac{1}{p}}
\bigl(\alpha_{\mathrm{comm}}^{(p+1)}\bigr)^{\frac{1}{p}}
(\Upsilon T)^{1+\frac{1}{p}}
\varepsilon^{-\frac{1}{p}}
.
\end{equation}
For our resource comparison, we take the Suzuki-Trotter formulae \cite{suzuki1991general}, for which the number of stages for the $2k$th order formula (for $k\geq 1$) is
\begin{equation}
    \Upsilon_{2k} = 2 \cdot 5^{k-1}\,,
\end{equation}
Consequently, the total gate depth for a Suzuki-Trotter formula is 
\begin{align}
    \Gamma
\bigl(\alpha_{\mathrm{comm}}^{(p+1)}\bigr)^{\frac{1}{p}}
T^{1+\frac{1}{p}}
\varepsilon^{-\frac{1}{p}} \quad &\text{for $p=1$\,,}\\
\Gamma \left(\frac{2}{1+p}\right)^{\frac{1}{p}} \cdot \left( 2 \cdot 5^{k-1} \right)^{2+\frac{1}{p}}
\bigl(\alpha_{\mathrm{comm}}^{(p+1)}\bigr)^{\frac{1}{p}}
T^{1+\frac{1}{p}}
\varepsilon^{-\frac{1}{p}} \quad &\text{for $p=2k$; $k \geq 1$\,.}
\end{align}

\medskip

\textbf{Refined condition number.} We make a further refinement to the bounds in Lemmas \ref{lem:generic-Richardson-error-2} and \ref{lem:specialized-Richardson-error}. We recall in both derivations we  amplify the error series remainder by the full extrapolation coefficient norm $\|\vec{b}\|_1$. However this is an overestimate. Recall from Lemma~\ref{lem:time-evol-error-series} that each Trotterized time evolution operator admits the operator-valued expansion
\begin{equation}\label{eq:trotter-series-recall}
\cP^{1/s_k}(s_k T) - e^{-iHT} = \sum_{\substack{j \in \sigma \mathbb{Z}_+\\
j \geq p}} s_{k}^{\,j}\, \tilde{E}_{j+1}(T)\,,
\end{equation} with $s_k =  s/ q_k$. Inserting this into the extrapolation $\cP^{(R)}_{p, m}(T) =  \sum_{k=1}^m b_k \cP^{1/s_k} (s_k T)$ and using the fact that the schedule $\{b_k, q_k\}_{k=1}^m$ cancels the powers $s^{p}, s^{p+\sigma}, \ldots, s^{p + \sigma (m-2)}$ (as in Lemma~\ref{lem:specialized-Richardson-error}), the residual error operator takes the form
\begin{align}\label{eq:tilde-b-bound}
\cP^{(R)}_{p, m}(T)- e^{-iHT}\;&=\;\sum_{\substack{j \in \sigma\mathbb{Z}_+ \\ j \geq \sigma(m-1)+p}} s^{\,j}\ \Big(\sum_{k=1}^{m} \frac{b_k}{q_k^{\,j}}\Big) \tilde{E}_{j+1}(T)\,. \\
&= \;\sum_{\substack{j \in \sigma\mathbb{Z}_+ \\ j \geq \sigma(m-1)+p}} s_1^{\,j}\ \sum_{k=1}^{m} {b_k} \Big(\frac{q_1}{q_k}\Big)^{j} \tilde{E}_{j+1}(T) \\
\implies \left\| \cP^{(R)}_{p, m}(T)- e^{-iHT} \right\| &\leq \; \sum_{k=1}^{m} |{b_k}| \Big(\frac{q_1}{q_k}\Big)^{p} \sum_{\substack{j \in \sigma\mathbb{Z}_+ \\ j \geq \sigma(m-1)+p}} | s_1^{\,j}| \left\| \tilde{E}_{j+1}(T) \right\| \,.
\end{align}
As a result, for Lemma \ref{lem:specialized-Richardson-error}, rather than $\|\vec{b}\|_1$, the remainder can be thought of as being amplified by $\|\vec{b}^{(p)}\|_1$, the norm of the vector whose entries are $b_kq_1^p/q_k^p$. Similarly, for Lemma \ref{lem:generic-Richardson-error-2} the remainder is amplified by $\|\vec{b}^{(1)}\|_1$. In practice this can lead to much smaller bounds when few extrapolation points are used. We note however that the sample complexity is still controlled by $\|\vec{b}\|_1$.

\medskip

\textbf{Constant factors for plane-wave dual basis.} For electronic-structure Hamiltonians in the plane-wave dual basis with $n$ basis functions it is known that 
\begin{equation}\label{eq:comm-scaling-basic}
\alpha^{(j)}_{\text{comm}} = O(n^{j})\,,
\qquad
\lambda_{\text{comm}} = O(n)\,,
\end{equation}
from \cite{childs2021TheoryTrotter} and \cite{chakraborty2025quantum} respectively. Thus, asymptotically, the dependence of extrapolated product formulae on $n$ is the same as that for standalone product formulae. In the following Lemma we prove a result with refined constant factors for any exponentially growing $\acomm^{(j)}$.

\begin{lemma}[Bound on $\lambda_{\mathrm{comm}}$ for exponentially growing $\acomm^{(j)}$]\label{lem:lcomm-tight-bound}
Consider a $p$-th order product formula and let $\alpha_{\mathrm{comm}}^{(j)}$ and $\lcomm$ be defined as in Definitions \ref{def:acomm} and \ref{def:lcomm} respectively.
Assume there exists $A > 0$ such that
\begin{equation}\label{eq:alpha-comm-bound}
\alpha_{\mathrm{comm}}^{(t)} \le A\, n^{j}, \qquad \forall\, j \ge 1.
\end{equation}
Then, we have
\begin{equation}\label{eq:lambda-comm-bound-explicit}
\lambda_{\mathrm{comm}}
\;\le\;
n \cdot \max\{A^{\frac{1}{p+1}}, 1\} \cdot \sup_{0\leq k \leq 1/p} \left(\frac{e }{(p + 1)^{2}}\left( \frac{1}{k} - p \right) \right)^{k}\,.
\end{equation}
\end{lemma}

\begin{proof}
    Starting from Definition \ref{def:lcomm} and imposing the bound in Eq.~\eqref{eq:alpha-comm-bound}, we have
    \begin{equation}
    \prod_{\kappa=1}^{\ell}
    \frac{\alpha_{\mathrm{comm}}^{(j_\kappa + 1)}}{(j_\kappa + 1)^2}
    \le
    A^{\ell} n^{j+\ell}
    \prod_{\kappa=1}^{\ell} (j_\kappa + 1)^{-2}.
    \end{equation}
    Since $j_\kappa + 1 \ge p + 1$, we have $\prod_{\kappa} (j_\kappa + 1)^{-2} \le (p + 1)^{-2\ell}$.
    The number of integer tuples $\{j_1, \ldots, j_\ell\}$ satisfying $\sum_{\kappa} j_\kappa = j$ with $j_\kappa \ge p$ is bounded by
    \begin{equation}
    \sum_{\substack{j_1, \dots, j_\ell \ge p \\ j_1 + \cdots + j_\ell = j}} 1
    =
    \binom{j - \ell p + \ell - 1}{\ell - 1}
    \le
    \left(\frac{e (j-\ell p+\ell-1)}{\ell-1}\right)^{\ell-1}\,.
    \end{equation}
    by the standard bound on binomial coefficients.
    Combining these bounds yields
    \begin{equation}
    \sum_{\substack{j_1, \dots, j_\ell \ge p \\ j_1 + \cdots + j_\ell = j}}
    \prod_{\kappa=1}^{\ell}
    \frac{\alpha_{\mathrm{comm}}^{(j_\kappa + 1)}}{(j_\kappa + 1)^2}
    \le n^{j+\ell}
    \frac{A^{\ell} }{(p + 1)^{2\ell}}   \left(\frac{e (j-\ell p+\ell-1)}{\ell-1}\right)^{\ell-1} \leq n^{j+\ell}
    A^{\ell}    \left(\frac{e }{(p + 1)^{2}} \frac{  j-\ell p+\ell-1}{\ell-1}\right)^{\ell-1}\,.
    \end{equation}
    Taking the $(j+\ell)$-th root and the supremum, we can write
    \begin{equation}
        \lcomm \leq n \sup_{\substack{j \in \mathbb{Z}^+ {\geq p} \\ 1 \le \ell \le \lfloor j/p \rfloor}} \left( A^{\frac{\ell}{j + \ell}} \right) \sup_{\substack{j \in \mathbb{Z}^+ {\geq p} \\ 1 \le \ell \le \lfloor j/p \rfloor}} \left(\left(\frac{e }{(p + 1)^{2}} \frac{j-\ell p+\ell-1}{\ell-1}\right)^{\frac{\ell-1}{j+ \ell}} \right)\,.
    \end{equation}
    The first supremum takes value $A^{\frac{1}{p+1}}$ for $A \geq 1$.
    The second supremum can be bounded as
    \begin{align}
        \sup_{\substack{j \in \mathbb{Z}^+ {\geq p} \\ 1 \le \ell \le \lfloor j/p \rfloor}} \left( \left( \frac{e }{(p + 1)^{2}} \frac{j-\ell p+\ell-1}{\ell-1}\right)^{\frac{\ell-1}{j+ \ell}} \right) &\leq  \sup_{0\leq k \leq 1/p} \inf_{\ell; \geq 0} \left(\frac{e }{(p + 1)^{2}} \left( \frac{1}{k} - \frac{\ell p +1}{\ell-1} \right) \right)^{k}  \\
        &\leq  \sup_{0\leq k \leq 1/p} \left(\frac{e }{(p + 1)^{2}}
         \left( \frac{1}{k} - p \right) \right)^{k}\,,
    \end{align}
    where $\inf_{\ell; ()\geq 0}$ indicates an infimum such that the quantity inside the brackets is non-negative. 
\end{proof}

When $A \geq 1$, the ratio of $\lcomm^{1+1/p}$ to $(\acomm^{(p+1)})^{1/p}$ (which represents respective contributions to gate complexities) is
\begin{equation}
    \frac{\lcomm^{1+1/p}}{(\acomm^{(p+1)})^{1/p}} = \sup_{0\leq k \leq 1/p} \left(\frac{e }{(p + 1)^{2}}\left( \frac{1}{k} - p \right) \right)^{k(1+1/p)}\,.
\end{equation}
We can numerically determine this value and find it to be $\leq 1.5035$ for $p=1$, $\leq 1.1487$ for $p=2$ and $\leq 1.0445$ for $p=4$. Assuming $A\geq 1$, the choice $A=1$ maximizes the bound for $\lcomm$ (it is the worst-case choice). From hereon we take this choice.

\medskip

\textbf{Choice of extrapolation schedule.} In Lemmas \ref{lem:generic-Richardson-error-2} and \ref{lem:specialized-Richardson-error} we demonstrate bounds on the maximum number of Trotter steps for two types of extrapolation schedules. The first extrapolation schedule removes error series powers starting from $s^1$ (or $s^2$ for symmetric error series), and is compatible with the well-conditioned extrapolation schedule of \cite{low2019Multiproduct} which we used to derive our main theoretical results (Theorems \ref{thm:overlap}, \ref{thm:overlap-with-observable}, \ref{thm:distribution-recovery}). The second uses the information that a $p$-th order Trotterized time evolution error series has leading order power $s^p$. In this setting we will find extrapolation schedules by simple brute-force optimization. In both settings, we take Trotter-Suzuki product formulae where we can use the bound $a_{\max} \leq 1$.

We note that in both bounds there is a ceiling function, and we omit this for simplicity -- this would be valid if $\Upsilon \lcomm T^{1+1/p} \log(1/\varepsilon)$ were large compared to the smallest error amplification factor $q_1$. This simplification allows us to compare bounds with regular Trotter formulae (Eq.~\eqref{eq:trotter-steps}) such that we can disregard the dependency on $\Upsilon$ and $T$, as the dependence of all bounds with respect to these parameters is the same.

For the LKW well-conditioned grid we plot the following quantity (referring back to Theorem \ref{thm:time-signal}):
\begin{equation}
    \frac{q_{\max}}{q_{\min}} \lambda_{\text{scale}}^{1+1/p} \left(\frac{a(\varepsilon)\|\vec{b}^{(1)}\|_1}{\varepsilon}\right)^{\frac{1}{\sigma m}}\,,
\end{equation}
where (as with all schemes) we omit dependencies on $\Upsilon$, $T$ and $n$ due to identical contribution -- we denote $\lambda_{\text{scale}}$ as the upper bound in Lemma \ref{lem:lcomm-tight-bound} divided by $n$, i.e.~the constant factor contribution to the extrapolated commutator factor. The quantity $a(\varepsilon)$ was chosen to be $4$ for simplicity in our previous results following the analysis in Lemma \ref{lem:generic-Richardson-error} -- here, we numerically pick a refined constant for each target precision chosen to be the smallest number such that an upper bound on step number condition $s_1 (a_{\max } \Upsilon \lcomm T)^{1+1/p}$ is consistent. 

For our brute-force search we optimize and plot the following objective over a simple search space of all possible integer sequences $\{q_k\}_{k=1}^m$ (referring back to Theorem \ref{thm:time-signal-generic}):
\begin{equation}\label{eq:optimization-objective}
\mathcal{C}(q_1, \ldots, q_m) = \frac{q_{\max}}{q_{\min}} \lambda_{\text{scale}}^{1+1/p} \left(\frac{a(\varepsilon)\|\vec{b}^{(p)}\|_1}{\varepsilon}\right)^{\frac{1}{\sigma (m-1)+p}},
\end{equation}
in the domain $q_k \in [1, 10]$, and where $\vec{b}$ is computed as the solution of the Vandermonde linear system as in Eq.~\eqref{eq:vandermonde-system}. 

We present our plot of upper bounds of the LKW well-conditioned grid, the brute-force optimized grid, alongside the regular Trotter bound across orders $p \in {1,2,4}$ in Figure~\ref{fig:overhead-plots}. We recall that the LKW grid only works for symmetric product formulae, and so is not available for $p=1$. We record the sample overhead amplification factor $\|\vec{b}\|_{1}^{2}$ as the color of each data marker. During brute-force search, candidate grids with $\|\vec{b}\|_{1}^{2} > \texttt{thresh}$ are rejected for $\texttt{thresh} \in  \{10,100,1000\}$, leading to three distinct lines. 

\begin{figure}[t]
    \centering
    \includegraphics[width=\linewidth]{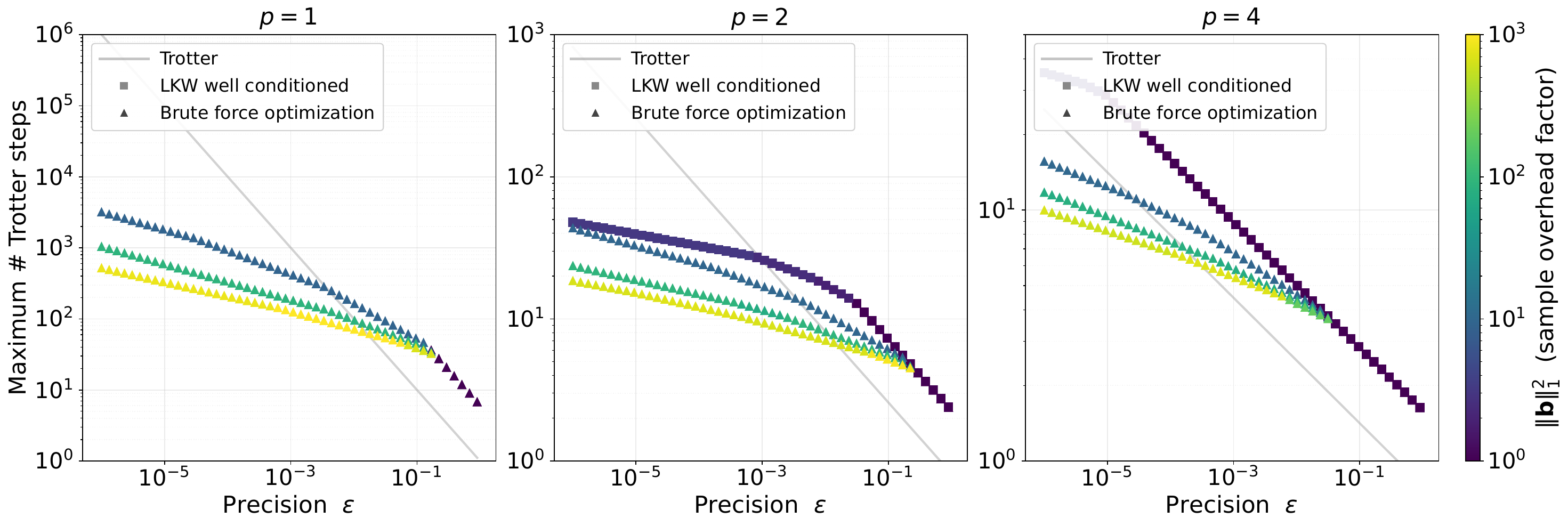}
    \caption{\textbf{Bounds on maximum number of Trotter steps.} We plot our best known bounds on the maximum required Trotter step number to attain a given additive error for the time signal task. The number we plot at each target precision is proportional to the true Trotter step number, where the true value contains a factor dependent on the target simulation time and problem size.}
    \label{fig:overhead-plots}
\end{figure}

We observe first that despite adopting a very simplistic search strategy for extrapolation schedules, our brute-force optimization improves upon the prescribed extrapolation schedule of~\cite{low2019Multiproduct}. This demonstrates that heuristic optimization could lead to useful extrapolation schedules that go beyond the performance of our (asymptotically strong) analytical results. Second, our bounds on extrapolated Trotter formulae outperform Trotter formulae for small enough precision, but is not found for large precision. We remark that for each order the extrapolation schemes at the largest target precision are the trivial one (no extrapolation). Thus, the deviation between extrapolated data and the regular Trotter curve is purely due to difference in constant factor; were constant factors to be equal, Randomized Extrapolation Trotterization would outperform regular Trotter unconditionally across our studied orders. 
We thus envisage improvements can be found with refined analysis of constant factors and a refined optimization procedure in future work.

\medskip

\textbf{Gate depths.} In the previous figure we have only investigated Trotter step number, which is only one part of the contribution to total gate complexity per circuit. Inspection of Figure \ref{fig:overhead-plots} shows that the crossover point where Randomized Extrapolation Trotterization is more efficient than regular Trotter appears to get worse as the order increases. In reality, the number of stages prohibits use of high-order product formulae (e.g.~the sixth-order Trotter--Suzuki formula requires $50$ operators per Hamiltonian term per Trotter step, compared to $2$ for second-order). Figure~\ref{fig:gate-depth} presents (a number proportional to) the gate depth for the same data displayed in Figure \ref{fig:overhead-plots}. For simplicity, for lines corresponding to extrapolated data we have selected only extrapolation schedules found via heuristic search with sample amplification factor at most $10$, and recall that stronger performance can be obtained if a larger sample budget is allowed. We choose a system size of $n=100$ basis functions, deemed a large enough regime to potentially observe quantum advantage for the uniform electron gas (Jellium) \cite{babbush2018LowDepthQSimMaterial}. In the setting chosen for comparison we see that Randomized Extrapolation Trotterization beats regular product formulae (in bounds) across order $p=1$ to $p=6$ below any error rate $\leq 10^{-2}$.

 \begin{figure}[ht]
    \centering
\includegraphics[width=\linewidth]{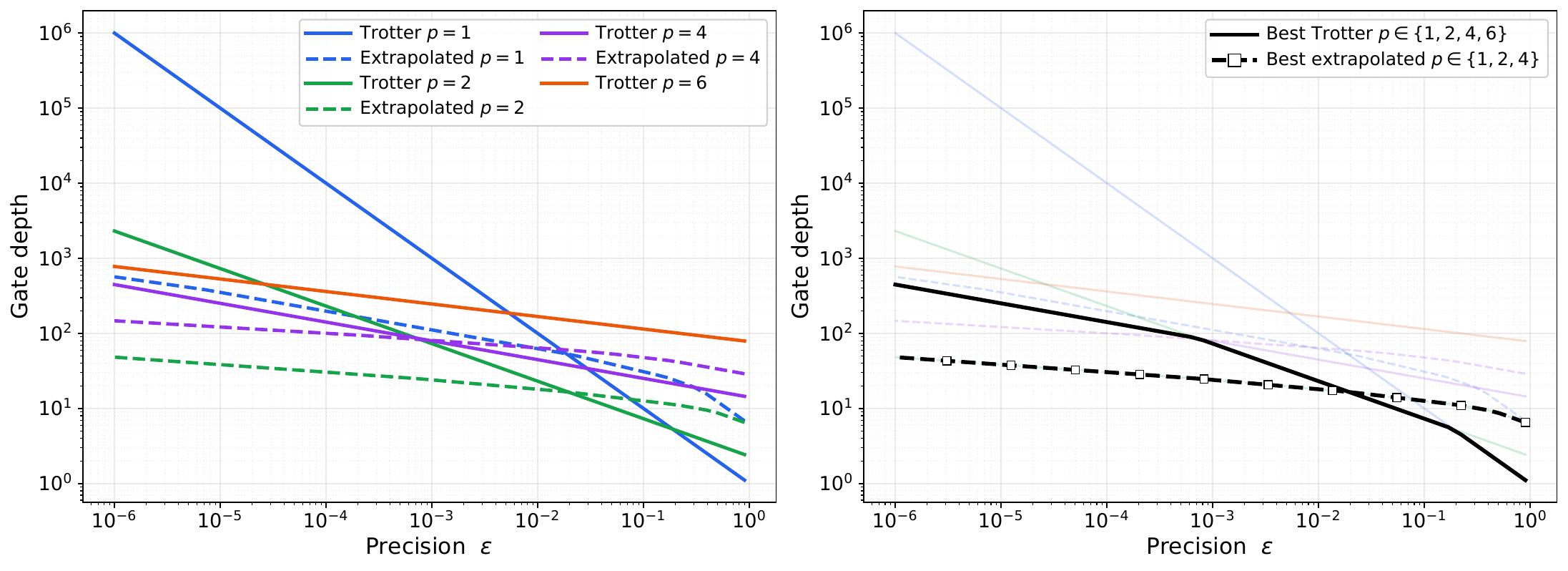}
    \caption{\textbf{Bounds on gate depth.} We plot best available bounds on (a quantity proportional to) the gate depth for a plane-wave dual basis Hamiltonian with $100$ basis functions. \textit{Left:} Superposed data taken from Figure \ref{fig:overhead-plots} where we have selected only brute-force optimized extrapolation schedules with sample overhead factor 10. \textit{Right:} We highlight the best curve across all orders from the plot on the left.}
    \label{fig:gate-depth}
\end{figure}

\subsection{Empirical error}

So far we have compared analytical bounds on the gate depth to reach some target precision. This is useful if one wishes to perform resource estimates, or to perform an experiment with precision guarantees (assuming perfect execution). However, bounds only tell one side of the story, and the true error attained by Randomized Extrapolation Trotterization may be much better than the bounds we derive in practice. Indeed, the true error of regular Trotter circuits may also be substantially better. 

In order to gain some insight into the exact error attainable by Randomized Extrapolation Trotterization, we take the extrapolation schedule data from our previous numerical studies and import them in a "blind" fashion into a new problem: simulating a Heisenberg Hamiltonian. We consider an anisotropic Heisenberg Hamiltonian with an external magnetic field
\begin{equation}
    H = J \sum_{i=1}^{L-1} ( X_i X_{i+1} + Y_i Y_{i+1} + 0.9 Z_i Z_{i+1}) - h\sum_{i=1}^L Z_i\,,  
\end{equation}
with a choice of $J=h=1$.

\begin{figure}[t]
    \centering
    \includegraphics[width=\linewidth]{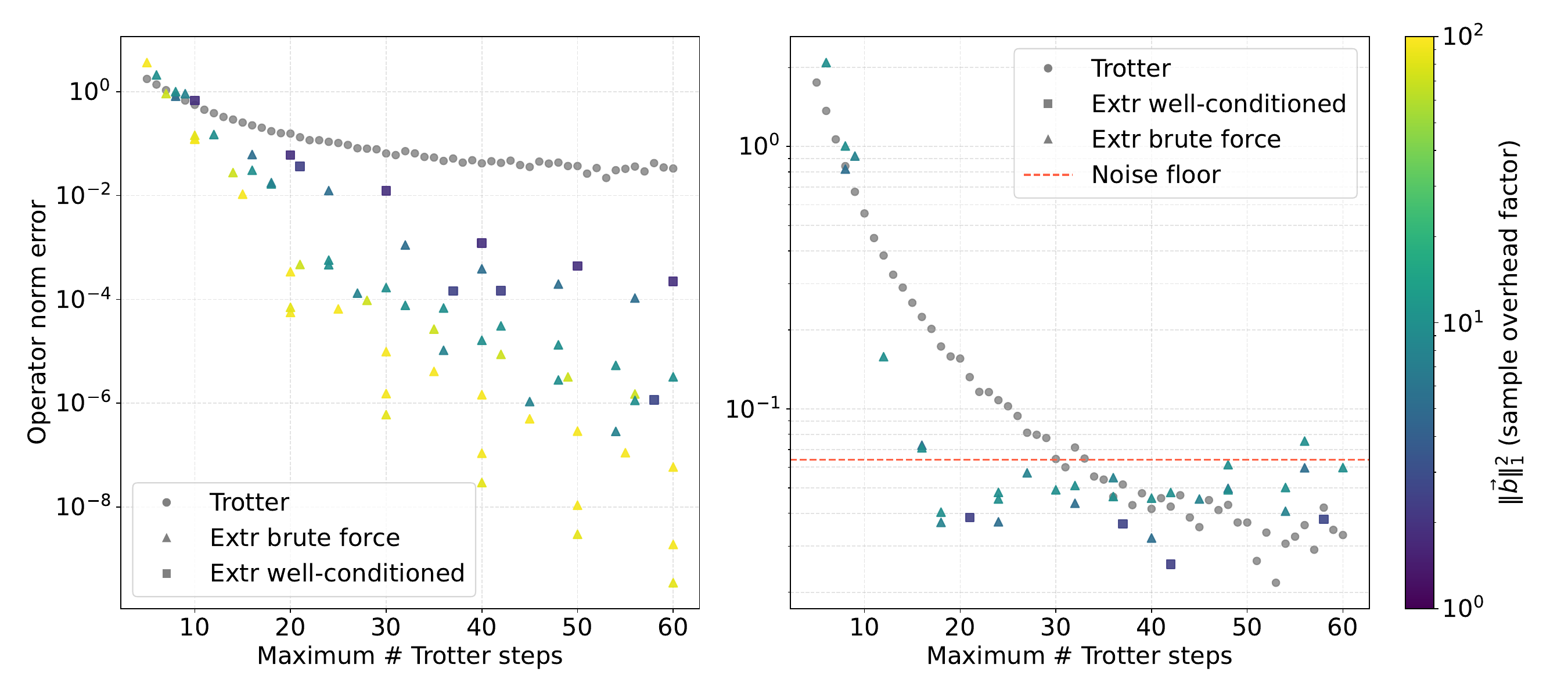}
    \caption{\textbf{Exact error for Heisenberg model.} We plot the exact error on the time evolution operator for a Heisenberg chain Hamiltonian on $8$ qubits using the second-order Suzuki-Trotter formula. Brute-force extrapolation schedules are the same that were found in the optimization of Figures \ref{fig:overhead-plots} and \ref{fig:gate-depth}. \textit{Left:} We plot the operator norm error for every extrapolation schedule previously found. \textit{Right:} We introduce a layer of random single-qubit $X$-rotations to simulate measurement noise. The red dotted line signifies the noise rate corresponding to one standard deviation.}
    \label{fig:exact}
\end{figure}

On a system of 8 qubits we evaluate the exact operator norm error of a simulation of the time evolution operator for both a Suzuki-Trotter circuit and Randomized Extrapolation Trotterization implementations for $p=2$. We import all of our extrapolation schemes found for $p=2$ in our previous numerical study. We display our results in the left subfigure of Figure \ref{fig:exact}, and plot the operator norm error against the maximum number of Trotter steps required in one coherent circuit run. As before, sample overhead is indicated by the color of each data point. We see that blow any large target error ($\sim 1$), all the previously found extrapolation schemes allow substantially reduced error rates compared to regular Trotter circuits.

In the right subfigure of Figure \ref{fig:exact} we investigate the robustness of our extrapolation schemes in the presence of high circuit noise. We consider a simple model of measurement noise described by random Pauli-$X$ rotations. This incurs a noise floor below which no strategy can pass. We plot only our best conditioned extrapolation schemes, that is schemes with a condition number $\|\vec{b}\|_1 \leq \sqrt{10}$. We see that the extrapolated data is able to hit the noise floor using fewer Trotter steps than the regular Trotter circuits (see region between 10-30 Trotter steps). We interpret this as evidence that heuristically found extrapolation schemes are in practice robust to noise, despite having non-trivial condition numbers. Additionally, such schemes can improve on the performance of regular Trotter circuits, despite the presence of a high noise rate.

\section*{Acknowledgments}

The authors thank the following for insightful discussions (listed in alphabetical order): Fernando Brand\~ao, Garnet Chan, Abhinav Deshpande, Sam McArdle, Yuan Su. AM was supported by Caltech’s Summer Undergraduate Research Fellowship (SURF) program. We acknowledge funding provided by the Institute for Quantum Information and Matter, an NSF Physics Frontiers Center (NSF Grant PHY-2317110).

\vspace{-0.5cm}
\printbibliography[heading=bibintoc]

@preamble{ "\newcommand{\lName}{0}" }

@preamble{ "\newcommand{\arxiv}[1]{arXiv:\href{https://arxiv.org/abs/#1}{\ttfamily{#1}}\removefirstdot}" }

@preamble{ "\newcommand{\arXiv}[1]{arXiv:\href{https://arxiv.org/abs/#1}{\ttfamily{#1}}\removefirstdot}" }

@preamble{ "\newcommand{\iacr}[1]{ePrint:\href{https://eprint.iacr.org/#1}{\ttfamily{#1}}\removefirstdot}" }

@preamble{ "\def\removefirstdot#1{\if.#1{}\else#1\fi}" }

@preamble{ "\providecommand{\multiletter}[1]{#1}\renewcommand{\multiletter}[1]{#1}" }

@preamble{ "\DeclareRobustCommand{\dutchPrefix}[2]{#2}" }

@preamble{ "\providecommand{\dutchPrefix}[2]{#2}\renewcommand{\dutchPrefix}[2]{#2}" }

@preamble{ "\newcommand{\skp}[3]{#2}" }

@preamble{"\newcommand{\focs       }[1]{Proceedings of the #1 {IEEE} Symposium on Foundations of Computer Science ({FOCS})}"}

@preamble{"\newcommand{\stoc       }[1]{Proceedings of the #1 {ACM} Symposium on the Theory of Computing ({STOC})}"}

@preamble{"\newcommand{\soda       }[1]{Proceedings of the #1 {ACM-SIAM} Symposium on Discrete Algorithms ({SODA})}"}

@preamble{"\newcommand{\stacs      }[1]{\if\lName1\skp{  }{Proceedings of the #1 Symposium on Theoretical Aspects of Computer Science ({STACS})}{                          }\else{STACS}\fi}"}

@preamble{"\newcommand{\itcs       }[1]{\if\lName1\skp{  }{Proceedings of the #1 Innovations in Theoretical Computer Science Conference ({ITCS})}{                         }\else{ITCS}\fi}"}

@preamble{"\newcommand{\fsttcs     }[1]{\if\lName1\skp{  }{Proceedings of the #1 International Conference on Foundations of Software Technology and Theoretical Computer Science ({FSTTCS})}{ }\else{FSTTCS}\fi}"}

@preamble{"\newcommand{\mfcs       }[1]{\if\lName1\skp{  }{Proceedings of the #1 International Symposium on Mathematical Foundations of Computer Science ({MFCS})}{        }\else{MFCS}\fi}"}

@preamble{"\newcommand{\ccc        }[1]{\if\lName1\skp{  }{Proceedings of the #1 {IEEE} Conference on Computational Complexity ({CCC})}{                                   }\else{CCC}\fi}"}

@preamble{"\newcommand{\isit       }[1]{\if\lName1\skp{  }{Proceedings of the #1 {IEEE} International Symposium on Information Theory ({ISIT})}{                           }\else{ISIT}\fi}"}

@preamble{"\newcommand{\colt       }[1]{\if\lName1\skp{  }{Proceedings of the #1 Conference On Learning Theory ({COLT})}{                                                  }\else{COLT}\fi}"}

@preamble{"\newcommand{\nips       }[1]{\if\lName1\skp{  }{Advances in Neural Information Processing Systems #1 ({NeurIPS})}{                                                 }\else{NeurIPS}\fi}"}

@preamble{"\newcommand{\aistats    }[1]{\if\lName1\skp{  }{Proceedings of the #1 International Conference on Artificial Intelligence and Statistics ({AISTATS})}{          }\else{AISTATS}\fi}"}

@preamble{"\newcommand{\icml       }[1]{\if\lName1\skp{  }{Proceedings of the #1 International Conference on Machine Learning ({ICML})}{                                   }\else{ICML}\fi}"}

@preamble{"\newcommand{\icalp      }[1]{\if\lName1\skp{  }{Proceedings of the #1 International Colloquium on Automata, Languages, and Programming ({ICALP})}{              }\else{ICALP}\fi}"}

@preamble{"\newcommand{\esa        }[1]{\if\lName1\skp{  }{Proceedings of the #1 Annual European Symposium on Algorithms ({ESA})}{                                         }\else{ESA}\fi}"}

@preamble{"\newcommand{\tqc        }[1]{\if\lName1\skp{  }{Proceedings of the #1 Conference on the Theory of Quantum Computation, Communication, and Cryptography ({TQC})}{}\else{TQC}\fi}"}

@preamble{"\newcommand{\isca        }[1]{\if\lName1\skp{  }{Proceedings of the #1 International Symposium on       Computer Architecture ({ISCA})}{}\else{ISCA}\fi}"}

@preamble{"\newcommand{\isaac      }[1]{\if\lName1\skp{  }{Proceedings of the #1 International Symposium on Algorithms and Computation ({ISAAC})}{                         }\else{ISAAC}\fi}"}

@preamble{"\newcommand{\aaai       }[1]{\if\lName1\skp{  }{Proceedings of the #1 AAAI Conference on Artificial Intelligence}{                                              }\else{AAAI}\fi}"}

@preamble{"\newcommand{\socg       }[1]{\if\lName1\skp{  }{Proceedings of the #1 Annual Symposium on Computational geometry}{                                              }\else{SoCG}\fi}"}

@preamble{"\newcommand{\sofsem     }[1]{\if\lName1\skp{  }{SOFSEM #1: Theory and Practice of Computer Science}{                                                            }\else{SOFSEM}\fi}"}

@preamble{"\newcommand{\ecc        }[1]{\if\lName1\skp{  }{#1 European Control Conference (ECC)}{                                                                          }\else{ECC}\fi}"}

@preamble{"\newcommand{\crypto     }[1]{\if\lName1\skp{  }{Advances in Cryptology -- CRYPTO #1}{                                                                           }\else{CRYPTO}\fi}"}

@preamble{"\newcommand{\asiacrypt  }[1]{\if\lName1\skp{  }{Advances in Cryptology -- ASIACRYPT #1}{                                                                        }\else{ASIACRYPT}\fi}"}

@preamble{"\newcommand{\eurocrypt  }[1]{\if\lName1\skp{  }{Advances in Cryptology -- EUROCRYPT #1}{                                                                        }\else{EUROCRYPT}\fi}"}

@preamble{"\newcommand{\pqcrypto   }[1]{\if\lName1\skp{  }{Post-Quantum Cryptography}{                                                                                     }\else{PQCrypto}\fi}"}

@preamble{"\newcommand{\scConference}[1]{\if\lName1\skp{  }{Proceedings of the International Conference for High Performance Computing, Networking, Storage and Analysis}{  }\else{SC}\fi}"}

@preamble{"\newcommand{\fccm}[1]{\if\lName1\skp{  }{#1 {IEEE} Annual International Symposium on Field-Programmable Custom Computing Machines}{  }\else{FCCM}\fi}"}

@preamble{"\newcommand{\lattice       }[1]{\if\lName1\skp{  }{Proceedings of the #1 International Symposium on Lattice Field Theory}{                                              }\else{Lattice}\fi}"}

@preamble{"\newcommand{\aft       }[1]{\if\lName1\skp{  }{Proceedings of the #1 ACM Conference on Advances in Financial Technologies}{                                              }\else{AFT}\fi}"}

@preamble{"\newcommand{\secConference       }[1]{\if\lName1\skp{  }{Proceedings of the IEEE/ACM #1 Symposium on Edge Computing}{                                              }\else{SEC}\fi}"}

@preamble{"\newcommand{\ims       }[1]{\if\lName1\skp{  }{#1 {IEEE} {MTT}-{S} International Microwave Symposium (IMS)}{                                              }\else{IMS}\fi}"}

@preamble{"\newcommand{\jacm          }{\if\lName1\skp{  }{Journal of the ACM}{                             }\else{Journal of the ACM}\fi}"}

@preamble{"\newcommand{\acmta         }{\if\lName1\skp{  }{ACM Transactions on Algorithms}{                 }\else{{ACM} Trans. Algorithms}\fi}"}

@preamble{"\newcommand{\acmtct        }{\if\lName1\skp{  }{ACM Transactions on Computation Theory}{         }\else{ACM Trans. Comput. Theory}\fi}"}

@preamble{"\newcommand{\acmtqc        }{\if\lName1\skp{  }{ACM Transactions on Quantum Computing}{          }\else{ACM Trans. Quantum Comput.}\fi}"}

@preamble{"\newcommand{\acmjetcs        }{\if\lName1\skp{  }{   ACM Journal on Emerging Technologies in Computing Systems }{          }\else{ACM J. Emerg. Technol. Comput. Syst.}\fi}"}

@preamble{"\newcommand{\canadianjmath }{\if\lName1\skp{  }{Canadian Journal of Mathematics      }{          }\else{Can. J. Math.}\fi}"}

@preamble{"\newcommand{\jams          }{\if\lName1\skp{  }{Journal of the AMS}{                             }\else{J. AMS}\fi}"}

@preamble{"\newcommand{\bullAMS        }{\if\lName1\skp{  }{Bulletin of the American Mathematical Society}{                                }\else{Bull. AMS}\fi}"}

@preamble{"\newcommand{\pams          }{\if\lName1\skp{  }{Proceedings of the AMS}{                         }\else{Proc. AMS}\fi}"}

@preamble{"\newcommand{\linalgappl    }{\if\lName1\skp{  }{Linear Algebra and its Applications}{            }\else{Linear Algebra App.}\fi}"}

@preamble{"\newcommand{\jalgo         }{\if\lName1\skp{  }{Journal of Algorithms}{                          }\else{J. Algorithms}\fi}"}

@preamble{"\newcommand{\jcss          }{\if\lName1\skp{  }{Journal of Computer and System Sciences}{        }\else{J. Comput. Syst. Sci.}\fi}"}

@preamble{"\newcommand{\jcomputapplmath}{\if\lName1\skp{  }{Journal of Computational and Applied Mathematics}{        }\else{J. Comput. Appl. Math.}\fi}"}

@preamble{"\newcommand{\cc            }{\if\lName1\skp{  }{Computational Complexity}{                       }\else{Comput. Complex.}\fi}"}

@preamble{"\newcommand{\algor         }{\if\lName1\skp{  }{Algorithmica}{                                   }\else{Algorithmica}\fi}"}

@preamble{"\newcommand{\comb          }{\if\lName1\skp{  }{Combinatorica}{                                  }\else{Combinatorica}\fi}"}

@preamble{"\newcommand{\cacm          }{\if\lName1\skp{  }{Communications of the ACM}{                      }\else{Commun. ACM}\fi}"}

@preamble{"\newcommand{\sigart        }{\if\lName1\skp{  }{SIGART Bulletin}{                                }\else{SIGART Bull.}\fi}"}

@preamble{"\newcommand{\sigactn       }{\if\lName1\skp{  }{SIGACT News}{                                    }\else{SIGACT News}\fi}"}

@preamble{"\newcommand{\eatcsbul      }{\if\lName1\skp{  }{Bulletin of the {EATCS}}{                        }\else{Bull. {EATCS}}\fi}"}

@preamble{"\newcommand{\siamrev       }{\if\lName1\skp{  }{SIAM Review}{                                    }\else{SIAM Rev.}\fi}"}

@preamble{"\newcommand{\siamjc        }{\if\lName1\skp{  }{SIAM Journal on Computing}{                      }\else{SIAM J. Comp.}\fi}"}

@preamble{"\newcommand{\siamjo        }{\if\lName1\skp{  }{SIAM Journal on Optimization}{                   }\else{SIAM J. Opt.}\fi}"}

@preamble{"\newcommand{\siamjdm       }{\if\lName1\skp{  }{SIAM Journal on Discrete Mathematics}{           }\else{SIAM J. Disc. Math.}\fi}"}

@preamble{"\newcommand{\siamjnum      }{\if\lName1\skp{  }{SIAM Journal on Numerical Analysis}{             }\else{SIAM J. Num. Anal.}\fi}"}

@preamble{"\newcommand{\siamjmathanal }{\if\lName1\skp{  }{SIAM Journal on Mathematical Analysis}{          }\else{SIAM J. Math. Anal.}\fi}"}

@preamble{"\newcommand{\discmath      }{\if\lName1\skp{  }{Discrete Mathematics}{                           }\else{Disc. Math.}\fi}"}

@preamble{"\newcommand{\das           }{\if\lName1\skp{  }{Discrete Applied Mathematics}{                   }\else{Disc. App. Math.}\fi}"}

@preamble{"\newcommand{\annmath      }{\if\lName1\skp{  }{Annals of Mathematics}{              }\else{Ann. Math.}\fi}"}

@preamble{"\newcommand{\amatstat      }{\if\lName1\skp{  }{Annals of Mathematical Statistics}{              }\else{Ann. Math. Stat.}\fi}"}

@preamble{"\newcommand{\rms           }{\if\lName1\skp{  }{Russian Mathematical Surveys}{                   }\else{Russ. Math. Surv.}\fi}"}

@preamble{"\newcommand{\invmath       }{\if\lName1\skp{  }{Inventiones Mathematicae}{                       }\else{Inv. Math.}\fi}"}

@preamble{"\newcommand{\jnumber       }{\if\lName1\skp{  }{Journal of Number Theory}{                       }\else{J. Num. Th.}\fi}"}

@preamble{"\newcommand{\tcs           }{\if\lName1\skp{  }{Theoretical Computer Science}{                   }\else{Theor. Comput. Sci.}\fi}"}

@preamble{"\newcommand{\numeralgorithms}{\if\lName1\skp{  }{Numerical Algorithms}{                          }\else{Numer. Algorithms}\fi}"}

@preamble{"\newcommand{\toc           }{\if\lName1\skp{  }{Theory of Computing}{                            }\else{Theory Comput.}\fi}"}

@preamble{"\newcommand{\cjtcs         }{\if\lName1\skp{  }{Chicago Journal of Theoretical Computer Science}{}\else{Chicago J. Theoret. Comput. Sci.}\fi}"}

@preamble{"\newcommand{\mathprogram   }{\if\lName1\skp{  }{Mathematical Programming}{}\else{Math. Program.}\fi}"}

@preamble{"\newcommand{\mathcomput    }{\if\lName1\skp{  }{Mathematics of Computation}{}\else{Math. Comput.}\fi}"}

@preamble{"\newcommand{\jfourieranalappl   }{\if\lName1\skp{  }{Journal of Fourier Analysis and Applications}{}\else{J. Fourier Anal. Appl.}\fi}"}

@preamble{"\newcommand{\dcg           }{\if\lName1\skp{  }{Discrete \& Computational Geometry}{}\else{Discrete Comput. Geom.}\fi}"}

@preamble{"\newcommand{\randstructalg }{\if\lName1\skp{  }{Random Structures \& Algorithms}{}\else{Rand. Struct. Algorithms}\fi}"}

@preamble{"\newcommand{\intjunconvent }{\if\lName1\skp{  }{International Journal of Unconventional Computing}{}\else{Int. J. Unconv. Comput.}\fi}"}

@preamble{"\newcommand{\machlearnscitech }{\if\lName1\skp{  }{Machine Learning: Science and Technology}{}\else{Mach. Learn.: Sci. Technol.}\fi}"}

@preamble{"\newcommand{\machlearning  }{\if\lName1\skp{  }{Machine Learning}{}\else{Mach. Learn.}\fi}"}

@preamble{"\newcommand{\neuralprocesslett } {\if\lName1\skp{  }{Neural Processing Letters}{}\else{Neural Process. Lett.}\fi}"}

@preamble{"\newcommand{\neuralcomput }{\if\lName1\skp{  }{Neural Computation}{}\else{Neural Comput.}\fi}"}

@preamble{"\newcommand{\frontartifintell}{\if\lName1\skp{  }{Frontiers in Artificial Intelligence}{}\else{Front. Artif. Intell.}\fi}"}

@preamble{"\newcommand{\jgloboptim}{\if\lName1\skp{  }{Journal of Global Optimization}{}\else{J. Glob. Optim.}\fi}"}

@preamble{"\newcommand{\epjdatasci}{\if\lName1\skp{  }{EPJ Data Science}{}\else{EPJ Data Sci.}\fi}"}

@preamble{"\newcommand{\epjquantumtechnol}{\if\lName1\skp{  }{EPJ Quantum Technology}{}\else{EPJ Quantum Technol.}\fi}"}

@preamble{"\newcommand{\applsci       }{\if\lName1\skp{  }{Applied Sciences}{    }\else{Appl. Sci.}\fi}"}

@preamble{"\newcommand{\applphysrev   }{\if\lName1\skp{  }{Applied Physics Reviews}{    }\else{Appl. Phys. Rev.}\fi}"}

@preamble{"\newcommand{\advquantumtechnol          }{\if\lName1\skp{  }{Advanced Quantum Technologies}{    }\else{Adv. Quantum Technol.}\fi}"}

@preamble{"\newcommand{\annphys       }{\if\lName1\skp{  }{Annals of Physics}{    }\else{Ann. Phys.}\fi}"}

@preamble{"\newcommand{\annualrevCMP          }{\if\lName1\skp{  }{Annual Review of Condensed Matter Physics}{    }\else{Annu. Rev. Condens. Matter Phys.}\fi}"}

@preamble{"\newcommand{\annualrevNPS          }{\if\lName1\skp{  }{Annual Review of Nuclear and Particle Science}{    }\else{Annu. Rev. Nucl. Part. Sci.}\fi}"}

@preamble{"\newcommand{\quantum       }{\if\lName1\skp{  }{Quantum}{                                          }\else{Quantum}\fi}"}

@preamble{"\newcommand{\quantumscitechnol       }{\if\lName1\skp{  }{Quantum Science and Technology}{                                          }\else{Quantum Sci. Technol.}\fi}"}

@preamble{"\newcommand{\cmp           }{\if\lName1\skp{  }{Communications in Mathematical Physics}{             }\else{Commun. Math. Phys.}\fi}"}

@preamble{"\newcommand{\frontphys     }{\if\lName1\skp{  }{Frontiers in Physics}{                               }\else{Front. Phys.}\fi}"}

@preamble{"\newcommand{\jmp           }{\if\lName1\skp{  }{Journal of Mathematical Physics}{                    }\else{J. Math. Phys.}\fi}"}

@preamble{"\newcommand{\rspa          }{\if\lName1\skp{  }{Proceedings of the Royal Society A}{                 }\else{Proc. R. Soc. A}\fi}"}

@preamble{"\newcommand{\philostransroyal}{\if\lName1\skp{  }{Philosophical Transactions of the Royal Society A}{                 }\else{Philos. Trans. R. Soc. A}\fi}"}

@preamble{"\newcommand{\qic           }{\if\lName1\skp{  }{Quantum Information and Computation}{                }\else{Quantum Inf. Comput.}\fi}"}

@preamble{"\newcommand{\qip           }{\if\lName1\skp{  }{Quantum Information Processing}{                }\else{Quantum Inf. Process.}\fi}"}

@preamble{"\newcommand{\physrev       }{\if\lName1\skp{  }{Physical Review}{                                    }\else{Phys. Rev.}\fi}"}

@preamble{"\newcommand{\pra           }{\if\lName1\skp{  }{Physical Review A}{                                  }\else{Phys. Rev. A}\fi}"}

@preamble{"\newcommand{\prb           }{\if\lName1\skp{  }{Physical Review B}{                                  }\else{Phys. Rev. B}\fi}"}

@preamble{"\newcommand{\prc           }{\if\lName1\skp{  }{Physical Review C}{                                  }\else{Phys. Rev. C}\fi}"}

@preamble{"\newcommand{\prd           }{\if\lName1\skp{  }{Physical Review D}{                                  }\else{Phys. Rev. D}\fi}"}

@preamble{"\newcommand{\pre           }{\if\lName1\skp{  }{Physical Review E}{                                  }\else{Phys. Rev. E}\fi}"}

@preamble{"\newcommand{\prr           }{\if\lName1\skp{  }{Physical Review Research}{                           }\else{Phys. Rev. Res.}\fi}"}

@preamble{"\newcommand{\prx           }{\if\lName1\skp{  }{Physical Review X}{                                  }\else{Phys. Rev. X}\fi}"}

@preamble{"\newcommand{\prxq          }{\if\lName1\skp{  }{PRX Quantum}{                                        }\else{PRX Quantum}\fi}"}

@preamble{"\newcommand{\prl           }{\if\lName1\skp{  }{Phys. Rev. Lett.}{                            }\else{Phys. Rev. Lett.}\fi}"}

@preamble{"\newcommand{\europhyslett  }{\if\lName1\skp{  }{Europhysics Letters}{                                }\else{Europhys. Lett.}\fi}"}

@preamble{"\newcommand{\epja          }{\if\lName1\skp{  }{The European Physical Journal A}{                    }\else{Euro. Phys. J. A}\fi}"}

@preamble{"\newcommand{\epjd          }{\if\lName1\skp{  }{The European Physical Journal D}{                    }\else{Euro. Phys. J. D}\fi}"}

@preamble{"\newcommand{\njp           }{\if\lName1\skp{  }{New Journal of Physics}{                             }\else{New J. Phys.}\fi}"}

@preamble{"\newcommand{\prapp         }{\if\lName1\skp{  }{Physical Review Applied}{                            }\else{Phys. Rev. Appl.}\fi}"}

@preamble{"\newcommand{\physrep       }{\if\lName1\skp{  }{Physics Reports}{                                    }\else{Phys. Rep.}\fi}"}

@preamble{"\newcommand{\rmp           }{\if\lName1\skp{  }{Reviews of Modern Physics}{                          }\else{Rev. Mod. Phys.}\fi}"}

@preamble{"\newcommand{\repprogphys   }{\if\lName1\skp{  }{Reports on Progress in Physics}{                     }\else{Rep. Prog. Phys.}\fi}"}

@preamble{"\newcommand{\physplasmas   }{\if\lName1\skp{  }{Physics of Plasmas}{                                 }\else{Phys. Plasmas}\fi}"}

@preamble{"\newcommand{\phystoday     }{\if\lName1\skp{  }{Physics Today}{                                      }\else{Phys. Today}\fi}"}

@preamble{"\newcommand{\physics       }{\if\lName1\skp{  }{Physics}{                                            }\else{Phys.}\fi}"}

@preamble{"\newcommand{\nature        }{\if\lName1\skp{  }{Nature}{                                             }\else{Nature}\fi}"}

@preamble{"\newcommand{\natcomm       }{\if\lName1\skp{  }{Nature Communications}{                              }\else{Nat. Commun.}\fi}"}

@preamble{"\newcommand{\natcomputsci  }{\if\lName1\skp{  }{Nature Computational Science}{                       }\else{Nat. Comput. Sci.}\fi}"}

@preamble{"\newcommand{\natphys       }{\if\lName1\skp{  }{Nature Physics}{                                     }\else{Nat. Phys.}\fi}"}

@preamble{"\newcommand{\natphotonics  }{\if\lName1\skp{  }{Nature Photonics}{                                     }\else{Nat. Photonics}\fi}"}

@preamble{"\newcommand{\natrevphys    }{\if\lName1\skp{  }{Nature Reviews Physics}{                                     }\else{Nat. Rev. Phys.}\fi}"}

@preamble{"\newcommand{\natrevmater   }{\if\lName1\skp{  }{Nature Reviews Materials}{                                     }\else{Nat. Rev. Mater.}\fi}"}

@preamble{"\newcommand{\natrevmethodsprimers}{\if\lName1\skp{}{Nature Reviews Methods Primers}{                               }\else{Nat. Rev. Methods Primers}\fi}"}

@preamble{"\newcommand{\npjqi         }{\if\lName1\skp{  }{npj Quantum Information}{                            }\else{npj Quant. Inf.}\fi}"}

@preamble{"\newcommand{\scirep        }{\if\lName1\skp{  }{Scientific Reports}{                                 }\else{Sci. Rep.}\fi}"}

@preamble{"\newcommand{\science       }{\if\lName1\skp{  }{Science}{                                            }\else{Science}\fi}"}

@preamble{"\newcommand{\sciadv      }{\if\lName1\skp{  }{Science Advances}{                                   }\else{Sci. Adv.}\fi}"}

@preamble{"\newcommand{\scibull      }{\if\lName1\skp{  }{Science Bulletin}{                                   }\else{Sci. Bull.}\fi}"}

@preamble{"\newcommand{\jchemphys           }{\if\lName1\skp{  }{The Journal of Chemical Physics}{ }\else{J. Chem. Phys.}\fi}"}

@preamble{"\newcommand{\jphyschemlett      }{\if\lName1\skp{  }{The Journal of Physical Chemistry Letters}{ }\else{J. Phys. Chem. Lett.}\fi}"}

@preamble{"\newcommand{\jpa           }{\if\lName1\skp{  }{Journal of Physics A: Mathematical and Theoretical}{ }\else{J. Phys. A}\fi}"}

@preamble{"\newcommand{\jpg           }{\if\lName1\skp{  }{Journal of Physics G: Nuclear and Particle Physics}{ }\else{J. Phys. G}\fi}"}

@preamble{"\newcommand{\ijtp          }{\if\lName1\skp{  }{International Journal of Theoretical Physics}{       }\else{Int. J. Th. Phys.}\fi}"}

@preamble{"\newcommand{\jmo           }{\if\lName1\skp{  }{Journal of Modern Optics}{                           }\else{J. Mod. Opt.}\fi}"}

@preamble{"\newcommand{\jhep           }{\if\lName1\skp{  }{Journal of High Energy Physics}{                           }\else{J. High Energy Phys.}\fi}"}

@preamble{"\newcommand{\jstatph       }{\if\lName1\skp{  }{Journal of Statistical Physics}{                     }\else{J. Stat. Phys.}\fi}"}

@preamble{"\newcommand{\jcompphys       }{\if\lName1\skp{  }{Journal of Computational Physics}{                 }\else{J. Comput. Phys.}\fi}"}

@preamble{"\newcommand{\computphyscommun}{\if\lName1\skp{  }{Computer Physics Communications}{               }\else{Comput. Phys. Commun.}\fi}"}

@preamble{"\newcommand{\jstatmech     }{\if\lName1\skp{  }{Journal of Statistical Mechanics: Theory and Experiment}{ }\else{J. Stat. Mech. Theory Exp.}\fi}"}

@preamble{"\newcommand{\pnas          }{\if\lName1\skp{  }{Proceedings of the National Academy of Sciences}{    }\else{Proc. Natl. Acad. Sci.}\fi}"}

@preamble{"\newcommand{\avsquantsci          }{\if\lName1\skp{  }{AVS Quantum Science}{    }\else{AVS Quantum Sci.}\fi}"}

@preamble{"\newcommand{\quantummachintell    }{\if\lName1\skp{  }{Quantum Machine Intelligence}{    }\else{Quantum Mach. Intell.}\fi}"}

@preamble{"\newcommand{\natmachintell    }{\if\lName1\skp{  }{Nature Machine Intelligence}{    }\else{Nat. Mach. Intell.}\fi}"}

@preamble{"\newcommand{\scichinainfsci   }{\if\lName1\skp{  }{ Science China Information Sciences}{    }\else{Sci. China Inf. Sci.}\fi}"}

@preamble{"\newcommand{\jmagnreson   }{\if\lName1\skp{  }{Journal of Magnetic Resonance}{    }\else{J. Magn. Reson.}\fi}"}

@preamble{"\newcommand{\lncs          }{\if\lName1\skp{  }{Lecture Notes in Computer Science}{                  }\else{L. Notes Comp. Sci.}\fi}"}

@preamble{"\newcommand{\lnai          }{\if\lName1\skp{  }{Lecture Notes in Artificial Intelligence}{           }\else{L. Notes Art. Int.}\fi}"}

@preamble{"\newcommand{\lnm           }{\if\lName1\skp{  }{Lecture Notes in Mathematics}{                       }\else{L. Notes Math.}\fi}"}

@preamble{"\newcommand{\tams          }{\if\lName1\skp{  }{Transactions of the American Mathematical Society}{  }\else{Trans. AMS}\fi}"}

@preamble{"\newcommand{\ieeetit       }{\if\lName1\skp{  }{{IEEE} Transactions on Information Theory}{          }\else{{IEEE} Trans. Inf. Theory}\fi}"}

@preamble{"\newcommand{\ieeetnnls       }{\if\lName1\skp{  }{{IEEE} Transactions on Neural Networks and Learning Systems}{          }\else{{IEEE} Trans. Neural Netw. Learn. Syst.}\fi}"}

@preamble{"\newcommand{\ieeetcad       }{\if\lName1\skp{  }{{IEEE} Transactions on Computer-Aided Design of Integrated Circuits and Systems}{          }\else{{IEEE} Trans. Comput.-Aided Des. Integr. Circuits Syst.}\fi}"}

@preamble{"\newcommand{\ieeetqe       }{\if\lName1\skp{  }{{IEEE} Transactions on Quantum Engineering}{          }\else{{IEEE} Trans. Quantum Eng.}\fi}"}

@preamble{"\newcommand{\ieeejsac       }{\if\lName1\skp{  }{{IEEE} Journal on Selected Areas in Communications}{          }\else{{IEEE} J. Sel. Areas Commun.}\fi}"}

@preamble{"\newcommand{\ieeebits       }{\if\lName1\skp{  }{{IEEE} {BITS} the Information Theory Magazine}{          }\else{{IEEE} {BITS} Inf. Theory Mag.}\fi}"}

@preamble{"\newcommand{\quantumeng       }{\if\lName1\skp{  }{Quantum Engineering}{          }\else{Quantum Eng.}\fi}"}

@preamble{"\newcommand{\iscs          }{\if\lName1\skp{  }{International Series in Computer Science}{           }\else{Int. Ser. Comp. Sci.}\fi}"}

@preamble{"\newcommand{\tocl          }{\if\lName1\skp{  }{Theory of Computing Library}{                        }\else{Th. Comp. Lib.}\fi}"}

@preamble{"\newcommand{\actanumer     }{\if\lName1\skp{  }{Acta Numerica}{                        }\else{Acta Numer.}\fi}"}

@preamble{"\newcommand{\jderiv     }{\if\lName1\skp{  }{The Journal of Derivatives}{                        }\else{J. Deriv.}\fi}"}

@preamble{"\newcommand{\chemrev     }{\if\lName1\skp{  }{Chemical Reviews}{                        }\else{Chem. Rev.}\fi}"}

@preamble{"\newcommand{\WIREsCMS     }{\if\lName1\skp{  }{WIREs Molecular Computational Science}{                        }\else{WIREs Comput. Mol. Sci.}\fi}"}

@preamble{"\newcommand{\accchemres     }{\if\lName1\skp{  }{Accounts of Chemical Research}{                        }\else{Acc. Chem. Res.}\fi}"}

@preamble{"\newcommand{\chemphyslett     }{\if\lName1\skp{  }{Chemical Physics Letters}{                        }\else{Chem. Phys. Lett.}\fi}"}

@preamble{"\newcommand{\molphys     }{\if\lName1\skp{  }{Molecular Physics}{                        }\else{Mol. Phys.}\fi}"}

@preamble{"\newcommand{\commchem     }{\if\lName1\skp{  }{Communications Chemistry}{                        }\else{Commun. Chem.}\fi}"}

@preamble{"\newcommand{\acscentsci     }{\if\lName1\skp{  }{ACS Central Science}{                        }\else{ACS Cent. Sci.}\fi}"}

@preamble{"\newcommand{\jchemtheorycomput}{\if\lName1\skp{  }{Journal of Chemical Theory and Computation}{                        }\else{J. Chem. Theory Comput.}\fi}"}

@preamble{"\newcommand{\chemsci        }{\if\lName1\skp{  }{Chemical Science}{                        }\else{Chem. Sci.}\fi}"}

@preamble{"\newcommand{\intjqchem        }{\if\lName1\skp{  }{International journal of quantum chemistry}{                        }\else{Int. J. Quantum Chem.}\fi}"}

@preamble{"\newcommand{\jfqa        }{\if\lName1\skp{  }{Journal of Financial and Quantitative Analysis}{                        }\else{J. Financial Quant. Anal.}\fi}"}

@article{low2023complexityTrotter,
   title={Complexity of Implementing Trotter Steps},
   volume={4},
   ISSN={2691-3399},
   url={http://dx.doi.org/10.1103/PRXQuantum.4.020323},
   DOI={10.1103/prxquantum.4.020323},
   number={2},
   journal={PRX Quantum},
   publisher={American Physical Society (APS)},
   author={Low, Guang Hao and Su, Yuan and Tong, Yu and Tran, Minh C.},
   year={2023}, note={\arxiv{2211.09133}}
   }

@unpublished{aftab2024multi,
  title={Multi-product Hamiltonian simulation with explicit commutator scaling},
  author={Aftab, Junaid and An, Dong and Trivisa, Konstantina},
  note={\arxiv{2403.08922}},
  year={2024}
}

@inproceedings{aharonov2003adiabQStateGeneration,
  title={Adiabatic quantum state generation and statistical zero knowledge},
  author={Aharonov, Dorit and Ta-Shma, Amnon},
  booktitle={Proceedings of the thirty-fifth annual ACM symposium on Theory of computing},
  pages={20--29},
  year={2003},
  note={\arxiv{quant-ph/0301023}}
}

@Article{	  babbush2018LowDepthQSimMaterial,
  author	= {Babbush, Ryan and Wiebe, Nathan and Mcclean, Jarrod and
		  Mcclain, James and Neven, Hartmut and Chan, Garnet
		  Kin-Lic},
  doi		= {10.1103/PhysRevX.8.011044},
  issn		= {2160-3308},
  journal	= {\prx},
  number	= {1},
  pages		= {11044},
  publisher	= {American Physical Society},
  title		= {{Low-Depth Quantum Simulation of Materials}},
  url		= {https://doi.org/10.1103/PhysRevX.8.011044},
  volume	= {8},
  year		= {2018}
}

@article{blunt2023statistical,
  title={Statistical phase estimation and error mitigation on a superconducting quantum processor},
  author={Blunt, Nick S and Caune, Laura and Izs{\'a}k, R{\'o}bert and Campbell, Earl T and Holzmann, Nicole},
  journal={PRX Quantum},
  volume={4},
  number={4},
  pages={040341},
  year={2023},
  publisher={APS},
doi={https://doi.org/10.1103/PRXQuantum.4.040341},
note={\arxiv{2304.05126}}
}

@article{blunt2025monteCarlo,
  title={A Monte Carlo approach to bound Trotter error},
  author={Blunt, Nick S and Ivanov, Aleksei V and Bay-Smidt, Andreas Juul},
  note={\arxiv{2510.11621}},
  year={2025}
}

@InCollection{	  brassard2002AmpAndEst,
  author	= {Gilles Brassard and Peter H{\o}yer and Michele Mosca and
		  Alain Tapp},
  booktitle	= {Quantum Computation and Quantum Information: A Millennium
		  Volume},
  doi		= {10.1090/conm/305/05215},
  note		= {\arxiv{quant-ph/0005055}},
  pages		= {53--74},
  publisher	= {AMS},
  series	= {Contemporary Mathematics Series},
  title		= {Quantum Amplitude Amplification and Estimation},
  volume	= {305},
  year		= {2002}
}

@article{campbell2021earlyFTHubbard,
  title={Early fault-tolerant simulations of the Hubbard model},
  author={Campbell, Earl T},
  journal={Quantum Science and Technology},
  volume={7},
  number={1},
  pages={015007},
  year={2021},
  publisher={IOP Publishing},
doi={10.1088/2058-9565/ac3110},
note={\arxiv{2012.09238}}
}

@Article{	  campbell2019randomCompiler,
  author	= {Earl Campbell},
  doi		= {10.1103/physrevlett.123.070503},
  journal	= {\prl},
  month		= {8},
  note		= {\arxiv{1811.08017}},
  number	= {7},
  publisher	= {American Physical Society ({APS})},
  title		= {Random Compiler for Fast Hamiltonian Simulation},
  volume	= {123},
  year		= 2019
}

@unpublished{chakraborty2025quantum,
  title={Quantum singular value transformation without block encodings: Near-optimal complexity with minimal ancilla},
  author={Chakraborty, Shantanav and Hazra, Soumyabrata and Li, Tongyang and Shao, Changpeng and Wang, Xinzhao and Zhang, Yuxin},
  note={\arxiv{2504.02385}},
  year={2025}
}

@Unpublished{	  chen2021conTrotter,
  author	= {Chen, Chi-Fang and Brand{\~a}o, Fernando G. S. L.},
  note		= {\arXiv{2111.05324}},
  title		= {Average-case Speedup for Product Formulas},
  year		= {2021}
}

@Article{	  childs2021TheoryTrotter,
  author	= {Andrew M. Childs and Yuan Su and Minh C. Tran and Nathan
		  Wiebe and Shuchen Zhu},
  doi		= {10.1103/physrevx.11.011020},
  journal	= {\prx},
  month		= {2},
  note		= {\arxiv{1912.08854}},
  number	= {1},
  publisher	= {American Physical Society ({APS})},
  title		= {Theory of Trotter Error with Commutator Scaling},
  volume	= {11},
  year		= 2021
}

@unpublished{clinton2024quantumPhaseEstimation,
  title={Quantum Phase Estimation without Controlled Unitaries},
  author={Clinton, Laura and Cubitt, Toby S and Garcia-Patron, Raul and Montanaro, Ashley and Stanisic, Stasja and Stroeks, Maarten},
  journal={arXiv preprint arXiv:2410.21517},
  year={2024},
note={\arxiv{2410.21517}},
doi={https://doi.org/10.48550/arXiv.2410.21517}
}

@unpublished{dalzell2023EndToEnd,
  title={Quantum algorithms: A survey of applications and end-to-end complexities},
  author={Dalzell, Alexander M and McArdle, Sam and Berta, Mario and Bienias, Przemyslaw and Chen, Chi-Fang and Gily{\'e}n, Andr{\'a}s and Hann, Connor T and Kastoryano, Michael J and Khabiboulline, Emil T and Kubica, Aleksander and others},
  journal={arXiv preprint arXiv:2310.03011},
  year={2023},
note={\arxiv{2310.03011}}
}

@unpublished{ding2025end2end,
  title={End-to-end efficient quantum thermal and ground state preparation made simple},
  author={Ding, Zhiyan and Zhan, Yongtao and Preskill, John and Lin, Lin},
  journal={arXiv preprint arXiv:2508.05703},
  note={\arxiv{2508.05703}},
  year={2025}
}

@article{dong2022groundStatePrep,
  title={Ground-state preparation and energy estimation on early fault-tolerant quantum computers via quantum eigenvalue transformation of unitary matrices},
  author={Dong, Yulong and Lin, Lin and Tong, Yu},
  journal={PRX Quantum},
  volume={3},
  number={4},
  pages={040305},
  year={2022},
  publisher={APS},
note={\arxiv{2204.05955}},
doi={https://doi.org/10.1103/PRXQuantum.3.040305}
}

@article{faehrmann2022randomizing,
  title={Randomizing multi-product formulas for Hamiltonian simulation},
  author={Faehrmann, Paul K and Steudtner, Mark and Kueng, Richard and Kieferova, Maria and Eisert, Jens},
  journal={Quantum},
  volume={6},
  pages={806},
  year={2022},
  publisher={Verein zur F{\"o}rderung des Open Access Publizierens in den Quantenwissenschaften},
  doi={10.22331/q-2022-09-19-806},
  note={\arxiv{2101.07808}}
}

@InProceedings{	  gilyen2018QSingValTransf,
  author	= {Andr{\'a}s Gily{\'e}n and Yuan Su and Guang Hao Low and
		  Nathan Wiebe},
  booktitle	= {\stoc{51st}},
  doi		= {10.1145/3313276.3316366},
  note		= {\arxiv{1806.01838}},
  numpages	= {12},
  pages		= {193--204},
  title		= {Quantum singular value transformation and beyond:
		  {E}xponential improvements for quantum matrix arithmetics},
  year		= {2019}
}

@article{google2025observation,
  title={Observation of constructive interference at the edge of quantum ergodicity},
  journal={Nature},
  volume={646},
  number={8086}, author={Abanin, Dmitry A and Acharya, Rajeev and Aghababaie-Beni, Laleh and Aigeldinger, Georg and Ajoy, Ashok and Alcaraz, Ross and Aleiner, Igor and Andersen, Trond I and Ansmann, Markus and Arute, Frank and others},
  pages={825--830},
  year={2025},
  publisher={Nature Publishing Group UK London}, 
doi={10.1038/s41586-025-09526-6},
    note={\arxiv{2506.10191}}
}

@unpublished{gunther2025phase,
  title={Phase estimation with partially randomized time evolution},
  author={G{\"u}nther, Jakob and Witteveen, Freek and Schmidhuber, Alexander and Miller, Marek and Christandl, Matthias and Harrow, Aram},
  journal={arXiv preprint arXiv:2503.05647},
  year={2025},
note={\arxiv{2503.05647}}
}

@unpublished{haghshenas2025digitalMagnetism,
  title={Digital quantum magnetism at the frontier of classical simulations},
  author={Haghshenas, Reza and Chertkov, Eli and Mills, Michael and Kadow, Wilhelm and Lin, Sheng-Hsuan and Chen, Yi-Hsiang and Cade, Chris and Niesen, Ido and Begu{\v{s}}i{\'c}, Tomislav and Rudolph, Manuel S and others},
  journal={arXiv preprint arXiv:2503.20870},
  note={\arxiv{2503.20870}},
  year={2025}
}

@Article{	  harrow2009QLinSysSolver,
  author	= {Harrow, Aram W. and Hassidim, Avinatan and Lloyd, Seth},
  doi		= {10.1103/PhysRevLett.103.150502},
  journal	= {\prl},
  note		= {\arxiv{0811.3171}},
  number	= {15},
  pages		= {150502},
  title		= {Quantum algorithm for linear systems of equations},
  volume	= {103},
  year		= {2009}
}

@unpublished{hejazi2024LowEnergyProduct,
  title={Better bounds for low-energy product formulas},
  author={Hejazi, Kasra and Zini, Modjtaba Shokrian and Arrazola, Juan Miguel},
  note={\arxiv{2402.10362}},
  year={2024}
}

@unpublished{keen2021quantumGSGreens,
      title={Quantum Algorithms for Ground-State Preparation and Green's Function Calculation}, 
      author={Trevor Keen and Eugene Dumitrescu and Yan Wang},
      year={2021},
      note={\arxiv{2112.05731}}
}

@article{kim2023evidence,
  title={Evidence for the utility of quantum computing before fault tolerance},
  author={Kim, Youngseok and Eddins, Andrew and Anand, Sajant and Wei, Ken Xuan and Van Den Berg, Ewout and Rosenblatt, Sami and Nayfeh, Hasan and Wu, Yantao and Zaletel, Michael and Temme, Kristan and others},
  journal={Nature},
  volume={618},
  number={7965},
  pages={500--505},
  year={2023},
  publisher={Nature Publishing Group UK London},
doi={10.1038/s41586-023-06096-3}
}

@article{kiss2025earlyFTinPractice,
  title={Early fault-tolerant quantum algorithms in practice: Application to ground-state energy estimation},
  author={Kiss, Oriel and Azad, Utkarsh and Requena, Borja and Roggero, Alessandro and Wakeham, David and Arrazola, Juan Miguel},
  journal={Quantum},
  volume={9},
  pages={1682},
  year={2025},
  publisher={Verein zur F{\"o}rderung des Open Access Publizierens in den Quantenwissenschaften},
doi={https://doi.org/10.22331/q-2025-04-01-1682},
note={\arxiv{2405.03754}}
}

@Article{	  lin2022HeisenbergLimited,
  author	= {Lin, Lin and Tong, Yu},
  doi		= {10.1103/PRXQuantum.3.010318},
  number		= {1},
  journal	= {\prxq},
  note		= {\arxiv{2102.11340}},
  numpages	= {21},
  pages		= {010318},
  publisher	= {American Physical Society},
  title		= {Heisenberg-Limited Ground-State Energy Estimation for
		  Early Fault-Tolerant Quantum Computers},
  volume	= {3},
  year		= {2022}
}

@Article{	  low2016HamSimQSignProc,
  author	= {Low, Guang Hao and Chuang, Isaac L.},
  doi		= {10.1103/PhysRevLett.118.010501},
  journal	= {\prl},
  note		= {\arxiv{1606.02685}},
  number	= {1},
  numpages	= {5},
  pages		= {010501},
  title		= {Optimal {H}amiltonian Simulation by Quantum Signal
		  Processing},
  volume	= {118},
  year		= {2017}
}

@Article{	  low2016HamSimQubitization,
  author	= {Low, Guang Hao and Chuang, Isaac L.},
  doi		= {10.22331/q-2019-07-12-163},
  journal	= {\quantum},
  note		= {\arxiv{1610.06546}},
  pages		= {163},
  title		= {Hamiltonian Simulation by Qubitization},
  volume	= {3},
  year		= {2019}
}

@Unpublished{	  low2019Multiproduct,
  author	= {Low, Guang Hao and Kliuchnikov, Vadym and Wiebe, Nathan},
  journal	= {arXiv preprint arXiv:1907.11679},
  note		= {\arxiv{1907.11679}},
  title		= {Well-conditioned multiproduct Hamiltonian simulation},
  year		= {2019}
}

@Article{	  martyn2021GrandUnificationQAlgs,
  author	= {Martyn, John M. and Rossi, Zane M. and Tan, Andrew K. and
		  Chuang, Isaac L.},
  doi		= {10.1103/PRXQuantum.2.040203},
  journal	= {\prx},
  note		= {\arXiv{2105.02859}},
  number	= {4},
  numpages	= {40},
  pages		= {040203},
  title		= {Grand Unification of Quantum Algorithms},
  volume	= {2},
  year		= {2021}
}

@unpublished{maxwell2026practical,
  title={Practical Estimation of Trotter Error for Hamiltonian Simulation},
  author={Maxwell, William and Casares, Pablo AM and Lang, Robert A and Fomichev, Stepan and Arrazola, Juan Miguel and Jahangiri, Soran and Asadi, Ali and Meneses, Luis Alfredo Nunez and Germain, Thomas and Motlagh, Danial},
  note={\arxiv{2606.30738}},
  year={2026}
}

@Article{	  mcardle2022ExploitingFermionNumber,
  author	= {McArdle, Sam and Campbell, Earl and Su, Yuan},
  doi		= {10.1103/PhysRevA.105.012403},
  issue		= {1},
  journal	= {\pra},
  month		= {1},
  note		= {\arxiv{2107.07238}},
  numpages	= {21},
  pages		= {012403},
  publisher	= {American Physical Society},
  title		= {Exploiting fermion number in factorized decompositions of
		  the electronic structure Hamiltonian},
  url		= {https://link.aps.org/doi/10.1103/PhysRevA.105.012403},
  volume	= {105},
  year		= {2022}
}

@article{mizuta2025trotterizationLowEnergy,
  title={Trotterization is substantially efficient for low-energy states},
  author={Mizuta, Kaoru and Kuwahara, Tomotaka},
  note={\arxiv{2504.20746}},
  year={2025}
}

@article{mizuta2026commutator,
  title={On the commutator scaling in Hamiltonian simulation with multi-product formulas},
  author={Mizuta, Kaoru},
  journal={Quantum},
  volume={10},
  pages={1974},
  year={2026},
  publisher={Verein zur F{\"o}rderung des Open Access Publizierens in den Quantenwissenschaften},
note={\arxiv{2507.06557}},
doi={10.22331/q-2026-01-19-1974}
}

@Article{	  montanaro2015QMonteCarlo,
  author	= {Ashley Montanaro},
  doi		= {10.1098/rspa.2015.0301},
  journal	= {\rspa},
  note		= {\arxiv{1504.06987}},
  number	= 2181,
  title		= {Quantum speedup of {M}onte {C}arlo methods},
  volume	= 471,
  year		= 2015
}

@article{patel2026quantumPhaselift,
  title={Quantum Phaselift},
  author={Patel, Dhrumil and Clinton, Laura and Flammia, Steven T and Garc{\'\i}a-Patr{\'o}n, Ra{\'u}l},
  note={\arxiv{2602.09119}},
  year={2026}
}

@Article{	  Sahinoglu2021Hamiltonian,
  author	= {Burak {\c{S}}ahino{\u{g}}lu and Rolando D. Somma},
  doi		= {10.1038/s41534-021-00451-w},
  journal	= {\npjqi},
  month		= {7},
  note		= {\arxiv{2006.02660}},
  number	= {1},
  publisher	= {Springer Science and Business Media {LLC}},
  title		= {Hamiltonian simulation in the low-energy subspace},
  url		= {https://doi.org/10.1038\%2Fs41534-021-00451-w},
  volume	= {7},
  year		= 2021
}

@article{somma2019QEE,
  title={Quantum eigenvalue estimation via time series analysis},
  author={Somma, Rolando D},
  journal={New Journal of Physics},
  volume={21},
  number={12},
  pages={123025},
  year={2019},
  publisher={IOP Publishing},
note={\arxiv{1907.11748}}
}

@Article{	  su2021NearlyTightTrottInerElect,
  archiveprefix	= {arXiv},
  arxivid	= {2012.09194},
  author	= {Su, Yuan and Huang, Hsin Yuan and Campbell, Earl T.},
  doi		= {10.22331/Q-2021-07-05-495},
  eprint	= {2012.09194},
  issn		= {2521327X},
  journal	= {\quantum},
  note		= {\arxiv{2012.09194}},
  number	= {1},
  pages		= {1--58},
  title		= {{Nearly tight Trotterization of interacting electrons}},
  volume	= {5},
  year		= {2021}
}

@article{suzuki1991general,
  author  = {Masuo Suzuki},
  title   = {General theory of higher-order decomposition of exponential operators and symplectic integrators},
  journal = {Physics Letters A},
  volume  = {165},
  number  = {5-6},
  pages   = {387--395},
  year    = {1992},
  doi     = {10.1016/0375-9601(92)90962-N}
}

@Article{	  wan2021RandPhaseEst,
  author	= {Wan, Kianna and Berta, Mario and Campbell, Earl T.},
  doi		= {10.1103/PhysRevLett.129.030503},
  number		= {3},
  journal	= {\prl},
  note		= {\arxiv{2110.12071}},
  numpages	= {7},
  pages		= {030503},
  publisher	= {American Physical Society},
  title		= {Randomized Quantum Algorithm for Statistical Phase
		  Estimation},
  url		= {https://link.aps.org/doi/10.1103/PhysRevLett.129.030503},
  volume	= {129},
  year		= {2022}
}

@article{wang2023QAlgGSEE,
  title={Quantum algorithm for ground state energy estimation using circuit depth with exponentially improved dependence on precision},
  author={Wang, Guoming and Fran{\c{c}}a, Daniel Stilck and Zhang, Ruizhe and Zhu, Shuchen and Johnson, Peter D},
  journal={Quantum},
  volume={7},
  pages={1167},
  year={2023},
  publisher={Verein zur F{\"o}rderung des Open Access Publizierens in den Quantenwissenschaften},
note={\arxiv{2209.06811}},
doi={https://doi.org/10.22331/q-2023-11-06-1167}
}

@article{wang2023fasterGSEE,
  title={Efficient ground-state-energy estimation and certification on early fault-tolerant quantum computers},
  author={Wang, Guoming and Fran{\c{c}}a, Daniel Stilck and Rendon, Gumaro and Johnson, Peter D},
  note={\arxiv{2304.09827}},
  journal={Physical Review A},
  volume={111},
  number={1},
  pages={012426},
  year={2025},
  publisher={APS},
doi={https://doi.org/10.1103/PhysRevA.111.012426}
}

@article{	  wang2023qubitefflinalg,
  author	= {Wang, Samson and McArdle, Sam and Berta, Mario},
  note		= {\arxiv{2302.01873}},
  title		= {Qubit-efficient randomized quantum algorithms for linear
		  algebra},
  journal={PRX Quantum},
  volume={5},
  number={2},
  pages={020324},
  year={2024},
  publisher={APS},
  doi={https://doi.org/10.1103/PRXQuantum.5.020324}
}

@Unpublished{	  watkins2022TimeDependentClockSimulation,
  author	= {Watkins, Jacob and Wiebe, Nathan and Roggero, Alessandro
		  and Lee, Dean},
  journal	= {arXiv preprint arXiv:2203.11353},
  note		= {\arxiv{2203.11353}},
  title		= {Time-dependent Hamiltonian simulation using discrete clock
		  constructions},
  year		= {2022}
}

@Unpublished{	  watson2023QuantumAlgorithmNuclearEFT,
  author	= {Watson, James D and Bringewatt, Jacob and Shaw, Alexander
		  F and Childs, Andrew M and Gorshkov, Alexey V and Davoudi,
		  Zohreh},
  note		= {\arxiv{2312.05344}},
  title		= {Quantum Algorithms for Simulating Nuclear Effective Field
		  Theories},
  year		= {2023}
}

@article{watson2024exponentiallyReducedTrotter,
  title = {Exponentially Reduced Circuit Depths Using Trotter Error Mitigation},
  author = {Watson, James D. and Watkins, Jacob},
  journal = {PRX Quantum},
  volume = {6},
  issue = {3},
  pages = {030325},
  numpages = {31},
  year = {2025},
  month = {Aug},
  publisher = {American Physical Society},
  doi = {10.1103/kw39-yxq5},
  note={\arxiv{2408.14385}}
}

@unpublished{watson2024exponentiallyReducedTrotterb,
  title={Exponentially Reduced Circuit Depths Using Trotter Error Mitigation},
  author={Watson, James D and Watkins, Jacob},
note={\arxiv{2408.14385}},
  year={2024},
  addendum={(in arXiv update to appear)}
}

@article{yoshioka2025krylovDiagonalization,
  title={Krylov diagonalization of large many-body Hamiltonians on a quantum processor},
  author={Yoshioka, Nobuyuki and Amico, Mirko and Kirby, William and Jurcevic, Petar and Dutt, Arkopal and Fuller, Bryce and Garion, Shelly and Haas, Holger and Hamamura, Ikko and Ivrii, Alexander and others},
  journal={Nature Communications},
  volume={16},
  number={1},
  pages={5014},
  year={2025},
  publisher={Nature Publishing Group UK London},
  doi={10.1038/s41467-025-59716-z},
  note={\arxiv{2407.14431}}
}

@Article{	  Zhao2021HamiltonianSW,
  author	= {Zhao, Qi and Zhou, You and Shaw, Alexander F. and Li,
		  Tongyang and Childs, Andrew M.},
  doi		= {10.1103/PhysRevLett.129.270502},
  issue		= {27},
  journal	= {\prl},
  month		= {12},
  note		= {\arxiv{2111.04773}},
  numpages	= {7},
  pages		= {270502},
  publisher	= {American Physical Society},
  title		= {Hamiltonian Simulation with Random Inputs},
  url		= {https://link.aps.org/doi/10.1103/PhysRevLett.129.270502},
  volume	= {129},
  year		= {2022}
}

\appendix

\section{Proof of \Cref{lem:time-evol-error-series}}
\label{sec:proof_approx_series}

\begin{proof}[Proof of \cref{lem:time-evol-error-series}]
Starting from Eq.~\eqref{eq:var-of-params-1},
we can iteratively substitute the same formula into itself.
Doing so $K \in \mathbb{Z}^+$ times, we can write the following:
\begin{align}
    &e^{-i{T}\Heff(s)} = e^{-iH{T}} + \\
    &\ + \sum_{l=1}^{K-1} \int^T_0 d\tau_1 \int^{\tau_1}_0 d\tau_2 \dots \int^{\tau_{l-1}}_0 d\tau_l e^{-i(T-\tau_1)H } i\Delta(sT) e^{-i(\tau_1-\tau_2)H } i\Delta(sT) \dots  e^{-i(\tau_{l-1}-\tau_l)H } i\Delta(sT)e^{-iH\tau_l}   \label{Eq:Iterate_K} \\
    &\ + \int^T_0 d\tau_1\int^{\tau_1}_0 d\tau_2 \dots \int_0^{\tau_{K-1}} d\tau_K e^{-i(T-\tau_1)H } i\Delta(sT) e^{-i(\tau_1-\tau_2)H } i\Delta(sT) \dots  e^{-i(\tau_{K-1}-\tau_K)H } i\Delta(sT)e^{-i\tau_K \Heff(s)}. \label{Eq:Remainder_Term}
\end{align}
 Taking line~(\ref{Eq:Iterate_K}) and expanding the definition of $\Delta(sT)$, 
\begin{align}
        \sum_{l=1}^{K-1} \int^T_0 &d\tau_1\int^{\tau_1}_0 d\tau_2\dots \int^{\tau_{l-1}}_0 d\tau_l e^{-i(T-\tau_1)H } i\Delta(sT) e^{-i(\tau_1-\tau_2)H} i\Delta(sT) \dots  e^{-i\tau_{l-1}-\tau_l)H } i\Delta(sT)e^{-iH\tau_l} 
        \nonumber\\ 
        =& \sum_{l=1}^{K-1} \int^T_0 d\tau_1 \int^{\tau_1}_0 d\tau_2  \dots \int^{\tau_{l-1}}_0 d\tau_l \left( \prod_{\kappa=l}^1\left( \sum_{j_\kappa\in \sym\mathbb{Z}_+\geq p } e^{-i(\tau_{\kappa-1}-\tau_\kappa)H}    i{E_{j_\kappa+1}} t^{j_\kappa}  \right)   \right)
        e
        ^{iH\tau_l}
        \\ 
        =&\sum_{l=1}^{K-1} \int^T_0 d\tau_1 \int^{\tau_1}_0 d\tau_2 \dots \int^{\tau_{l-1}}_0 d\tau_l \left( \sum_{\substack{j\in \sym\mathbb{Z}_+\geq pl}} t^j\sum_{\substack{j_1\dots j_l\in \sym\mathbb{Z}_+\geq p\\ j_1+\dots+ j_l=j}}\left(\prod_{\kappa=l}^1 e^{-i(\tau_{\kappa-1}-\tau_\kappa)H}    i{E_{j_\kappa+1}}\right) \right)e^{-iH\tau_l} \,, 
    \end{align}
    {where we have denoted $\tau_0 = T$.}
    We now reinsert $t = sT$, and make a change of variables $s_i = \tau_i /T$. This gives
    \begin{align}
        \sum_{l=1}^{K-1} T^l \int^1_0 ds_1 \int^{s_1}_0 ds_2 \dots \int^{s_{l-1}}_0 ds_l  \left( \sum_{j\in \sym\mathbb{Z}_+ \geq pl} (sT)^j\sum_{\substack{j_1\ldots j_l\in \sym\mathbb{Z}_+\geq p \\ j_1+\dots+j_l=j}} \left(\prod_{\kappa=l}^1  e^{-i(s_{\kappa-1}-s_\kappa)T
        H}    i{E_{j_\kappa+1}}\right)\right)e^{-iH s_lT} . 
    \end{align}
    Next, we regroup the sum according to the degree of $s$, which yields
    \begin{align}
       \sum_{j \in \sym\mathbb{Z}_+\geq p} &(sT)^j \sum_{l=1}^{\min\{K-1, \lfloor j/p \rfloor \}} T^l \int^1_0 ds_1 \int^{s_1}_0 ds_2 \dots \int^{s_{l-1}}_0 ds_l \sum_{\substack{j_1\dots j_l\in \sym\mathbb{Z}_+\geq p \\ j_1+\dots+j_l=j}}  \left(\prod_{\kappa=l}^1 e^{-iT(s_{\kappa-1}-s_\kappa)H}    i{E_{j_\kappa+1}}\right) e^{-iH s_lT} 
       \nonumber \\
       = &\sum_{j \in \sym\mathbb{Z}_+\geq p}  s^j \tilde{E}_{j+1,K}(T).
    \end{align}
     Here, we have defined
    \begin{multline} \label{eq:tilde_E_explicit}
        \tilde{E}_{j+1,K}(T) \coloneqq\\
        \sum_{l=1}^{\min\{K-1, \lfloor j/p \rfloor\}} T^{j+l} \int^1_0 ds_1 \int^{s_1}_0 ds_2 \dots \int^{s_{l-1}}_0 ds_l \sum_{\substack{j_1\dots j_l\in \sym\mathbb{Z}_+\geq p \\ j_1+\dots+j_l=j}}  \left(\prod_{\kappa=l}^1 e^{-i(s_{\kappa-1}-s_\kappa)TH} i{E_{j_\kappa+1}}\right) e^{-is_l T H}.
    \end{multline}
    We now want to put bounds on the norm of $\tilde{E}_{j+1,K}$. Using the triangle inequality, unitarity of $e^{-i \tau H}$, and evaluating the remaining integral,
     \begin{align}
     \begin{aligned}
        \norm{\tilde{E}_{j+1,K}(T)} &\leq \sum_{l=1}^{\min\{K-1, \lfloor j/p \rfloor  \}} T^{j+l} \int^1_0\int^{s_1}_0 \dots \int^{s_{l-1}}_0 ds_l \sum_{\substack{j_1\dots j_l\in \sym\mathbb{Z}_+\geq p \\ j_1+\dots+j_l=j}}  \left( \prod_{\kappa=1}^l \norm{{E_{j_\kappa+1}}} \right)   \\
        &\leq  T^j \sum_{l=1}^{\min\{K-1, \lfloor j/p \rfloor  \}}  \frac{T^l}{l!}\sum_{\substack{j_1\dots j_l\in \sym\mathbb{Z}_+\geq p \\ j_1+\dots+j_l=j}}  \left(\prod_{\kappa=1}^l     \norm{{E_{j_\kappa+1}} }  \right).   
    \end{aligned}
    \end{align}
    Applying~\cref{lem:Effective_Hamiltonian_Error_Series},
    \begin{align}
    \begin{aligned}
        \norm{\tilde{E}_{j+1,K}(T)} &\leq   T^j \sum_{l=1}^{\min\{K-1, \lfloor j/p \rfloor  \}}  \frac{T^l}{l!}\sum_{\substack{j_1\dots j_l\in \sym\mathbb{Z}_+\geq p \\ j_1+\dots+j_l=j}}  \left(\prod_{\kappa=1}^l  \acomm^{(j_\kappa + 1)} \frac{(a_\mathrm{max} \Upsilon)^{j_\kappa + 1}}{(j_\kappa + 1)^2}\right) \\
        &= (a_\mathrm{max}\Upsilon T)^j \sum_{l=1}^{\min\{K-1, \lfloor j/p \rfloor  \}}  \frac{(a_\mathrm{max}\Upsilon T)^l}{l!}\sum_{\substack{j_1\dots j_l\in \sym\mathbb{Z}_+\geq p \\ j_1+\dots+j_l=j}}  \left(\prod_{\kappa=1}^l  \frac{\acomm^{(j_\kappa + 1)}}{(j_\kappa + 1)^2}\right).
    \end{aligned}
    \end{align} 
    So far we have considered the terms in line~(\ref{Eq:Iterate_K}).
    We now consider line~(\ref{Eq:Remainder_Term}), which we will denote as $\tilde{F}_K(T,s)$.
    Applying the triangle inequality and utilizing unitarity in a similar manner as above, we can check that the operator norm of $\tilde{F}_{K}(T,s)$ is bounded by
    \begin{align}
    \begin{aligned}
        \norm{\tilde{F}_K(T,s)} 
        &\leq \frac{T^K}{K!} \norm{\Delta(sT)}^K \\
        &\leq \frac{T^K}{K!}\left( \sum_{j\in \sym\mathbb{Z}_+\geq p} \norm{E_{j+1}} (sT)^j\right)^K \,,
    \end{aligned}
    \end{align}
    which vanishes for $K \rightarrow \infty$, and we take this limit to complete the proof and simplify our notation as $\tilde{E}_{j+1,\infty}(T) \rightarrow \tilde{E}_{j+1}(T)$.
\end{proof}

\end{document}